\newcommand{\onlyFull}[1]{#1}
\newcommand{\onlyConference}[1]{}

\documentclass[numberwithinsect,a4paper,UKenglish,cleveref, autoref, thm-restate]{lipics-v2021}

\usepackage{hyperref} 
\usepackage{amsmath, amssymb, amscd, amsthm, amsfonts}
\newtheorem*{lemma*}{Lemma} %

\usepackage{algorithm,algorithmicx,algpseudocode}
\usepackage{tikzsymbols}
\usepackage{pifont}
\usepackage[percent]{overpic}
\usepackage{lineno}

\usepackage[noadjust,compress]{cite}
\usepackage{thm-restate}

\usepackage{mathtools}

\usepackage{qtree}
\usepackage[linguistics]{forest}

\usepackage{xcolor}

\usepackage{etoolbox}
\usepackage{color,colortbl}
 \usepackage{array}

\usepackage{bbm}
\usepackage{bm}
\usepackage{comment}
\usepackage{enumerate}

\usepackage{xspace}
\usepackage{rotating}
\usepackage{xifthen}
\usepackage{url}
\usepackage{csquotes}
\usepackage{graphicx}
\usepackage{wrapfig}
\usepackage{multirow}
\usepackage[binary-units=true]{siunitx}
\usepackage{stmaryrd}

\usepackage{float}
\usepackage[utf8]{inputenc} 
\usepackage{dsfont}

\usepackage[normalem]{ulem}
\usepackage{pgf}
\usepackage{tikz}
\usetikzlibrary{trees,decorations,arrows,automata,shadows,positioning,plotmarks,backgrounds,shapes}
\usetikzlibrary{calc,matrix,fit,petri,decorations.markings,decorations.pathmorphing,patterns,intersections,decorations.text}
\usetikzlibrary{arrows.meta, positioning}
\usepackage{pgfplots}
\usepackage{pgfplotstable}

\tikzstyle{mystate}=[state,inner sep=3pt,minimum size=20pt,line width=0.2mm]
\tikzstyle{fstate}=[state,accepting,inner sep=2pt,minimum size=3pt]
\tikzstyle{istate}=[state,initial,inner sep=2pt,minimum size=3pt]
\tikzstyle{mysquare}=[inner sep=3pt,minimum size=15pt,line width=0.2mm]
\tikzstyle{fmysquare}=[inner sep=3pt,minimum size=15pt,line width=0.5mm,accepting]

\usepackage{longtable}
\usepackage{caption}
\usepackage[position=b]{subcaption}

\usepackage{tabularx}
\usepackage{booktabs}
\usepackage{complexity}

\usepackage{etoc}
\etocsettocdepth{3}

\usepackage{xfrac}
\usepackage{faktor}
\usepackage{textcomp}

\usepackage{newfile}
\usepackage{totcount}

\regtotcounter{@todonotes@numberoftodonotes}

\providecommand{\EXPTIME}{\mathsf{EXP}}
\providecommand{\twoEXPTIME}{\mathsf{2EXP}}

\providecommand{\coNP}{\mathsf{coNP}}
\providecommand{\PTIME}{\mathsf{P}}

\crefname{theorem}{Thm.}{Thms.}
\Crefname{theorem}{Theorem}{Theorems}

\crefname{figure}{Fig.}{Figs.}
\Crefname{figure}{Figure}{Figures}

\crefname{observation}{Obs.}{Obs.}
\Crefname{observation}{Observation}{Observations}

\newclass{\logspace}{logspace}

\newcommand{\myparagraph}[1]{\vspace{-0.3cm}\subparagraph*{#1}}

\newcommand{\Obj}{\textsf{Obj}}
\newcommand{\RIO}{\textsf{RIO}}
\newcommand{\Minsky}{\mathcal{S}}

\newcommand{\calP}{\mathcal{P}}
\newcommand{\bzero}{\bm{0}}
\newcommand{\bx}{\bm{x}}
\newcommand{\bu}{\bm{u}}

\renewcommand{\theenumi}{\alph{enumi}}

\newcommand{\N}{\mathbb{N}}
\newcommand{\Mon}{\mathbb{M}}

\newcommand{\Z}{\mathbb{Z}}

\newcommand{\Vdim}{\mathsf{Vdim}}

\mathchardef\mhyphen="2D
\newcommand{\Gr}{\mathcal{G}}
\newcommand{\gr}{\mathtt{G}}
\newcommand{\VASS}{\mathsf{VASS}}
\newcommand{\ZVASS}{\Z\mhyphen\mathsf{VASS}}
\newcommand{\Group}{\mathsf{Grp}}
\newcommand{\Pushdown}{\mathsf{PD}}
\newcommand{\FICEG}{\mathsf{FV}}
\newcommand{\UICEG}{\mathsf{UV}}

\newcommand{\RInv}[1]{\mathcal{R}(#1)}

\newcommand{\init}{\mathit{init}}
\newcommand{\trap}{\mathit{trap}}

\newcommand{\Conf}{\mathsf{Conf}}

\usepackage[colorinlistoftodos,prependcaption,textsize=tiny]{todonotes}

\newcommand{\Pthree}{\mathsf{P3}}

\newcommand{\pp}{\mathtt{PuPo}}
\newcommand{\expp}{\mathtt{guess}}

\makeatletter
\newcommand*{\encircled}[1]{\relax\ifmmode\mathpalette\@encircled@math{#1}\else\@encircled{#1}\fi}
\newcommand*{\@encircled@math}[2]{\@encircled{$\m@th#1#2$}}
\newcommand*{\@encircled}[1]{%
  \tikz[baseline,anchor=base]{\node[draw,circle,outer sep=0pt,inner sep=.2ex] {#1};}}
\makeatother

\newtheorem{maintheorem}{Main Theorem}[section]

\crefname{maintheorem}{Main Theorem}{Main Theorems}

\newcommand{\nrgs}{viability games}

\newcommand{\nrg}{viability game}
\newcommand{\Nrgs}{Viability games}

\newcommand{\pdtwoz}[1]{%
\begin{tikzpicture}[every circle/.style={}, scale=#1]
\fill (0,0) circle (2pt) node (a) {}    (1,0) circle (2pt) node (b)  {}   (2,0) circle (2pt) node (c) {};
\draw (b.center) -- (c.center);
\draw (b.center) ++(90:3pt) circle (3pt);
\draw (c.center) ++(90:3pt) circle (3pt);
\node[below=0cm of a] {$a$};
\node[below=0cm of b] {$b_0$};
\node[below=0cm of c] {$b_1$};
\end{tikzpicture}
}
\newcommand{\twoz}[1]{%
\begin{tikzpicture}[every circle/.style={}, scale=#1]
\fill (0,0) circle (2pt) node (a) {}    (1,0) circle (2pt) node (b)  {};
\draw (a.center) -- (b.center);
\draw (a.center) ++(90:3pt) circle (3pt);
\draw (b.center) ++(90:3pt) circle (3pt);
\end{tikzpicture}
}
\newcommand{\twob}[1]{%
\begin{tikzpicture}[every circle/.style={}, scale=#1]
\fill (0,0) circle (2pt) node (a) {}    (0.75,0) circle (2pt) node (b)  {};
\draw (a.center) -- (b.center);
\end{tikzpicture}
}

\newcommand{\twocounters}[1]{%
\begin{tikzpicture}[every circle/.style={}, scale=#1]
\fill (0,0) circle (2pt) node (a) {}    (0.75,0) circle (2pt) node (b)  {};
\draw (a.center) -- (b.center);
\draw[densely dotted] (a.center) ++(90:3pt) circle (3pt);
\draw[densely dotted] (b.center) ++(90:3pt) circle (3pt);
\end{tikzpicture}
}

\newcommand{\btimesz}[1]{%
\begin{tikzpicture}[every circle/.style={}, scale=#1]
\fill (0,0) circle (2pt) node (a) {}    (0.75,0) circle (2pt) node (b)  {};
\draw (a.center) -- (b.center);
\draw (b.center) ++(90:3pt) circle (3pt);
\end{tikzpicture}
}

\title{Infinite-state games with energy objectives beyond counters}%

\authorrunning{I. Sa\u{g}lam, G. Zetzsche} %

\author{Irmak Sa\u{g}lam}{Max Planck Institute for Software Systems (MPI-SWS), Germany}{isaglam@mpi-sws.org}{https://orcid.org/0000-0002-4757-1631}{}

\author{Georg Zetzsche}{Max Planck Institute for Software Systems (MPI-SWS), Germany}{georg@mpi-sws.org}{https://orcid.org/0000-0002-6421-4388}{}

\Copyright{Irmak Sa\u{g}lam and Georg Zetzsche}

\funding{\flag[3cm]{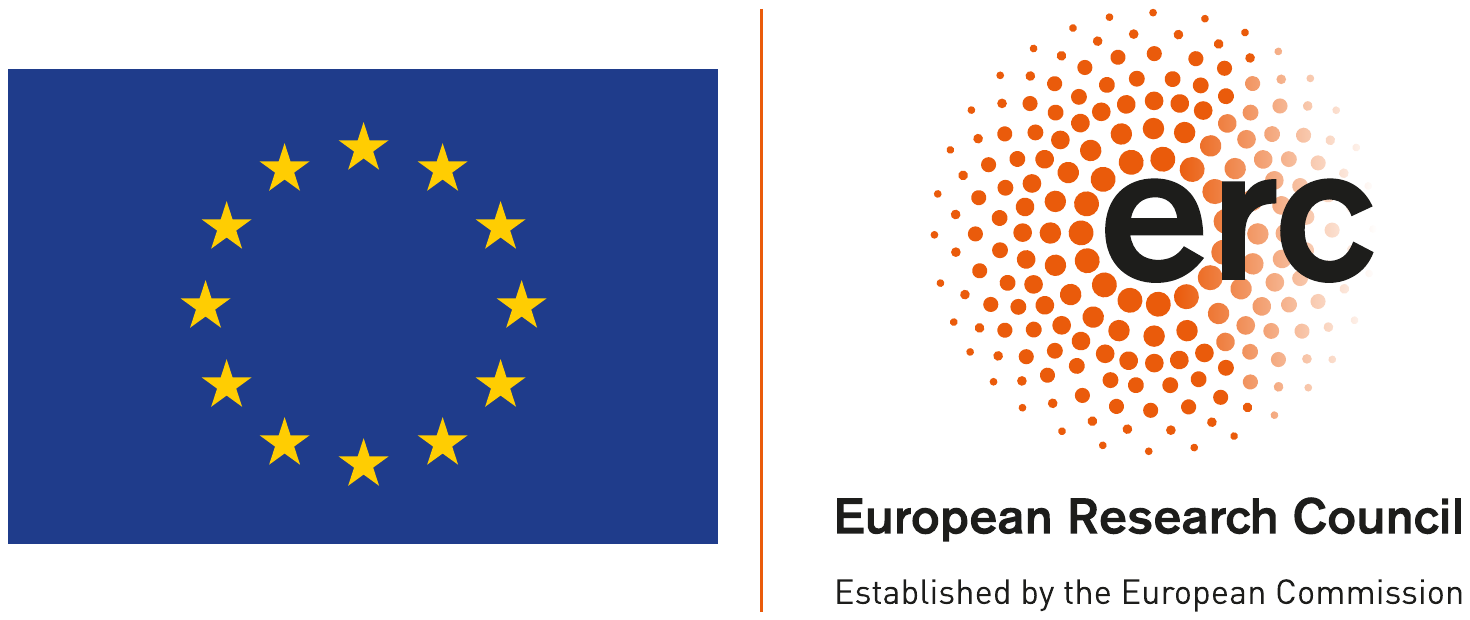}Funded by the European Union (ERC, FINABIS, 101077902). Views and opinions expressed are however those of the author(s) only and do not necessarily reflect those of the European Union or the European Research Council Executive Agency. Neither the European Union nor the granting authority can be held responsible for them.}

\ccsdesc[500]{Theory of computation~Problems, reductions and completeness}
\ccsdesc[300]{Theory of computation~Concurrency}
\ccsdesc[500]{Theory of computation~Formal languages and automata theory}

\keywords{Games on Graphs, Decidability, Complexity, Energy Games, Vector addition systems, Pushdown, Groups, Valence Systems} %
\onlyConference{\relatedversiondetails[cite=full]{Full Version}{http://arxiv.org/abs/2605.05935}}

\nolinenumbers %

\onlyFull{\hideLIPIcs}  %

\onlyConference{
\EventEditors{Sayan Bhattacharya, Danupon Nanongkai, Michael Benedikt, and Gabriele Puppis}
\EventNoEds{4}
\EventLongTitle{53rd International Colloquium on Automata, Languages, and Programming (ICALP 2026)}
\EventShortTitle{ICALP 2026}
\EventAcronym{ICALP}
\EventYear{2026}
\EventDate{July 7--10, 2026}
\EventLocation{Royal Holloway, University of London, Egham, United Kingdom}
\EventLogo{}
\SeriesVolume{374}
\ArticleNo{191}}

\begin{document}
\thispagestyle{empty} %

\maketitle

\begin{abstract}
	In the theory of games on infinite-state arenas, there is a stark contrast
between (i)~recursion-based models such as pushdown systems and extensions on
one hand, and (ii)~counter-based models like vector addition systems with
states (VASS) on the other. For pushdown systems and extensions, there is a
rich variety of decidable and well-understood games, whereas on VASS arenas,
even extremely simple games are undecidable. Here, a VASS is an automaton with
counters that can be incremented and decremented, but not tested for zero.
Crucially, the counters can only assume non-negative values.

However, certain VASS games become decidable when using energy semantics:  An
energy game is played on a system with counters, but the arena includes
configurations with negative counters. The requirement that the counters stay
non-negative is, instead, part of the winning condition of the existential
player.  

We study an analogue of energy semantics---legality of instructions as part
of the winning condition rather than arena---on a broad class of infinite-state
systems, where we call them viability games. Specifically, we study viability
games in the framework of valence systems over graph monoids, where
(undirected, loops allowed) graphs specify various infinite-state systems, such
as pushdowns, VASS counters, integer counters, and combinations thereof. 

In our main results, we provide a complete description of the decidability and
complexity landscape of viability games across valence systems over graph
monoids.  Our results reveal encouraging decidability properties.  For example,
in certain combinations of pushdowns and counters, viability games are
decidable, despite non-termination games being undecidable there.  Moreover,
viability games are even decidable for certain systems where (single-player)
control-state reachability is undecidable.

\end{abstract}

\newpage
\clearpage 
\pagenumbering{arabic} 
\section{Introduction}\label{introduction}
Games on finite-state arenas are a cornerstone of formal methods. This is due
to (i)~their applications to reactive synthesis, where the task is to
automatically construct a finite-state system to specification, but also (ii)~a
wealth of algorithmic methods for analyzing games. 

The situation is much more dire for infinite-state arenas. This is unfortunate, because software systems are typically represented using infinitely many states. Indeed, games on infinite-state arenas are usually only decidable for specialized classes of infinite-state systems: While for pushdown systems (and recursion-based extensions like higher-order pushdown automata or higher-order recursion schemes), a rich variety of games is decidable~\cite[Ch.~12]{GamesOnGraphs}, using counter-based systems like vector addition systems as arenas will almost always lead to undecidability~\cite[Ch.~13]{GamesOnGraphs}.

\myparagraph{Vector addition systems}  For example,
\emph{($d$-dimensional) vector addition systems ($d$-VASS)} are systems with
access to $d$ counters---ranging over the natural numbers---which can be incremented and decremented, but not
tested for zero. VASS are a pillar of verification and
infinite-state research, because of their ability to model concurrent systems.

Unfortunately, with VASS configuration graphs as arenas, almost all games are
undecidable.  For example, games where the objective is to (i)~reach a
particular control-state (or configuration), or even just to (ii)~sustain an
infinite run, are both undecidable, already in dimension
two~\cite[Thm.~13.1]{GamesOnGraphs}. In fact, it is a well-known phenomenon
that branching time logics are usually undecidable for VASS (e.g.\ see
\cite[Section 6]{esparza1994decidability} for a severely restricted branching
logic; and see~\cite[Section 5]{esparza2024decidability} for more results),
implying undecidability for many other kinds of games on VASS.

\myparagraph{Energy games} However, there is a ``VASS-like'' game that is
decidable: Energy games. Here, the configurations consist of a control-state
and counters that can be incremented and decremented.  There are two players,
the existential player ($\exists$), and the universal player ($\forall$). Here,
$\exists$ has the objective of just playing infinitely long, such that the
counter values stay non-negative. The difference to other games on VASS is,
however, that \emph{in the arena}, the configurations are over the
\emph{integers}. The fact that the counter values stay non-negative is
\emph{part of $\exists$'s winning condition}. The game can be viewed as
``VASS-like'', since $\exists$ has to construct a VASS run. Energy games have
received a significant amount of attention, with well-understood decidability
and
complexity~\cite{AbdullaAHMKT14,JurdzinskiLS15,GamesOnGraphs,DBLP:conf/icalp/BrazdilJK10,DBLP:journals/fuin/Chaloupka13,Secondary-ChatterjeeRR12,DBLP:journals/entcs/RaskinSB05,DBLP:conf/mfcs/CourtoisS14,DBLP:conf/lics/ColcombetJLS17,ChatterjeeDHR10}.
Thus, Energy semantics---making non-negativity part of the winning condition
rather than the arena---achieves decidability.

\myparagraph{In search of decidable infinite-state games} In the quest for
decidable games for infinite-state systems, the remarkable case of energy games
thus raises the question: \emph{Does energy semantics also lead to decidability
in broader classes of infinite-state systems?}

\myparagraph{\Nrgs} To this end, we study Energy-like semantics on
infinite-state systems other than VASS, and call them \emph{\nrgs}.  Here, we
take a class of infinite-state systems with instructions like \emph{push, pop,
increment, decrement} and allow these instructions to be applied anytime.
However, the winning condition of the existential player---whose goal it is to
sustain an infinite run---includes that this run is confined to
\emph{valid}/\emph{viable} configurations. Here, being valid/viable means that
the configuration existed in the original system: This is precisely how energy
games are related to games on VASS arenas. The term ``viability'' expresses
that the existential player has to maintain a more general viability condition
rather than specifically keeping energy levels non-negative.

\myparagraph{Valence systems over graph monoids} We formalize \nrgs\ in the
general framework of \emph{valence systems over graph
monoids}~\cite{Zetzsche2016c}. These systems consist of a finite-state
automaton with access to an infinite-state storage mechanism, which is decribed
by a finite undirected graph $\Gamma$. For example, if $\Gamma$ consists of
isolated vertices without self-loops, then valence systems over $\Gamma$ are
pushdown automata. If $\Gamma$ is a clique of $d$ vertices (without
self-loops), then this realizes a $d$-dimensional VASS.  Adding self-loops
would yield a $d$-dimensional integer VASS. The framework can also realize
combinations such as (i)~pushdown
VASS~\cite{DBLP:journals/corr/abs-2504-05015,DBLP:conf/icalp/LerouxST15}, which
feature both a pushdown and VASS counters or (ii)~pushdowns of
counters~\cite{DBLP:conf/icalp/Zetzsche13,DBLP:journals/iandc/Zetzsche21,DBLP:conf/stacs/Zetzsche15},
where each stack entry contains a valuation of a set of counters.

Valence systems over graph monoids have been studied with respect to various
algorithmic problems, such as reachability
problems~\cite{DBLP:journals/iandc/Zetzsche21,DBLP:conf/lics/GanardiMZ22,DBLP:conf/rp/Zetzsche21},
decidable
underapproximations~\cite{DBLP:conf/concur/MeyerMZ18,DBLP:conf/concur/ShettyKZ21},
finite-state abstractions~\cite{DBLP:conf/lics/AnandSSZ24,DBLP:conf/stacs/Zetzsche15,DBLP:conf/mfcs/BuckheisterZ13}, logics on
configuration graphs~\cite{DBLP:conf/lics/DOsualdoMZ16} and even games played
on configuration graphs, in Muskalla's dissertation~\cite[Ch.~19]{Muskalla23}.
Unfortunately, the results on games echo the abovementioned situation on
pushdowns and VASS: Essentially, games with a control-state or configuration
reachability objective are decidable for valence systems over $\Gamma$
precisely when $\Gamma$ realizes a pushdown~\cite[Thms.
19.2.1\&19.2.10]{Muskalla23}.

\myparagraph{Main results} Our main results provide a full
description of the decidability and complexity landscape of \nrgs\ across
valence systems over graph monoids. Specifically, we study the problem of
determining the winner of a \nrg\ in two variants, inherited from
energy games: In the \emph{fixed initial credit} setting, the initial
storage content is given; in the \emph{unknown initial credit} setting, the task
is to decide whether there exists a storage content (i.e.\ counter valuation,
stack content, etc.) from which the existential player wins.

Moreover, we consider these problems in the general setting where graph
$\Gamma$ is drawn from a class $\Gr$ of graphs (that is closed under taking induced
subgraphs).  This captures settings where aspects of the storage
mechanism are part of the input. For example, if $\Gr$ is the class of unlooped
cliques of size $\le d$, then \nrgs\ over $\Gr$ correspond to $d$-dimensional
energy games. If $\Gr$ contains all unlooped cliques, then \nrgs\ over $\Gr$
are energy games over arbitrary dimension (and the dimension is part of the input).

Our results show that for every choice of $\Gr$, the above problems are either
(i)~$\PTIME$-complete, (ii)~$\coNP$-complete, (iii)~$\EXPTIME$-complete,
(iv)~$\twoEXPTIME$-complete, or (v)~undecidable.

\myparagraph{Viability semantics achieves decidability} 
Viability games are indeed decidable for a relatively broad class of systems.
Among valence systems, viability games are decidable essentially for two kinds
of systems: VASS (i.e.~energy games) and what we denote as $\Pushdown(\Group)$.
The latter are an extension of pushdown systems where each stack entry carries
an element from some infinite group $G$.  With the instructions, one can push and
pop letters, but also multiply elements of $G$ onto the
top-most group element on the stack.  Here, pushing creates a fresh entry
initialized to $1\in G$. Popping is only possible if the top-most group element
equals the neutral $1\in G$. Thus, the invalid instructions (to be
avoided by the existential player) are pops with (i)~a mismatching stack
letter or (ii)~the top-most group element being $\ne 1$.

As an example, if $G$ is the group
$\Z^d$ for some $d\ge 1$, then we call these systems $\Pushdown(\ZVASS)$,
because one has a stack whose entries carry integer vectors that can be
incremented and decremented, as in $\Z$-VASS~\cite{DBLP:conf/rp/HaaseH14}. One
can thus regard this as \emph{stacks of $\Z$-VASS configurations}.

Decidability of viability games for $\Pushdown(\Group)$ is noteworthy for two
reasons. First, in the classic setting (i.e.\ arenas of valid configurations)
even the special case of $\Pushdown(\ZVASS)$ is as unfavorable to games
as VASS: even non-termination games are undecidable
\onlyFull{(see~\cref{app:infinite-runs-pdtwoz})}
\onlyConference{(see the full version~\cite{full})}%
, already for $d=2$.  Second, in
$\Pushdown(\Group)$ overall (i.e.~beyond the groups $\Z^d$), even ordinary
(single-player) control-state reachability is undecidable. Thus, even
though it is undecidable whether a valid pop will ever be executed, we can
still decide viability games over $\Pushdown(\Group)$.

\myparagraph{Related work} Energy games have received a significant amount of
attention in recent decades~\cite{Secondary-ChatterjeeRR12,AbdullaAHMKT14,JurdzinskiLS15,GamesOnGraphs,DBLP:conf/icalp/BrazdilJK10,DBLP:journals/fuin/Chaloupka13,DBLP:journals/entcs/RaskinSB05,DBLP:conf/mfcs/CourtoisS14,DBLP:conf/lics/ColcombetJLS17,ChatterjeeDHR10}.  They are equivalent to Z-reachability
games~\cite{DBLP:conf/icalp/BrazdilJK10,DBLP:journals/fuin/Chaloupka13}, B-VASS
games~\cite{DBLP:journals/entcs/RaskinSB05} and can also be viewed as
alternating VASS~\cite{DBLP:conf/mfcs/CourtoisS14}.  As observed by Abdulla, Mayr, Sangnier,
and Sprosten~\cite{DBLP:conf/concur/AbdullaMSS13}, instead of viewing energy
games as VASS games with a modified semantics, one can equivalently view them
as VASS games with the restriction of being \emph{one-sided}, meaning the
universal player cannot modify the counters. In the same way, viability games
correspond to one-sided games over valence system arenas.

Another combination of energy games with infinite-state systems beyond counters
has been studied in \cite{AbdullaAHMKT14}. Their approach is to add counters to
existing infinite-state systems (i.e.\ pushdown and one-counter systems) and
apply energy semantics only to the extra counters:  They show that (i)~pushdown
games with one energy dimension and (ii)~one-counter games with two extra
energy dimensions are undecidable, whereas one-counter games with one energy
dimension are decidable. In our work, in contrast, the energy/viability semantics
applies to the \emph{entire} system.  However, we observe that their
undecidable ``pushdown energy games'' (i.e.\ pushdown systems with extra
counters) reduce to pushdown VASS games under viability semantics \onlyFull{(see
\cref{app:undec-iii-iv}).}\onlyConference{(see the full version~\cite{full}).}

\myparagraph{Outline} In \cref{preview}, we provide a self-contained
description of the two types of infinite-state systems for which our results
yield decidability. After this, we start with preliminaries in
\cref{preliminaries} and define \nrgs\ in \cref{games}. In \cref{main-results},
we state our two main results (\cref{thm:FIC-EG,thm:UIC-EG}) and discuss
implications. Then \cref{proof-overview} gives an overview of the proof and of
novel ingredients.  In
\cref{sec:FICEG,sec:UICEG,sec:undecidability-results,sec:completeness} we then
sketch our proofs. \onlyFull{Remaining proof details are in the
appendix.}\onlyConference{Remaining proof details can be found in the full
version~\cite{full}.}

\section{The two decidable cases}\label{preview}
Within the framework of viability games over valence systems, our main results
identify two types of systems where these games are decidable. Since the entire
framework is somewhat involved, we are using this section to describe these
two system types independently of the framework.

\myparagraph{Model I: Vector addition systems with states} The first type of system is
that of VASS. A \emph{($d$-dimensional) vector addition system with states
(VASS)} is a pair $(Q,\delta,q_\init)$ consisting of a finite set $Q$ of
\emph{states} and a finite set $\delta\subseteq Q\times\Z^d\times Q$ of
\emph{transitions}, and a \emph{final state} $q_\init\in Q$. A
\emph{(valid) configuration}\footnote{The reason we call these \emph{valid} configurations is that in valence systems that correspond to VASS, we will have a more liberal notion of configurations (i.e.\ there are things called configurations that are not valid VASS configurations. However, to keep terminology simple in this section, we also just use the term configuration.} is a pair in $Q\times\N^d$.  A transition
$t=(p,\bu,q)\in\delta$ is \emph{valid} in a configuration $(p,\bx)$ if
$\bu+\bx\in\N^d$, in which case the resulting configuration is $(q,\bx+\bu)$,
and we write $(p,\bx)\xrightarrow{t}(q,\bx+\bu)$. A sequence $w=t_1\cdots t_m$
of transitions is \emph{valid} in a configuration $(p,\bx)$ if there are
configurations $(q_0,\bx_0),\ldots,(q_m,\bx_m)$ such that $(p,\bx)=(q_0,\bx_0)$
and $(q_i,\bx_i)\xrightarrow{t_i}(q_{i+1},\bx_{i+1})$. In this case, we write
$(p,\bx)\xrightarrow{w} (q_m,\bx_m)$.

\myparagraph{Finitely generated groups} The second type of system is an extension
of pushdown systems, where in between letters, one can store elements of a
group. In the framework of valence systems, this group is always a \emph{graph
group} (also called right-angled Artin group), but in fact, our decidability
proof applies to any finitely generated group with a decidable word problem.
(Moreover, the $\EXPTIME$ complexity bound holds as soon as the word problem of
the group is solvable in $\EXPTIME$). Since this more general setting is also
easier to describe, we choose it here. 

We begin with some terminology.  A \emph{group} is a set $G$ together with an
associative binary operation $\cdot\colon G\times G\to G$ such that (i)~there
is a neutral element $1\in G$, i.e.~$g\cdot 1=1\cdot g=g$, and (ii)~every
element $g\in G$ has an inverse $g^{-1}\in G$ such that $g\cdot
g^{-1}=g^{-1}\cdot g=1$.  For a subset $S\subseteq G$ and a sequence $w\in S^*$
over elements of $S$, we write $[w]\in G$ for the group element obtained by
multiplying the elements in $w$.  A \emph{generating set} of $G$ is a subset
$S\subseteq G$ such that every element of $G$ can be written as a product of
elements in $S$ or in $S^{-1}\colon=\{s^{-1}\mid s\in S\}$.  The group $G$ is
\emph{finitely generated} if it has a finite generating set.  If $S$ is a
finite generating set of $G$, then the \emph{word problem} is the following
decision problem: Given a string $w\in(S\cup S^{-1})^*$, decide whether
$[w]=1$. It is well-known that decidability (or complexity, up to logspace
reductions) of the word problem does not depend on the chosen generating
set~\cite[p.~88]{LyndonSchupp1977}. 

\myparagraph{Model II: Group pushdown systems} From now on, we fix a finitely generated
group $G$, with a finite generating set $W$. A \emph{$G$-pushdown system} is a
tuple $\calP=(Q,\Sigma,\delta,q_\init)$, where $Q$ is a finite set of
\emph{states}, $\Sigma$ is its \emph{stack alphabet}, \[ \delta\subseteq
Q\times (\Sigma\cup\bar{\Sigma}\cup G)\times Q \] is its finite set
of \emph{transitions}, and $q_\init\in Q$ is its \emph{initial state}.  Here,
edges labeled $\Sigma$ are \emph{push-transitions}, i.e.~an edge $(p,a,q)$ with
$a\in\Sigma$ pushes an $a$ on the stack. Edges labeled
$\bar{\Sigma}=\{\bar{a}\mid a\in\Sigma\}$ are \emph{pop-transitions}, i.e.\ an
edge $(p,\bar{a},q)$ will pop an $a$ from the stack. 

Formally, the elements of $G$ on edges of $\calP$ are specified by words over
$W\cup W^{-1}$, which is possible since $W$ is a generating set of $G$.

A \emph{(valid) configuration} of $\calP$ is a pair $(p,u)$, where $p\in Q$ and $u$ is
a sequence $g_0 a_1 g_1 \cdots a_m g_m$, where $g_0,\ldots,g_m\in G$ are group
elements, and $a_1,\ldots,a_m\in\Sigma$ are stack letters. In such a
configuration, a transition $t=(p,x,q)$ is \emph{valid} if
\begin{enumerate}
\item $x\in\Sigma$. In this case, $t$ is called a \emph{push transition} and
the resulting configuration is $(q,g_0a_1g_1\cdots a_mg_m x1)$.
\item $x\in\bar{\Sigma}$, say $x=\bar{a}$, and $a_m=a$ and $g_m=1$. In this
case, $t$ is called a \emph{pop transition}, and the resulting configuration is
$(q,g_0a_1g_1\cdots a_{m-1}g_{m-1})$.
\item $x\in G$. In this case, $t$ is called a \emph{group
transition} and the resulting configuration is $(q,g_0a_1g_1\cdots a_m
(g_mx))$. In other words, $x$ is multiplied onto the top-most group element.
\end{enumerate}
In these cases, we write $(p,u)\xrightarrow{t}(q,v)$, where $(q,v)$ is the
resulting configuration. As above, we define $(p,u)\xrightarrow{w}(q,v)$ to
mean that $(q,v)$ can be reached via a sequence $w$ of valid transitions. In this case, the sequence $w$ is called \emph{valid} in $(p,u)$.

\myparagraph{Viability games} Let us now describe viability games on VASS and
group pushdown systems. Such a game is played by two players, the existential
($\exists$) and the universal player ($\forall$), and is specified by a VASS
(resp.\ a group pushdown system), so that the set $Q$ of states is partitioned
as $Q=Q_\exists \cup Q_\forall$, where $Q_\exists$ belongs to the existential
player, and $Q_\forall$ belongs to the universal player. During a play, the two
players choose transitions. Here, the chosen transition must always respect the
current state. Moreover, the current state determines which player can choose
the next transition.

The key idea is that the \emph{winning condition} of the existential player
requires (i)~the play to be infinite and (ii)~all prefixes of the transition
sequence to be valid in the initial configuration. In the \emph{fixed initial
credit} setting, the initial configuration is the empty configuration:
$(q_\init,\bzero)$ for VASS and $(q_\init,\varepsilon)$ for group pushdown
systems (here, $\varepsilon$ is the empty sequence). In the \emph{unknown
initial credit} setting, we are asked whether there exists an initial
configuration from which the existential player wins.

\myparagraph{Viability games on VASS and energy games} Crucially, observe that
viability games on ($d$-dimensional) VASS are equivalent to ($d$-dimensional)
energy games.

\myparagraph{Example: Viability game on a $\Z \times \Z$-pushdown system}
Consider a $G$-pushdown system $\calP = (Q, \Sigma, \delta, q_\init)$, where
$Q = \{q_0, \ldots, q_5\}$, $\Sigma = \{a\}$, $q_{\init} = q_0$ and $G = \Z \times \Z$ is the
group with componentwise addition and neutral element $(0,0)$, generated by
$W = \{(1,0),(-1,0),(0,1),(0,-1)\}$. The game is depicted in
\cref{fig:example-pushdown}, where circle states belong to the existential
player ($Q_\exists$) and square states to the universal player ($Q_\forall$).

In the fixed initial credit setting, the existential player does not win this game, 
as she cannot win from
the initial configuration $(q_\init, \varepsilon)$. In the unknown initial
credit setting, however, she does win. Specifically, she wins from any initial
configuration of the form $(q_\init, u \cdot a (0,-1) )$, where $u = g_0 a g_1 \cdots a
g_m$ for $g_0, \ldots, g_m \in \Z \times \Z$, i.e.\ any configuration with at least one $a$ and the topmost
group element $(0,-1)$. To see why the top of the stack must be $a (0,-1)$:
in order to sustain an infinite play, the existential player must eventually
reach $q_2$, which requires applying the group transition $(0,1)$ followed by
the pop transition $\bar{a}$. For the pop transition to be valid, the topmost
stack entry must be $(0,-1)$, followed by the letter $a$, so that after the group
transition $(0,1)$ the topmost group element becomes $(0,0) = 1$ (the neutral element of $G$), as required
for a valid pop, so the $a$ can be popped.

Once the existential player reaches $q_2$ with stack content $u$, she wins by
the following strategy. From $q_2$, the universal player may take the self-loop
any number of times $m \in \N$ (each time pushing an $a$ with group element
$(0,0)$, leaving the stack as $u a^m$) before moving to $q_3$. From $q_3$, the
existential player moves to either $q_4$ or $q_5$, depending on which transition the universal player took to $q_3$, with the stack content $u a^m$. In all cases, the existential
player can navigate back to $q_2$ with stack content $u$, by popping the $m$
copies of $a$ using the self-loop on $q_5$. Since the same strategy applies
from $q_2$ with any stack content, this produces an infinite sequence of valid
transitions, and the existential player wins.

It follows from our results (see \cref{main-results}) that viability games are
decidable (and, even in $\EXPTIME$) for $\Z \times \Z$-pushdown systems---in both the fixed
and unknown initial credit settings (in fact, the decidability proof even applies to any group with a
decidable  word problem). However, perhaps surprisingly, if we replace
one of the $\Z$-counters with a VASS counter (i.e.\ one that must be non-negative in valid configurations), viability games become
undecidable.

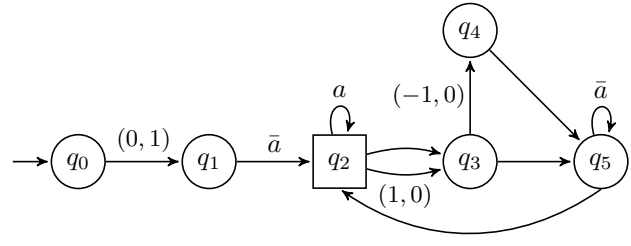
\begin{figure}[t]
  \centering
  \tikzset{every state/.style={minimum size=50pt}}
  \tikzset{every loop/.style={looseness=2}}
  \begin{tikzpicture}[
    ->, >=stealth', shorten >=1pt,
    auto,
    node distance=1.5cm,
    semithick
  ]  
  \node[circle, draw] (s) {$q_0$}; %
  \node[circle, draw] (p1) [right= 1cm of s] {$q_1$};
   \node[rectangle, draw, minimum width=7mm, minimum height=7mm] (p2) [right= 1cm of p1] {$q_2$};
  \node[circle, draw] (p3) [right= 1cm of p2] {$q_3$};
  \node[circle, draw] (p4) [above= 1cm of p3] {$q_4$}; 
  \node[circle, draw] (p5) [right= 1cm of p3] {$q_5$}; 
  \draw (s.west) ++(-0.5cm,0) -- (s);

  \draw (s)  -- node[above] {\small $(0,1)$} (p1);
  \draw (p1) -- node[above] {$\bar{a}$} (p2);
  \draw (p3) -- node[left, xshift=2pt] {\small $(-1,0)$} (p4);
  \draw (p3) -- node[below] {} (p5);
  \draw (p4) -- node[above, xshift=2pt] {} (p5);

  \draw[->] (p2) edge [loop above] node {$a$} ();
  \draw[->] (p5) edge [loop above] node {$\bar{a}$} ();

 \draw[->] (p2) to[bend right=15] node[below] {\small $(1,0)$} (p3);
  \draw[->] (p2) to[bend left=15] node[above] {} (p3);

    \draw[->] (p5.south) to[bend left=35] node[below] {} (p2.south);

  \end{tikzpicture}
\caption{A viability game on the $G$-pushdown system $\calP$, where $G = \Z
\times \Z$. Circle states belong to the existential player and square states to
the universal player. Push transitions are labeled with stack letters
($a$), pop transitions with their barred counterparts ($\bar{a}$), and group
transitions with elements of $G$. Unlabeled edges carry the neutral element
$(0,0) \in \Z \times \Z$, i.e.\ they leave the top group element unchanged.}  \label{fig:example-pushdown}
\end{figure}

\section{Preliminaries}\label{preliminaries}
\myparagraph{Graphs} In the general framework of valence systems over graph
monoids, the storage mechanism of an infinite-state system is specified by an
undirected graph that may have self-loops.  Formally, a \emph{graph} is a pair
$\Gamma=(V,E)$, where $V$ is its set of \emph{vertices} and
$E\subseteq \{S\subseteq V \mid |S|\le 2\}$ is its set of \emph{edges}. Nodes
$a,b\in V$ are \emph{adjacent} if $\{a,b\}\in E$. Here, a singleton $\{a\}$ in
$E$ is a \emph{self-loop} (on $a$), so we say that $a$ is \emph{looped}.  A \emph{clique} is a set of pairwise adjacent vertices. In an \emph{independent set}, all vertices are pairwise non-adjacent. An
\emph{induced subgraph} of $\Gamma$ is a graph $(V',E')$, where $V'\subseteq V$
and also $E'=\{S\in E \mid S\subseteq V'\}$, i.e.\ among its vertices, $\Gamma'$
has precisely those edges that appear in $\Gamma$.  If there is no danger of
confusion, we also call a graph $\Delta$ induced subgraph of $\Gamma$ if
$\Delta$ is isomorphic to an induced subgraph of $\Gamma$.

\myparagraph{The congruence relation} To define the storage mechanism
described by a graph $\Gamma$, we need to introduce the congruence
$\equiv_\Gamma$.  To $\Gamma=(V,E)$, we associate the alphabet $\Sigma=V\cup
\bar{V}$, where we regard $V$ as an alphabet and $\bar{V}=\{\bar{a}\mid a\in
V\}$ as a decorated version of $V$.  We also define the involution
$\bar{\cdot}\colon\Sigma^*\to\Sigma^*$, where $a\mapsto \bar{a}$ and
$\bar{a}\mapsto a$ and $w=a_1\cdots a_n\mapsto \bar{a}_n\cdots \bar{a}_1$.
On $\Sigma^*$, the set of words over $\Sigma$, we consider the smallest congruence $\equiv_\Gamma$ that satisfies:
\begin{align*}
	&a\bar{a}\equiv_\Gamma\varepsilon && \text{for each $a\in V$} \\
	&xy\equiv_\Gamma yx && \text{for each $\{a,b\}\in E$, $x\in\{a,\bar{a}\}$, $y\in\{b,\bar{b}\}$}.
\end{align*}
Here, $a=b$ is allowed, meaning if $a$ is looped, then $\bar{a}a\equiv_\Gamma
a\bar{a}\equiv_\Gamma\varepsilon$.
The congruence gives rise to a \emph{monoid}, which is a set $M$ together with an associative operation and a neutral element $1\in M$.  Specifially, we have the monoid $\Mon_\Gamma := \Sigma^*/\equiv_\Gamma$. For a word $w\in\Sigma^*$, we use $[w]$ to denote the congruence class of $w$. 

\myparagraph{Monoids as storage mechanisms} Traditionally, the monoid
$\Mon_\Gamma$ is used to define ``valence automata (over
$\Mon_\Gamma$)'' that accept languages. Here, each edge carries an input word
and a word over $\Sigma$, where words $w\in\Sigma^*$ are viewed as (sequences
of) storage instructions. Then a run with instructions $w_1,\ldots,w_n$ on its
edges is accepting if $[w_1\cdots w_n]=1$. Intuitively, this is because such a
sequence $w_1\cdots w_n$ has neutral effect: A pushdown automaton accepts if
the instruction sequence brings the empty stack back to the empty stack (or
equivalently, brings any possible stack back to itself); a VASS
accepts if the instruction sequence brings the empty counters back to the empty
counters (equivalently, brings any counters valuation back to itself).
Note that this way, elements $[w]$ can be viewed both as configurations (i.e. the one reached by applying $w$) and as actions: applying $[w]$ to $[u]$ yields $[w][u]=[wu]$.

\myparagraph{Examples} With suitable $\Gamma$, one can now realize classical
storage mechanisms like pushdowns or counters. For example, if $\Gamma$ has no
edges (not even self-loops), valence automata over
$\Gamma$ are essentially pushdown automata with stack alphabet $V$. A word $w\in\Sigma^*=(V\cup\bar{V})^*$ has
$[w]=1$ if and only if $w$ can be brought to the empty word by deleting factors
$a\bar{a}$. Viewing $V$ as push instructions and $\bar{V}$ as pop
instructions, this is equivalent to $w$ being a sequence of instructions that
bring the empty stack back to the empty stack. %

Let $\Gamma$ be an \emph{unlooped clique}, i.e.\ a clique without self-loops. Then
$w\in\Sigma^*=(V\cup \bar{V})^*$ has $[w]=1$ if and only if, for each $a\in V$,
every prefix $u$ of the projection of $w$ to just $\{a,\bar{a}\}$, has the
property $|u|_a\ge|u|_{\bar{a}}$, and also $|w|_a=|w|_{\bar{a}}$. Thus,
$[w]=1$ if and only if $w$ is a sequence of increment ($a$) and decrement
($\bar{a}$) operations that bring all counters from zero to zero, 
while maintaining non-negativity. Thus valence automata over $\Gamma$ are $d$-VASS
for $d=|V|$.

If $\Gamma$ is a \emph{looped clique}, i.e.\ a clique of looped vertices, then
$w\in\Sigma^*$ has $[w]=1$ if and only if  $|w|_a=|w|_{\bar{a}}$ for each $a\in
V$, and $\Mon_\Gamma\cong\Z^d$ for $d=|V|$. Hence valence automata over looped
cliques are \emph{$\Z$-VASS}~\cite{DBLP:conf/rp/HaaseH14}: counter machines
with integer-valued counters.

\newcommand{\smallppn}[1]{%
\begin{tikzpicture}[every circle/.style={}, scale=#1]
\fill (0,0) circle (2pt) node (a) {}    (1,0) circle (2pt) node (b)  {}   (2,0) circle (2pt) node (c) {};
\draw (a.center) -- (b.center) -- (c.center);
\end{tikzpicture}
}

Finally, suppose $\Pthree$ is the graph \smallppn{0.7}, i.e. the path on three
vertices. The two outer vertices alone behave like a pushdown as above. The
middle vertex is adjacent to both pushdown vertices, and in itself, behaves
like a $1$-VASS. Thus, a valence automaton over $\Gamma$ is a pushdown
automaton with access to one additional VASS counter, also known as
one-dimensional pushdown
VASS~\cite{DBLP:conf/stoc/BiziereC25,DBLP:journals/corr/abs-2504-05015}.

\myparagraph{Viable runs and right-invertibility} In the setting of \nrgs, we do
not need a notion of accepting run, but we need a notion of a run that is
``allowed'' (or ``viable''), i.e.\ does not perform any illegal instructions.
Given that a sequence $w\in\Sigma^*$ is considered complete in valence automata
if $[w]=1$, we can define $w$ (or $[w]$) as \emph{valid} or \emph{viable} if
$w$ can be completed to $1$, i.e.\ if the element $[w]$ is right-invertible in
$M$. Here, an element $x\in M$ is called \emph{right-invertible} if there is
$y\in M$ with $xy=1$. The set of right-invertible elements of $M$ is denoted
$\RInv{M}$. Thus, we call an instruction (sequence) $v\in (V\cup \bar{V})^*$ \emph{valid} (resp.\ \emph{invalid}) in a configuration $[u]$ if $[uv]\in\RInv{\Mon_\Gamma}$ (resp.\ $[uv]\notin\RInv{\Mon_\Gamma}$).

The idea to consider right-invertible elements as valid
configurations/instruction sequences has been introduced in work on first-order
logic with reachability over configuration graphs of valence
systems~\cite{DBLP:conf/lics/DOsualdoMZ16}. As observed there, for those
concrete storage mechanisms (such as pushdowns or VASS counters) that can be
represented by $\Gamma$,
$\RInv{\Mon_\Gamma}$ corresponds to real configurations. For example, if
$\Gamma$ has no edges and thus represents a pushdown, then $\RInv{\Mon_\Gamma}$
is the set of all $[w]$ with $w\in V^*$, which corresponds to the set of stack
words over $V$, hence $\RInv{\Mon_\Gamma}\cong V^*$. On the other hand, stack words containing a pop not preceded by a push of the same letter, e.g. $[\bar{a}]$ or $[b\bar{a}]$ where $a \neq b$, are \emph{not} valid configurations. Clearly, they are not right-invertible in $\Mon_\Gamma$, and thus not in $\RInv{\Mon_\Gamma}$. If $\Gamma$ is a loopless clique, then $\RInv{\Mon_\Gamma}$ is
also the set of all $[w]$ with $w\in V^*$, but since $[ab]=[ba]$ for $a,b\in
V$, we have an isomorphism $\RInv{\Mon_\Gamma}\cong \N^d$ for $d=|V|$. Hence,
$\RInv{\Mon_\Gamma}$ is the set of valid counter valuations in a $d$-VASS. Once more, for $a, b \in V$, $[\bar{a}]$ and $[b\bar{a}]$ for $a \neq b$ are not in $\RInv{\Mon_\Gamma}$ as those are not valid configurations.

\myparagraph{Valence system}
Let $\Gamma$ be a graph. A \emph{valence system over $\Gamma$} is a tuple $(\Gamma,Q,\delta,q_\init)$, where $\Gamma=(V,E)$ is a graph, $Q$ is a finite set of \emph{states}, $\delta\subseteq Q\times (V\cup\bar{V})^*\times Q$ is the \emph{transition relation}, and $q_\init\in Q$ is its \emph{initial state}. A \emph{configuration} is a pair $(q,x)\in Q\times\Mon_\Gamma$. We denote the set of all configurations with $\Conf$. A \emph{run} is a (finite or infinite) sequence of configurations $(q_0,x_0)(q_1,x_1)\ldots \in \Conf^* \cup \Conf^\omega$ such that for each $i$, there is a transition $(q_i,w,q_{i+1})$ in $\delta$ such that $x_{i+1}=x_iy_i$, where $y_i=[w]$. Thus, after applying transitions $w_1,\ldots,w_n$, we arrive in some $(q,x)$ with $x=[w_1]\cdots [w_n]=[w_1\cdots w_n]$. %

\section{\Nrgs}\label{games}
We now introduce \nrgs\ and the corresponding decision problems.
As mentioned above, \nrgs\ generalize Energy games by including ``invalid
configurations'' in the arena (i.e.\ the results of instructions that are not
allowed), and making it part of the winning condition of $\exists$ to only use
valid/viable configurations. This leads to our notion of arenas:

\myparagraph{Valence game arena}
A \emph{(valence) game arena} is a valence system $(\Gamma, Q = Q_\exists \sqcup Q_\forall, \delta, q_{init})$ where the set of states $Q$ are parititoned into two sets $Q_\exists$, the states owned by the existential (or, $\exists$-) player and $Q_\forall$, the states owned by the universal (or, $\forall$-) player. Moreover, in a game arena, every state has at least one outgoing transition, and 
for any $p,q\in Q$, there is at most one transition from $p$ to $q$.
Thus, each graph $\Gamma$ defines a class of valence games arenas, in the same way it defines a class of valence systems. 

A \emph{play} is an infinite run.
A \emph{valence game} $(\gr, \Obj)$, is a game arena $\gr$ with an objective $\Obj \subseteq \Conf^\omega$. %
Intuitively, $\exists$-player controls $Q_\exists$ and tries to satify $\Obj$, while $\forall$-player controls $Q_\forall$ and tries to violate $\Obj$. 
For $\iota\in\{\exists,\forall\}$, we denote $\{\exists, \forall\} \setminus \{\iota\}$ by $1 - \iota$. An \emph{$\iota$-strategy} is a map $\sigma_\iota\colon Q^* Q_\iota \to Q$ such that for every $\sigma_\iota(wq) = q'$, there is a $(q, x, q') \in \delta$. %
 A play $(q_1, x_1) (q_2, x_2) \ldots$ is \emph{compliant} with $\sigma_\iota$ if for every $i$ with $q_i \in Q_\iota$, we have $\sigma_\iota(q_1 \ldots q_i) = q_{i+1}$. Note that starting from an initial configuration $(q,x)$, any state sequence $q_1 q_2 \ldots $ with $q_1 = q$ and similarly, any two strategies $\sigma_\exists$ and $\sigma_\forall$, define a unique play because there is at most one transition between any two states.
An $\iota$-strategy $\sigma_\iota$ is \emph{winning} from a configuration $(q,x)$ if \emph{all plays} starting from $(q,x)$ and compliant with $\sigma_\iota$ are in $\Obj$ if $\iota=\exists$, and not in $\Obj$ if $\iota=\forall$. A configuration $(q,x)$ is \emph{won} by the player that has a winning strategy from $(q,x)$. If the objective is a Borel set, each configuration is won by one of the players~\cite{Borel}. %

\begin{remark}
	The assumption that every state has an outgoing transition is
	non-essential for the objectives we consider, but simplifies the
	definitions. \onlyConference{See the full version~\cite{full} for details.}\onlyFull{See \cref{app:dead-ends} for details.}
\end{remark}

\myparagraph{Graph classes} A single graph $\Gamma$ represents a specific
storage mechanism, such as $2$-VASS if $\Gamma$ is \twob{0.75}. Thus,
understanding a decision problem for each individual $\Gamma$ 
tells us little about the setting where part of the storage
mechanism (such as the dimension of a VASS) is given in the input. Therefore, we
consider the more general setting where the graph $\Gamma$ is also part of the
input, but drawn from some class $\Gr$. For example, if $\Gamma$ is the class
of unlooped cliques, then this captures the setting of $\VASS$ where the
dimension is arbitrary and part of the input.  We will assume that $\Gr$ is
closed under taking induced subgraphs.

\myparagraph{\Nrgs} We are ready to define \nrgs.  Here, the fact that only valid/viable configurations should be
used is part of the objective for $\exists$-player. %
Consider a valence game arena $(\Gamma,Q=Q_\exists\sqcup Q_\forall,\delta,q_\init)$.
We
define the objective $\RIO\subseteq \Conf^\omega$, short for
``right-invertibility objective'', as the set of plays $(q_1,x_1)(q_2,x_2) \ldots$ where $x_i\in\RInv{\Mon_\Gamma}$ for every $i\in\N$. In other words, the play only visits valid/viable configurations.
 $\RIO$ is a Borel set.
\myparagraph{Decision problems} We have the following decision problems related to \nrgs.
In the \emph{fixed initial credit} setting, the initial configuration (or, \emph{credit}) $[w]$ is given. For a class $\Gr$ of graphs, $\FICEG(\Gr)$ is the following problem:
\begin{description}
	\item[Given:] A graph $\Gamma=(V,E)$ from $\Gr$, a valence game arena over $\Gamma$, and a word $w\in(V\cup \bar{V})^*$.
	\item[Question:] Does $\exists$-player win $(q_\init,[w])$ for the objective $\RIO$?
\end{description}
We also have the \emph{unknown initial credit} setting, where the initial
configuration is to be determined.  For a class $\Gr$ of graphs, $\UICEG(\Gr)$
is the following problem:
\begin{description}
	\item[Given:] A graph $\Gamma=(V,E)$ from $\Gr$ and a valence game arena over $\Gamma$.
	\item[Question:] Is there a $w\in(V\cup \bar{V})^*$ such that $\exists$-player wins $(q_\init,[w])$ for the objective $\RIO$?
\end{description}

\section{Main results}\label{main-results}
In this section, we present the main results of this work.
\Cref{thm:FIC-EG} provides a full characterization of decidability and complexity of $\FICEG$, and \cref{thm:UIC-EG} provides such a characterization of $\UICEG$. We begin with some notation needed in the statements.

\myparagraph{The VASS dimension} %
Intuitively, the VASS dimension of $\Gr$ is the largest dimension of VASS that can be simulated using storage mechanisms over $\Gr$. A \emph{$d$-clique} is a clique of $d \in \N$ vertices. Recall that an (un)looped $d$-clique represents $d$-dim.\ $\Z$-VASS ($d$-dim.\ VASS). Therefore, we denote by $\ZVASS$ (resp.\ $\VASS$) the class of looped (resp.\ unlooped) cliques.
For a class $\Gr$ of graphs, 
we define its \emph{VASS dimension} 
 $\Vdim(\Gr)$ as the supremum over all $d$ such that some graph in $\Gr$ contains an unlooped $d$-clique as an induced subgraph.

\myparagraph{Stacks of configurations}
For some graphs $\Gamma$, valence systems over $\Gamma$ behave like pushdown systems where the stack entries can be configurations of other storage mechanisms.
Suppose a graph $\Gamma=(W \cup U,E)$ is obtained from a graph $\Delta=(W,E)$ by adding a set $U$ of unlooped independent (i.e.~pairwise non-adjacent) vertices. Then the elements of $\RInv{\Mon_\Gamma}$ can be written as $[w_{r+1}u_r w_r \cdots u_1 w_1]$, where $u_1,\ldots,u_r\in U$ and $[w_1],\ldots,[w_{r+1}]\in\RInv{\Mon_\Delta}$. This can be viewed as a stack content with letters $u_1,\ldots,u_r$, interspersed with configurations $[w_1],\ldots,[w_{r+1}]$ corresponding to $\Delta$. Then, multiplying some $w\in(W\cup\bar{W})^*$ will yield $[w_{r+1}u_rw_r\cdots u_1w_1w]$, i.e.~apply $w$ to the top-most entry. Multiplying some $u\in U$ will yield $[w_{r+1}u_rw_r\cdots u_1 w_1 u]$ and thus start a fresh entry on top. Finally, applying $\bar{u}$ will yield a right-invertible element (i.e.~be valid/viable) if and only if $u=u_1$ (i.e.~we are popping the right letter) and $[w_1]=1$ (the top-most entry is empty). If that is the case, then this yields $[w_{r+1}u_rw_r\cdots u_1w_1\bar{u}]=[w_{r+1}u_rw_r\cdots u_{2}w_{2}]$, thus popping the entire top-most entry.

For these reasons, we denote by $\Pushdown(\Gr)$ the class of graphs $\Gamma$ as above, where $\Delta$ is drawn from $\Gr$: We think of these as pushdown systems over $\Gr$. Formally, $\Pushdown(\Gr)$ is the class of graphs $\Gamma = (W \cup U, E)$ where $U$ is a set of unlooped independent vertices, and $(W, E) \in \Gr$. We also write $\Pushdown$ for all graphs that are just unlooped independent graphs, i.e. $(U, \emptyset)$.

\myparagraph{Direct products}
Suppose $\Gamma$ is obtained from $\Delta_1$ and $\Delta_2$ by taking a disjoint union of the vertices, and making any two vertices $v_1$ from $\Delta_1$ and $v_2$ from $\Delta_2$ adjacent. Then $\Mon_\Gamma\cong\Mon_{\Delta_1}\times\Mon_{\Delta_2}$ and thus valence systems over $\Gamma$ can be viewed as having access to \emph{both} storage mechanisms $\Delta_1$ and $\Delta_2$, used independently. We therefore write $\Gamma=\Delta_1\times\Delta_2$, and extend this notation to classes of graphs, i.e. $\Gr_1\times\Gr_2$. For example, $\Pushdown\times\VASS$ are the graphs corresponding to pushdown vector addition systems~\cite{DBLP:conf/stoc/BiziereC25,DBLP:conf/icalp/LerouxST15,DBLP:journals/corr/abs-2504-05015}.

\myparagraph{Groups}
If $\Gamma$ has only looped vertices, then $\Mon_\Gamma$ is a \emph{group},
meaning for every $x\in\Mon_\Gamma$, there is a $y\in\Mon_\Gamma$ with $xy=1$:
Indeed, for $w\in(V\cup\bar{V})^*$, we have $[w\bar{w}]=1$, since
$\bar{a}a\equiv_\Gamma a\bar{a}\equiv_\Gamma\varepsilon$ for any $a\in V$.
Therefore, we denote by $\Group$ the class of graphs with only looped vertices.
Note that for a class $\Gr$, an algorithm for $\FICEG(\Gr\times\Group)$ or
$\UICEG(\Gr\times\Group)$ can just ignore all letters in the $\Group$ components,
since they never stand in the way of right-invertibility.

\myparagraph{The main results}
We are now ready to state our main results. We first have a full description of
the decidability and complexity landscape of the fixed initial credit setting:
\begin{maintheorem}\label{thm:FIC-EG}
    Let $\Gr$ be a class of graphs closed under induced subgraphs.
    \begin{enumerate}
        \item If $\Vdim(\Gr)<\infty$ and $\Gr\subseteq \VASS \times \Group$, then $\FICEG(\Gr)$ is $\PTIME$-complete. \label{num:A1}
	\item If $\Vdim(\Gr)<\infty$ and $\Gr\subseteq \Pushdown (\Group) \times \Group \cup \VASS \times \Group$, but $\Gr\not\subseteq\VASS \times \Group$, then $\FICEG(\Gr)$ is $\EXPTIME$-complete. \label{num:A2}
        \item If $\Vdim(\Gr)=\infty$ and $\Gr\subseteq \Pushdown (\Group) \times \Group \cup \VASS \times \Group $, then $\FICEG(\Gr)$ is $\twoEXPTIME$-complete. \label{num:A3}
        \item In all other cases, $\FICEG(\Gr)$ is undecidable.\label{num:A4}
    \end{enumerate}
\end{maintheorem}
Likewise, we have a full description of the decidability and complexity
landscape of the unknown initial credit setting:
\begin{maintheorem}\label{thm:UIC-EG}
    Let $\Gr$ be a class of graphs closed under induced subgraphs.
    \begin{enumerate}
        \item If $\Vdim(\Gr)<\infty$ and $\Gr\subseteq \VASS \times \Group$, then $\UICEG(\Gr)$ is $\PTIME$-complete. \label{num:B1}
         \item If $\Vdim(\Gr)=\infty$ and $\Gr\subseteq\VASS \times \Group$, then $\UICEG(\Gr)$ is $\coNP$-complete. \label{num:B2}
         \item If $\Gr\subseteq \Pushdown (\Group) \times \Group \cup \VASS \times \Group $, but $\Gr\not\subseteq \VASS \times \Group$, then $\UICEG(\Gr)$ is $\EXPTIME$-complete. \label{num:B3}
        \item In all other cases, $\UICEG(\Gr)$ is undecidable.\label{num:B4}
    \end{enumerate}
\end{maintheorem}

Let us briefly interpret \cref{thm:FIC-EG,thm:UIC-EG} in terms of concrete
storage mechanisms. First, note that $\Gr \subseteq \VASS \times \Group$ means
that $\Gr$ corresponds to VASS (up to a group component, which does not affect
the game). Then, $\Vdim(\Gr) < \infty$ means the VASS dimension is bounded by
some fixed $d \in \N$, while $\Vdim(\Gr) = \infty$ means $\Gr$ corresponds to
VASS of arbitrary dimension. 
Moreover, graphs in $\Pushdown(\Group)$ correspond to group pushdown systems
(see \cref{preview}) where the group is a graph group.

Thus, in both theorems, case~(a) corresponds to fixed-dimensional VASS.  In
\cref{thm:FIC-EG}, case~(b) means we allow group pushdowns or fixed-dimensional
VASS, and case~(c) is the same as (b), except that the VASS have arbitrary
dimension.  In \cref{thm:UIC-EG}, case~(b) corresponds to arbitrary VASS,
whereas case~(c) means we have some group pushdown system, and VASS of either
fixed or arbitrary dimension.

\myparagraph{Discussion: Decidability through viability semantics}
Our results show that fixed and unknown initial credit problems are decidable precisely for a class $\Gr$ included in $\Pushdown(\Group)\times\Group~\cup~\VASS\times\Group$. Since a
direct product with $\Group$ is irrelevant for $\FICEG$ and $\UICEG$, this means
the interesting settings are $\Pushdown(\Group)$ and $\VASS$. 

Here, $\FICEG(\VASS)$ and $\UICEG(\VASS)$ are the same as multi-dimensional energy games,
meaning our results reveal $\Pushdown(\Group)$ as a new class of systems with
decidable games.

This new class $\Pushdown(\Group)$ includes $\Pushdown(\ZVASS)$, i.e.~stacks
where entries contain vectors over $\Z^d$, for some $d\in\N$. These are part of
the class of \emph{stacked counter
automata}~\cite{DBLP:conf/stacs/Zetzsche15,DBLP:journals/iandc/Zetzsche21}.
Having a decidable game type for $\Pushdown(\ZVASS)$ is noteworthy, because
these systems behave as unfavorably to classical games as $\VASS$: In games
where the arena ensures validity of instructions, (i)~control-state
reachability games, (ii)~configuration reachability games, and
(iii)~non-termination games (where the objective is just to construct an arbitrary
infinite run), are all undecidable for $\Pushdown(\ZVASS)$. For (i),(ii), this
was shown in \cite[Thms.~19.2.1,19.2.10]{Muskalla23}. \onlyFull{For (iii), we provide a
simple proof in \cref{app:infinite-runs-pdtwoz}.}\onlyConference{For (iii), we provide a
simple proof in the full version~\cite{full}.}

Therefore, in $\Pushdown(\Group)$, we identify a new class of systems where
the viability semantics achieves decidability.

Decidability of \nrgs\ for $\Pushdown(\Group)$ is noteworthy for another reason:
In $\Pushdown(\Group)$ (beyond $\Pushdown(\ZVASS)$), even ordinary
(i.e.\ single-player) control-state reachability is undecidable. This is
because the rational subset membership problem for graph groups (which is
undecidable in general~\cite[Theorem 2]{lohrey2008submonoid}) reduces to
control-state reachability in $\Pushdown(\Group)$.

\myparagraph{Discussion: Subtle decidability border}
Let us briefly discuss the subtle decidability border for $\FICEG$ and $\UICEG$ for graphs of the form $\Pushdown(\twocounters{0.75})$, i.e.\ a stack of two counters, which are either $\Z$-VASS counters or VASS counters. If both are VASS counters, then \nrgs\ are undecidable by a similar reduction as for ordinary safety games for $2$-VASS. It turns out, the construction can be adapted if one counter is a $\Z$ counter, i.e.\ for $\Pushdown(\btimesz{0.75})$. However, if both are $\Z$ counters, then we are in $\Pushdown(\Group)$, where \nrgs\ are decidable.

\section{Proof overview and key ingredients}\label{proof-overview}

\definecolor{myBlue}{HTML}{648FFF}
\definecolor{myPurple}{HTML}{785EF0}
\definecolor{myMagenta}{HTML}{DC267F}
\definecolor{myOrange}{HTML}{FE6100}
\definecolor{myYellow}{HTML}{FFB000}

\definecolor{myGreen}{HTML}{1B9E77} %

In this section, we give an overview of the proofs of \cref{thm:FIC-EG,thm:UIC-EG}. The individual ingredients will be shown in the subsequent sections.

\newcommand{\nodedistance}{0.7}
\begin{figure}[h]
  \centering
  \begin{minipage}{0.23\textwidth}
    \centering
    \begin{tikzpicture}[scale=1]
      \foreach \x in {0,1,2} {
	\coordinate (p\x) at (\nodedistance*\x,0);
	\fill (p\x) circle (2pt);
	}
      \draw (p1) -- (p2);
      \draw [densely dotted] (p0.center) ++(90:3pt) circle (3pt);
      \draw (p2.center) ++(90:3pt) circle (3pt);
      \node[left=1.5mm of p0] {\small (i)};
    \end{tikzpicture}
  \end{minipage}
  \begin{minipage}{0.23\textwidth}
    \centering
    \begin{tikzpicture}[scale=1]
      \foreach \x in {0,1,2} {
	\coordinate (p\x) at (\nodedistance*\x,0);
	\fill (p\x) circle (2pt);
	}
      \draw (p1) -- (p2);
      \draw[densely dotted] ($(p0.center)+(0,3pt)$) circle [radius=3pt]; 
      \node[left=1.5mm of p0] {\small (ii)};
    \end{tikzpicture}
  \end{minipage}
   \begin{minipage}{0.23\textwidth}
    \centering
    \begin{tikzpicture}[scale=1]
      \foreach \x in {0,1,2} {
	\coordinate (p\x) at (\nodedistance*\x,0);
	\fill (p\x) circle (2pt);
	}
      \draw (p0) -- (p1);
      \draw (p1) -- (p2);
      \draw[draw=none] ($(p0.center)+(0,3pt)$) circle [radius=3pt];
      \node[left=1.5mm of p0] {\small (iii)};
    \end{tikzpicture}
  \end{minipage}
    \begin{minipage}{0.23\textwidth}
    \centering
    \begin{tikzpicture}[scale=1]
      \foreach \x in {0,1,2} {
	\coordinate (p\x) at (\nodedistance*\x,0);
	\fill (p\x) circle (2pt);
	}
      \draw (p0) -- (p1);
      \draw (p1) -- (p2);
      \draw ($(p2.center)+(0,3pt)$) circle [radius=3pt];
      \node[left=1.5mm of p0] {\small (iv)};
    \end{tikzpicture}
  \end{minipage} 
  \caption{Illegal graphs. Dotted edges indicate that both versions (with or without the edge) are illegal. For each graph, we name the vertices $a$, $b$ and $c$ from left-to-right.}   \label{fig:illegal-graphs}
\end{figure}
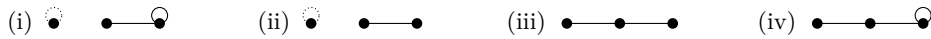

\myparagraph{Undecidable cases}
Let us begin by mentioning those graphs for which we show
undecidability (both of $\FICEG$ and $\UICEG$). We call the graphs shown in \cref{fig:illegal-graphs} the \emph{illegal graphs}. Note that in \cref{fig:illegal-graphs}, a dotted edge means it can be present or not (and both variants are deemed illegal). We will show that for these graphs, both $\FICEG$ and $\UICEG$ are undecidable:
\begin{restatable}{theorem}{undecidability}\label{thm:undecidability}
If $\Gamma$ is an illegal graph, then both $\FICEG(\Gamma)$ and $\UICEG(\Gamma)$ are undecidable.
\end{restatable}

\myparagraph{Pure VASS-like cases} As apparent from \cref{thm:FIC-EG,thm:UIC-EG}, the decidable cases are those where every graph in $\Gr$ belongs to $\Pushdown(\Group)\times\Group$ or to $\VASS\times\Group$ (or both). Since in our setting, a direct product with a group will have no influence on the game (as right-invertibility is not affected), we may as well assume that the graphs belongs to $\Pushdown(\Group)\cup\VASS$.
For graphs in $\VASS$, the problems $\FICEG$ and $\UICEG$ are just energy games with unary-encoded numbers, and thus complexity follows from existing results:
\begin{restatable}{theorem}{complexityVASScase}\label{complexity-vass-case}
Let $\Gr\subseteq\VASS\times\Group$. If $\Vdim(\Gr)=\infty$, then 
$\FICEG(\Gr)$ is $\twoEXPTIME$-complete and $\UICEG(\Gr)$ is $\coNP$-complete.
If $\Vdim(\Gr)<\infty$, then both
$\FICEG(\Gr)$ and $\UICEG(\Gr)$ are $\PTIME$-complete.
\end{restatable}
It is immediate from the definitions that the case $\Vdim(\Gr)=\infty$
corresponds to energy games where the dimension is arbitrary (and part of the
input). Here, the fixed initial credit setting $\FICEG$ is
$\twoEXPTIME$-complete~\cite[Cor.~6.1]{JurdzinskiLS15} and the unknown initial
credit setting $\UICEG$ is $\coNP$-complete~\cite[Thm.~9]{ChatterjeeDHR10}.
Moreover, $\Vdim(\Gr)<\infty$ is the case of fixed dimension. Here, both
settings are $\PTIME$-complete, see \cite[Thms.~3.3\&3.5]{JurdzinskiLS15}.

\myparagraph{Pure Pushdown-like cases}
Now suppose $\Gr\subseteq\Pushdown(\Group)\times\Group$, or equivalently
$\Gr\subseteq\Pushdown(\Group)$. Here we discover a novel decidability result:
While control-state reachability games and configuration reachability games
(and even ordinary control-state reachability, without games!) are undecidable,
viability games are decidable for $\Pushdown(\Group)$ storage mechanisms.
Moreover, we will show that both $\FICEG$ and $\UICEG$ are $\EXPTIME$-complete
for this class:
\begin{restatable}{theorem}{mainPushdown}\label{main-pushdown}
Suppose $\Gr\subseteq\Pushdown(\Group)\times\Group$. Then $\FICEG(\Gr)$ and $\UICEG(\Gr)$ are in $\EXPTIME$.  $\EXPTIME$-hardness holds as soon as $\Gr$ contains a graph that does not belong to $\VASS\times\Group$.
\end{restatable}
Regarding hardness, note that a graph in $\Pushdown(\Group)\times\Group$ belongs to $\VASS\times\Group$ if and only if it contains at most one unlooped vertex $v$, and all the unlooped vertices are adjacent to $v$. Thus, $\Gamma$ does not belong to $\VASS\times\Group$ if and only if $\Gamma$ contains (i)~two unlooped vertices or (ii)~one unlooped vertex $v$ and a looped vertex $u$ that is not adjacent to $v$. In both cases, one can reduce classic pushdown (non-termination) games to $\FICEG(\Gamma)$ and $\UICEG(\Gamma)$ which is known to be $\EXPTIME$-hard~\cite{Walukiewicz01}.

The interesting part of \cref{main-pushdown} is the decidability and $\EXPTIME$ upper bound of $\FICEG$ and $\UICEG$. These will be shown in \cref{sec:FICEG,sec:UICEG}.

\myparagraph{Completeness of the characterization} In addition to \cref{thm:undecidability,complexity-vass-case,main-pushdown}, we also need to show that every class $\Gr$ of graphs falls into one of the cases in \cref{thm:FIC-EG,thm:UIC-EG}. This will follow from the following, which we prove in \cref{sec:completeness}:
\begin{restatable}{theorem}{withoutIllegalSubgraph}\label{thm:without-illegal-subgraph}
Every graph that contains no illegal graph as an induced subgraph belongs to
$\Pushdown (\Group ) \times \Group \cup \VASS \times \Group$.
\end{restatable}

\myparagraph{Key technical ingredients} Our proofs contain several novel
technical ingredients. First, proving undecidability in the case of (i)
(\cref{fig:illegal-graphs}) requires a subtle encoding of two-counter machine:
Valence systems over (i) have a pushdown, where each entry contains one
$\N$-counter and one $\Z$-counter. Note that such an undecidability
proof is not possible if both are $\Z$-counters, as then we are in
the $\Pushdown(\ZVASS)$ setting, which is decidable. 

Second, our algorithm for $\FICEG(\Pushdown(\Group))$ relies on the insight
that if the stack contains a group element that \emph{does not have} a
representation among polynomial length words, then $\forall$ can force a
victory, because these elements are ``too big to cancel in time''. This implies
that detecting winning strategies can be confined to group elements with
representations of polynomial length $\ell$. However, a naive implementation of
this would encode the group elements in an exponential-sized stack alphabet
(resulting in a $\twoEXPTIME$ algorithm). Indeed, we cannot store the group
elements as words on the stack: After expanding them (potentially beyond length
$\ell$), we need to find a new representation of length $\le\ell$ (if one exists),
but this requires substantial rewriting far from the ends of the word. To
mitigate this, we devise an algorithm that, as in pure pushdown
systems~\cite{Cachat02}, saturates only sets of (sets of) control states. This
yields the optimal $\EXPTIME$ algorithm.

Third, our algorithm for $\UICEG(\Pushdown(\Group))$ reduces to
$\FICEG(\Pushdown(\Group))$. To this end, a combinatorial argument shows that
initial configurations of exponential stack height (and polynomial-sized
group-valued entries) always suffice. However, an additional trick is needed to
allow $\exists$ to guess an exponential-sized stack without having the
opportunity to grow it infinitely. We achieve this by adapting a recent
technique from indexed
languages~\cite{mandel2026complexitydownwardclosuresindexed}.

Fourth, a combinatorial argument shows that all graphs that avoid the illegal
graphs (\cref{fig:illegal-graphs}) belong to $\VASS\times\Group$ or to
$\Pushdown(\Group)\times\Group$. 

Note that in the literature on valence
systems, complete descriptions of decidability (let alone complexity)
landscapes are rare. For decidability, this has been achieved for first-order logic with
reachability~\cite{DBLP:conf/lics/DOsualdoMZ16}, certain underapproximations of
reachability that are designed to guarantee
decidability~\cite{DBLP:conf/concur/MeyerMZ18,DBLP:conf/concur/ShettyKZ21}, and
games on valence system arenas~\cite{Muskalla23} (but here,
only trivial cases are decidable).

\section{Group pushdowns with fixed initial credit}\label{sec:FICEG}
In this section, we show that $\FICEG(\Pushdown(\Group)\times\Group)$ can be
solved in $\EXPTIME$. Since in a group, every element is right-invertible, this
problem trivially reduces to $\FICEG(\Pushdown(\Group))$.  

\myparagraph{Group elements on the stack} However, the group elements
\emph{on the stack} do pose a non-trivial challenge. Indeed, recall that if
$\Gamma\in\Pushdown(\Group)$, then an element in $\RInv{\Mon_\Gamma}$ is of the
form $[w_{r+1}\gamma_{r}w_{r}\cdots \gamma_1w_1]$, where $\gamma_1,\ldots,\gamma_r$ are
from an independent, unlooped set of vertices (i.e.\ pushdown letters), and
$w_1,\ldots,w_{r+1}$ are over looped vertices (i.e.\ generators of a group).  Now, if
we perform a pop of $\gamma$ in this configuration, then it depends on the
top-most group element $[w_1]$ (and $\gamma_1$) whether this leads to a valid
configuration: If $[w_1]=1$ and $\gamma_1=\gamma$, then this will lead to the
right-invertible $[w_{r+1}\gamma_rw_r\cdots \gamma_{1}w_{1}]$, but if $[w_1]\ne
1$ or $\gamma\ne \gamma_1$, then the result is not right-invertible. Since
\nrgs\ for pushdown systems can easily be reduced via classical safety games on
pushdown systems, the (infinitely many) group elements on the stack constitute
the main difficulty.

\myparagraph{Preparation: single pushdown letter} We may
assume that the initial credit $\lambda$ is empty, since this can easily be encoded
in the arena.  We may also assume that there
is only a single pushdown letter (aside from group elements): If there
are pushdown letters $\{\gamma_1,\ldots,\gamma_k\}$ and looped vertices
$W$, then each entry $\gamma_iw$ with $w\in (W\cup\bar{W})^*$ in a configuration is encoded by $\gamma g^i \gamma w$, where $\gamma$ is a single fresh pushdown letter, and $g$ is an arbitrary letter in $W$. Thus, we encode one entry by two entries, the first of which encodes the specific original letter $\gamma_i$ used. This works because for $g\in W$, the elements $g^i$ for $i\in\N$ are pairwise distinct.
\begin{restatable}{lemma}{lemSingleUnloopedNode}\label{single-unlooped-node}
	$\FICEG(\Pushdown(\Group))$ reduces in poly-time to the case of a
	single unlooped vertex.
\end{restatable}
\onlyFull{Details are in~\cref{sec:app-FICEG}}\onlyConference{Details can be found in the full version~\cite{full}}. Thus we assume $\Gamma=(V,E)$, with
$V=\{\gamma\}\cup W$, where $\gamma$ is unlooped, the vertices in $W$ are all
looped, and there are no edges between $\gamma$ and $W$. Clearly, we may assume
that all transitions carry either a letter in $V\cup\bar{V}$ or $\varepsilon$.

\myparagraph{Strategy trees} Winning strategies for $\exists$ describe
infinite runs%
, meaning it is not immediate how to represent them
finitely.  In contrast, winning strategies for $\forall$ have finite
representations: Since a play winning for $\forall$ must visit an invalid
configuration (i.e.~$(q,[u])$ with $[u]\notin\RInv{\Mon_\Gamma}$), and
everything thereafter is irrelevant, we can represent the strategy in a finite
tree. Therefore, our algorithm searches for a winning strategy for $\forall$.

Specifically, a $\iota$-\emph{strategy tree} for $\iota \in \{\exists, \forall\}$ is a tree with
nodes in $Q\times\Mon_\Gamma$, edges are labeled by $V\cup \bar{V}\cup\{\varepsilon\}$,
such that the following holds: If a node $(q,[u])$ has $q\in
Q_{1-\iota}$, then for every transition $q\xrightarrow{v} q'$ in $\delta$, the
node $(q,[u])$ has a child $(q',[uv])$. If $q\in Q_{\iota}$, then the node
has a single child $(q',[uv])$ for some transition $q\xrightarrow{v}q'$ in $\delta$.
Here, the edges are always labeled by the multiplied $v$.
We call the tree \emph{winning} if all branches (i.e.\ plays) are winning for $\iota$-player. For $\forall$-strategy trees this means each
branch visits some $(q,[u])$ with $[u]\notin\RInv{\Mon_\Gamma}$. We consider finite winning $\forall$-strategy trees where every branch has the first such $(q,[u])$ as a leaf.
Our task is thus to detect winning $\forall$-strategy trees.

\myparagraph{Saturating along stack height}
If we had no group elements, then detecting $\forall$-strategy trees would
amount to solving a safety game on a pushdown graph. Here, a standard approach
is a saturation procedure~\cite{Cachat02}. The data that is being saturated is, for each $q\in
Q$, the set $R_q$ of all subsets $P\subseteq Q$ such that, roughly speaking,
there is a $\forall$-strategy tree with root node $(q,[\varepsilon])$ where we allow
leaves that (i)~are invalid (i.e. $\notin\RInv{\Mon_\Gamma}$), or (ii)~end in
an empty configuration $[\varepsilon]$, but in $P$. We call such a ``relaxed
winning $\forall$-strategy tree'' a \emph{$P$-tree}. Then clearly, $\emptyset\in R_q$ if
and only if there is a winning $\forall$-strategy tree rooted in $(q,[\varepsilon])$. 

In the pushdown setting, $R_q$ can be computed via saturation w.r.t.\ the
maximal stack height: Let $R^i_{q}$ be the set of all $P\subseteq Q$ such that
there is a \emph{$(P,i)$-tree}, i.e.\ a $P$-tree where all stack heights across
branches are $\le i$. Without group elements, it is not hard to compute
$R^{i+1}_{q}$ from $R^i_{q}$ via emptiness of finite-state tree automata.

\myparagraph{Main challenge: Group elements}
The main challenge in our algorithm is to perform the above saturation in the
presence of group elements. Note that very closely related tasks are already
undecidable: It is undecidable (for certain $\Gamma$) whether a given regular
$L\subseteq (W\cup \bar{W})^*$ contains some $w\in(W\cup \bar{W})^*$ with
$[w]=1$~\cite[Thm.~2]{lohrey2008submonoid}. Thus, it is undecidable whether, even in
the single-player setting, there is ever a valid pop.

\myparagraph{Key idea: Big \& small elements}
Let us sketch how we deal with group elements, in the case of stack height $0$.
The crucial observation is that group elements can be ``big'' or ``small''. Intuitively,
small elements just need polynomially many bits, whereas big elements
are ``too big to cancel in time'' and thus do not need to be stored exactly. 

Let us make this precise. With $\Delta=(W,E)$, let $g\in\Mon_\Delta$ be some
group element.  Let $\|g\|$ denote the \emph{word length} of $g$, i.e.\ the
length of the shortest word $w\in (W\cup\bar{W})^*$ with $g=[w]$. We say that
$q\in Q$ is \emph{$0$-eligible} if there exists a $\forall$-strategy tree where (i)~$q$
is in the root and (ii)~every branch contains a pop, and no push. Clearly, in any
$(P,0)$-tree, every ocurring state must be $0$-eligible. The following observation about big elements is key:
\begin{lemma}\label{big-elements-zero}
	Suppose $p$ is $0$-eligible and $g\in\Mon_\Delta$ has $\|g\|>|Q|$.
	Then $(p,wg)$ is the root of some $(\emptyset,0)$-tree for every $w \in \{[\varepsilon]\} \cup \{[u \gamma] \mid u \in (V \cup \bar{V})^*\}$. 
\end{lemma}
This is because there is a tree $t$ witnessing $0$-eligibility of height at
most $|Q|$: In a tree higher than $|Q|$, we can cut a repeating occurrence of a
state, and still obtain a witness. However, on this witness, every branch will
multiply a group element of word length at most $|Q|$. Therefore, making
$(p,wg)$ the root node will yield a $(P,0)$-tree: In every branch, we multiply
some element $[v]$ with $|v|\le |Q|$ to $g$, but since $\|g\|>|Q|$, we know that
$g[v]\ne 1$: Otherwise we would have $g=[\bar{v}]$. Therefore, the pop in each
branch must be applied with a group element $\ne 1$ and is thus invalid---the
tree is a $(\emptyset,0)$-tree. We say $g\in\Mon_\Delta$ is \emph{big} if $\|g\|>|Q|$.

\myparagraph{Cayley graphs} By \Cref{big-elements-zero}, it suffices to track \emph{small} elements, i.e.\ those with word length $\le|Q|$, of $\Mon_\Delta$ exactly. We do this using a standard construction in group theory: The \emph{Cayley graph of $\Mon_\Delta$}, denoted $\mathcal{C}$, is the directed graph with $\Mon_{\Delta}$ as its vertex set, and edges $g\xrightarrow{x}g'$ for each $g,g'\in\Mon_\Delta$ and $x\in W\cup\bar{W}$ with $g'=g\cdot [x]$. Of course, $\mathcal{C}$ is infinite, but we consider a variant where all big elements are collapsed into one: We obtain the \emph{restricted Cayley graph} $\mathcal{C}^{\le |Q|}$ by removing all big nodes and inserting a fresh ``trap vertex'' $c^{\trap}$: All edges that arrive in big elements in $\mathcal{C}$ are redirected into $c^{\trap}$.  

We will regard $\mathcal{C}^{\le|Q|}$ as an automaton with exponentially many
states: Note that it can be computed in $\EXPTIME$: An element
$g\in\Mon_\Delta$ can be represented by the length-lexicographically minimal
word $u\in(W\cup \bar{W})^*$ with $g=[u]$. Moreover, given words $u,u'\in
(W\cup\bar{W})^*$, one can decide in poly-time whether $[u]=[u']$: This is
because the groups of the form $\Mon_\Delta$, also called \emph{right-angled
Artin groups}, embed into matrix groups (over the
reals)~\cite{DavisJanuszkiewicz2000,HsuWise1999}, for which the word problem is
in logspace~\cite{LiptonZalcstein1977}. One can thus find a minimal
representative in $\EXPTIME$. 

\myparagraph{The base case} With \cref{big-elements-zero} and
$\mathcal{C}^{\le|Q|}$ in hand, we can now compute $R^0_{q}$: Given a set $P\subseteq Q$,
we construct a finite (top-down) tree automaton that, roughly speaking, accepts $(P,0)$-trees with $q$ in their root. However, it restricts nodes to $0$-eligible states, and so when in some node $s$, the element of $\Mon_{\Delta}$ becomes big, then the automaton can ``declare victory'' (for $\forall$): It accepts this node as a leaf, because \cref{big-elements-zero} tells us that this node can be completed to even a $(\emptyset,0)$-tree. The automaton can therefore use $\mathcal{C}^{\le|Q|}$ to track the group element on the current branch. This tree automaton has exponentially many states, and it accepts some tree if and only if $P\in R^0_{q}$. Thus, we can compute all of $R^0_{q}$ in $\EXPTIME$.

\myparagraph{Saturation step} Once we have computed $R^i_{q}$ for every $q\in
Q$, we compute $R^{i+1}_{q}$ with similar ideas. Consider a $(P,i+1)$-tree with
some push of $\gamma$ in some node $\nu$ in state $q'$ that moves the stack
height from $0$ to $1$. In each branch below this push, we either
(i)~eventually come back to stack height $0$ and arrive in some state $p\in Q$,
or (ii)~encounter an invalid pop. Let $s$ be the subtree consisting of all
paths until (i) or (ii) occurs. Then the set $P'$ of all states $p$ as above,
satisfies $P'\in R^i_{q'}$. We can therefore ``shortcut'' $s$: Instead of
performing the $\gamma$ push, we take one subtree for each $p\in P'$ and place
it directly under $\nu$.  Repeating this shortcutting will lead to a tree with
stack height $0$, but with special nodes whose children are given by some
$P'\in R^i_{q'}$ for some $q'$. The trees of this form can now be detected by a
finite tree automaton: It is constructed as in the base case, except that each
$P'\in R^i_{q'}$ provides a new transition to simulate shortcuts. Here, we use
a fact analogous to \cref{big-elements-zero}: Note that there the bound will
still be $|Q|$, since the automata for eligibility for higher $i$ have
exponentially many transitions, but still only $|Q|$ states.

\myparagraph{Overall complexity} This way, we can compute $R^{i+1}_{q}$ in
$\EXPTIME$ (w.r.t.\ input size) from the sets $R^i_{q'}$ for $q'\in Q$. Since
$R^i_{q}\subseteq R^{i+1}_{q}$ for every $q,i$, this means the overall saturation
terminates in exponentially many steps, leading to an $\EXPTIME$ upper bound
overall.
\begin{remark}
	A naive implementation of the idea in
	\cref{big-elements-zero} would reduce $\FICEG(\Pushdown(\Group))$ to a
	pushdown game with exponential stack alphabet, resulting only in
	a $\twoEXPTIME$ upper bound.
\end{remark}
\begin{remark}
	Although our result is only about $\Pushdown(\Group)$, where the group
	$G$ on the stack is of the form $\Mon_\Delta$ for a looped 
	$\Delta$, the algorithm even works for any finitely generated group
	with a decidable word problem (which asks whether a given sequence of
	generators yields the identity): This suffices to compute
	an equivalent representatives of length $\le |Q|$, if it exists. 
\end{remark}

\section{Group pushdowns with unknown initial credit}\label{sec:UICEG}
In this section, we show that $\UICEG(\Pushdown(\Group))$ reduces to $\FICEG(\Pushdown(\Group))$ in polynomial time. Together with the $\EXPTIME$ upper bound for $\FICEG(\Pushdown(\Group))$ in \cref{sec:FICEG}, this establishes \cref{main-pushdown}.

Let $\Gamma=(V,E)$ belong to $\Pushdown(\Group)$ and consider a valence game
arena $\gr$. As in \cref{single-unlooped-node}, we may assume that
$V=\{\gamma\}\cup W$ such that $\gamma$ is isolated and $\Delta=(W,E)$ has a
self-loop on every vertex. Recall that an element $x\in\RInv{\Mon_\Gamma}$ can be
written as $x=[w]$ with $w=w_{r+1}\gamma w_r \cdots \gamma w_1$ for
$w_1,\ldots,w_{r+1}\in (W\cup \bar{W})^*$.  We call $r=|w|_\gamma$ the
\emph{$\Pushdown$-length} of $w$. By $\|w\|$, we denote the maximal word length
$\|[w_i]\|$ among all elements $[w_i]\in\Mon_\Delta$.

\myparagraph{Ingredient I: Bounding the initial credit} Our first ingredient is a
combinatorial argument that shows if a \nrg\ with $n$ states over
$\Pushdown(\Group)$ is won with some initial credit $\lambda\in (V\cup \bar{V})^*$, then
there is such a $\lambda$ where with $|\lambda|_\gamma < 2^n$ and $\|\lambda\|\le n$. 
\begin{restatable}{theorem}{initialCreditUb}\label{thm:initial-credit-ub}
	If the \nrg\ $\UICEG(\gr)$ is won by $\exists$, then $\exists$ wins $(q_{init}, [\lambda])$ where $q_{init}$ is the initial state of $\gr$ and
	$\lambda\in(V\cup\bar{V})^*$ where $|\lambda|_\gamma <
	2^n$ and $\|\lambda\|\le n$.
\end{restatable}
\begin{proof}[Proof sketch]
	We sketch the proof for $|\lambda|_\gamma < 2^n$. The bound $\|\lambda\|\le n$
	uses arguments similar to~\cref{big-elements-zero}. Consider a winning
	$\exists$-strategy tree $\tau$ from $(q_\init,[\lambda])$ for $\lambda=w_{r+1}\gamma w_{r}\cdots
	\gamma w_1$ with $r>2^n$. We will argue that we can cut some infix
	$\gamma w_i\cdots \gamma w_{j}$ from $\lambda$.

	For each branch in $\tau$, its \emph{popping depth} is the maximal $s$
	such that the suffix $\gamma w_{s}\cdots \gamma w_1$ is popped during
	this branch.
	Consider the case that all branches have a popping depth $\ge 2^n$. For
	every $i\in[1,2^n]$, let $S_i\subseteq Q$ be the (non-empty) set of all
	states visited on some branch after popping the suffix $\gamma
	w_i\cdots \gamma w_1$ for the first time. There must be
	$j,\ell\in[1,2^n]$, $j<\ell$, with $S_j\supseteq S_\ell$. But then
	we can remove the infix $\gamma w_\ell \cdots \gamma w_{j+1}$
	from $\lambda$: By reattaching subtrees whose root contains $w_{r+1}\gamma
	w_r\cdots \gamma w_{\ell+1}$ to other nodes that used to contain $w_{r+1}\gamma
	w_r\cdots \gamma w_{j+1}$, but have the same control-state, we obtain a winning
	$\exists$-strategy tree for the shorter initial credit $w_{r+1}\gamma
	w_r\cdots \gamma w_{\ell+1} \cdot \gamma
	w_j\cdots \gamma w_1$.

	If there are branches with popping depth $s<2^n$, then for $i\leq s$,
	they contribute to $S_i$ just as above. But for $i> s$, we
	expand $S_i$ by a state $p$ such that this branch visits $p$ infinitely
	often at its lowest stack height. Then the same cut-and-reattach procedure works.
\end{proof}

\myparagraph{Ingredient II: Guessing the initial credit} We now use
\cref{thm:initial-credit-ub} to reduce $\UICEG(\Pushdown(\Group))$ to
$\FICEG(\Pushdown(\Group))$. Here, we construct a valence game arena
$\gr^{\expp}$ where $\exists$ first guesses an initial credit $\lambda$ with
$|\lambda|_\gamma < 2^n$ and $\|\lambda\|\le n$, and then simulates the game. Here, the
difficulty is to guess a stack content with exponentially many entries without allowing
$\exists$ to win outright by just growing and growing the stack.

To this end, we modify a construction from \cite[Section 9]{mandel2026complexitydownwardclosuresindexed}: The idea is to not only guess a stack with $< 2^n$ entries, but a stack with decorated letters. Then, each time $\exists$ pushes a pushdown letter, $\forall$ can check whether the stack still has $< 2^n$ entries. But in fact, $\forall$ checks that the decoration adheres to a pattern that implies length $< 2^n$. This pattern is a suffix-closed set of words that can be recognized by running linearly many DFAs in parallel.
Let $\gamma_{[i, j]}$ to denote the set $\{\gamma_k \mid k \in [i,j]\}$. 

\begin{restatable}{lemma}{DFAs}\label{lemma:DFAs} For every $m \in \N$, there are $m$ two-state DFAs $(D_i)_{i \in [1, m]}$ over $\gamma_{[1, m]}$ such that $\bigcap_{i\in[1,m]} L(D_i)$ is suffix-closed and has a unique longest word of length $2^m-1$.
\end{restatable}
See \cref{fig:DFA-i} for an illustration of $D_i$. Intuitively, $D_i$ accepts all words where between any two $\gamma_{i}$'s, there is at least one occurrence of $\gamma_{[1,i-1]}$. Thus, a word in the intersection can have at most one $\gamma_1$, at most two $\gamma_2$'s, etc.\ and at most $2^{i-1}$-many $\gamma_i$'s. Thus, the unique longest word in the intersection has length $\Sigma_{i \in [1,m]} 2^{i-1}  = 2^m-1$ and is a palindrome.

\begin{remark}
In \cite[Section 9]{mandel2026complexitydownwardclosuresindexed} a similar
	DFA-based construction in indexed/grammars/alternating PDAs is used to
	ensure a stack height of exactly $2^n$. Here, we need to apply this
	check after every push, necessitating a check of ``$\le 2^n$''.  Thus,
	our DFAs need to be suffix-closed, and hence we need a different
	construction than their %
	``binary counting'' DFAs.
\end{remark}

\myparagraph{Constructing the game}
With \cref{thm:initial-credit-ub,lemma:DFAs}, we are prepared to present the reduction.
To this end, we construct the game $\gr^\expp$, which is given in~\cref{fig:G-expp}.
 First, $\exists$ collects an initial credit $\lambda'$ in the loop $p_0 - p_n$ before entering $\gr[\gamma \gets \gamma_{[1,n]}]$: the arena $\gr$ where every push/pop is replaced by $n$-many pushes/pops, one for every $\gamma_i \in \gamma_{[1, n]}$.
Say $\exists$ wins the game from $(q_\init,\lambda)$ for some $\lambda =  [w_{r+1} \gamma \cdots \gamma w_1]$ that obeys~\cref{thm:initial-credit-ub}. She collects $\lambda$ before entering $\gr[\gamma_i \gets \gamma_{[1, n]}]$ with one caveat: she collects $\lambda' = [w_{r+1} \gamma^{r} \cdots \gamma^1  w_1]$ such that the push-sequence $\gamma^{r} \cdots \gamma^1$ of $\lambda'$ is the $r$-length suffix of the $(2^n-1)$-length word $w_\cap$ accepted by $\cap_{[1, n]} L(D_i)$. Clearly, $\exists$ wins $(q_\init,\lambda)$ iff she wins on $(\gr[\gamma\gets \gamma_{[1,n]}]$ from $\lambda'$.

 Player $\forall$ can challenge the accumulated initial credit at every loop from his state $\texttt{ch}$ that has an outgoing edge to the initial state of $\texttt{D}_i$, \enquote{the game-version} of $D_i$, for each $i \in [1, n]$ (see~\cref{fig:DFA-i}(b)). Player $\forall$ wins the subgame $\texttt{D}_i$ if (i) the push-sequence of the accumulated credit is not accepted by $D_i$, or (ii) there is a group element of word length $>n$ inbetween two pushes. Otherwise, $\exists$ wins $\texttt{D}_i$. As $\lambda$ obeys~\cref{thm:initial-credit-ub} and the language of $\cap_{i \in [1,n]} L(D_i)$ is suffix-closed, $\exists$ wins the subgame $\cap_{i \in [1, n]}\texttt{D}_i$ entered with any prefix of $\lambda'$. As she also wins on $\gr[\gamma \gets \gamma_{[1,n]}]$ from $\lambda'$, she wins $\gr^{\expp}$ from $(p_{in}, [\varepsilon]$) -- $\gr^\expp$ is a positive $\FICEG$ instance.

Suppose $\gr$ is a negative $\UICEG$ instance, i.e.\ there is no initial credit with which $\exists$ wins on $\gr[\gamma \to \gamma_{[1,n]}]$. Then, $\exists$ stays in the loop $p_0 - p_n$, and $\forall$ wins by waiting until the $\Pushdown$-length of the accumulated stack word $w$ reaches $2^n$ and moving to a $\texttt{D}_i$ for which $w \not \in L(D_i)$.

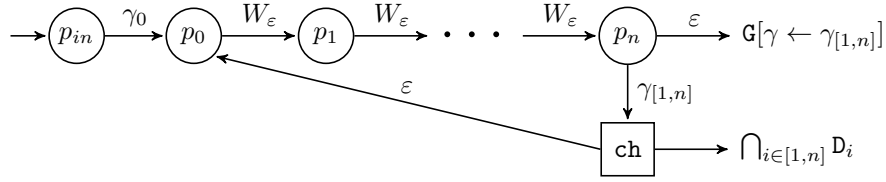
\begin{figure}[t]
  \centering
  \tikzset{every state/.style={minimum size=50pt}}
  \tikzset{every loop/.style={looseness=2}}
  \begin{tikzpicture}[
    ->, >=stealth', shorten >=1pt,
    auto,
    node distance=1.5cm,
    semithick
  ]  
  \node[circle, draw] (s) {\phantom{$p_0$}}; %
  \node at (s.center) {$p_{\tiny in}$};
  \node[circle, draw] (p0) [right= 0.8cm of s] {$p_0$};
  \node[circle, draw] (p1) [right= 1cm of p0] {$p_1$};

  \node (dots) [right= 1cm of p1] {\huge{$\cdots$}};

  \node[circle, draw] (pn) [right=1cm of dots] {$p_n$};
  \node (game) [right=1cm of pn] {$\gr[\gamma \gets \gamma_{[1,n]}]$};
  \node[rectangle, draw, minimum width=7mm, minimum height=7mm] (x) [below of=pn] {$\texttt{ch}$};

   \node (Dcap)  [right=1cm of x]  {$\bigcap_{i \in [1, n]}\texttt{D}_i$};

  \draw (s.west) ++(-0.5cm,0) -- (s);

  \draw (s)  -- node[above] {$\gamma_0$} (p0);
  \draw (p0) -- node[above] {${W}_\varepsilon$} (p1);
  \draw (p1) -- node[above] {${W}_{\varepsilon}$} (dots);
  \draw (dots) -- node[above] {${W}_{\varepsilon}$} (pn);

  \draw (pn) -- node[right] {$\gamma_{[1,n]}$} (x);
  \draw (pn) -- node[above] {$\varepsilon$} (game);

  \draw (x.west) -- node[above] {$\varepsilon$} (p0.south east);

    \draw (x) -- (Dcap);

  \end{tikzpicture}
  \caption{${W}_{\varepsilon}$ stands for $W \cup \bar{W} \cup \{\varepsilon\}$. $\texttt{D}_i$ is the game-version of the DFA $D_i$, depicted in~\cref{fig:DFA-i}.}
  \label{fig:G-expp}
\end{figure}

\begin{figure}[t]
  \centering
  \tikzset{every state/.style={minimum size=50pt}}
  \tikzset{every loop/.style={looseness=2}}
 \begin{subfigure}{0.45\textwidth}\label{fig:DFA-version}
    \centering
    \begin{tikzpicture}[
      auto,
      node distance=1.5cm,
      semithick
      ]
       \node[circle, draw, double] (di) {$d^i$};
        \node[circle, draw, double] (bar-di)  [right= 1.3cm of di] {$\tilde{d}^i$};
        \draw[->] (di) edge [loop above] node {$\gamma_{[1,n]} \setminus \{\gamma_i\}$} ();
        \draw[->] (bar-di) edge [loop above] node {$\gamma_{[i+1,n]}$} ();

        \draw[->] (di) to[bend right=15] node[below] {$\gamma_i$} (bar-di);
        \draw[->] (bar-di) to[bend right=15] node[above] {$\gamma_{[1,i-1]}$} (di);
         
        \draw[->] (di.west) ++(-0.5cm,0) -- (di);

    \end{tikzpicture}
    \caption{Figure of $D_i$. Both nodes are accepting.} %
    \end{subfigure}
    \hfill
    \begin{subfigure}{0.45\textwidth}\label{fig:DFA-i-game-version}
 \centering
    \begin{tikzpicture}[
      auto,
      node distance=1.5cm,
      semithick
      ]

       \node[circle, draw, line width=2.5pt, opacity=0.45] (di) {\phantom{$d^i$}};
        \node at (di.center) {$\mathbf{d}^i$};
        \node[circle, draw, line width=2.5pt, opacity=0.45] (bar-di)  [right= 1.3cm of di] {\phantom{$\tilde{d}^i$}};
        \node at (bar-di.center) {$\mathbf{\tilde{d}}^i$};
        \node (sadey) [right of=bar-di] {\scalebox{2.5}{$\Sadey$}};
        \node (smiley) [below right=9mm and 3mm of di] {\scalebox{2.5}{$\Smiley$}};
        \node (smiley-shadow) at ($(smiley.center)+(0,-0.5mm)$) {\phantom{\scalebox{1.5}{$\Smiley$}}};

        \draw[->] (di) edge [loop above] node {$\overline{\gamma}_{[1,n]} \setminus \{\overline{\gamma}_i\}$} ();
        \draw[->] (bar-di) edge [loop above] node {$\overline{\gamma}_{[i+1,n]}$} ();
        \draw[->] (bar-di) -- node[above] {$\overline{\gamma}_i$} (sadey); 
         
         \draw[->] (sadey) edge [loop above] node {$\overline{\gamma}_{0}$} ();

        \draw[->] (di) to[bend right=15] node[below] {$\overline{\gamma}_i$} (bar-di);
        \draw[->] (bar-di) to[bend right=15] node[above] {$\overline{\gamma}_{[1,i-1]}$} (di);
        \draw[->] (bar-di) to[bend left=15] node[below,  xshift=3mm, yshift=2mm] {$\overline{\gamma}_0$}  (smiley-shadow);
        \draw[->] (di) to[bend right=15] node[below, xshift=-2mm, yshift=2mm] {$\overline{\gamma}_0$} (smiley-shadow);

        \draw[->] (smiley) edge [loop right] node {$\varepsilon$} ();

        \draw[->] (di.west) ++(-0.5cm,0) -- (di);

    \end{tikzpicture}
    \caption{Figure of $\texttt{D}_i$. $\mathbf{d}^i$ and $\mathbf{\tilde{d}}^i$ 
    contain ($n$+1)-states letting $\exists$ read $g\in \Mon_\Delta$ with $\|g\|\leq n$.}
    \end{subfigure}
  \caption{The DFA $D_i$ used in~\cref{lemma:DFAs}, and its game version  $\texttt{D}_i$ used in $\gr^\expp$ (\cref{fig:G-expp}).}
  \label{fig:DFA-i}
\end{figure}

\section{Undecidability results}\label{sec:undecidability-results}
We now prove \cref{thm:undecidability}, i.e.\ the undecidability results
entailed by \cref{thm:FIC-EG,thm:UIC-EG}.  Recall that the vertices of the graphs (i)---(iv) (\cref{fig:illegal-graphs}) are called $a,b,c$, left to right.

 \begin{figure}
  \centering
  \tikzset{every state/.style={minimum size=50pt}}
  \tikzset{every loop/.style={looseness=2}}
    \begin{tikzpicture}[
      auto,
      node distance=1.5cm,
      semithick
      ]
       \node[circle, draw] (qi) {$q_i$};
       \node[rectangle, draw] (qij) [right of=qi] {$q^{\forall}_{i,j}$};
       \node[circle, draw] (qj) [right of=qij] {$q_j$};
       \node[circle, draw] (qij-exists) [above of=qij] {\phantom{$q_i$}};
       \node at (qij-exists) (qij-label) {$q^{\exists}_{i,j}$};
       \node (smiley) [right= 1.3cm of qij-exists] {\scalebox{2.5}{$\Smiley$}};
       \draw[->] (qi) -- node[above] {$\varepsilon$} (qij); 
       \draw[->] (qij) -- node[above] {$\varepsilon$} (qj); 
       \draw[->] (qij) -- node[left] {$\varepsilon$} (qij-exists); 
       \draw[->] (qij-exists) -- node[above] {$\bar{a}\bar{b}$} (smiley); 
       \draw[->] (smiley) edge [loop above] node {$\varepsilon$} ();
       \draw[->] (qij-exists) edge [in=120, out=153, loop] node {} ();

       \node[above left=2mm of qij-exists, yshift=-0.5cm, xshift=-0.05cm] (labelc) {\textcolor{myMagenta}{$c^{-}$}};   
       \node[above left=2mm of qij-exists, xshift=0.7cm, yshift=0.13cm] (labelb) {\textcolor{myBlue}{$b^{-}$}};   

    \end{tikzpicture}
    \caption{Reduction of (ii): $\texttt{zero}(k)$ gadgets for $k \in \{1, 2\}$. The gadgets differ only in the colored edges, where
    $\texttt{zero}(1)$ gadget takes the pink ($\bar{c}$) edge and $\texttt{zero}(2)$ gadget takes the blue ($\bar{b}$) edge.}
    \label{fig:gadgets-ii}
  \end{figure}
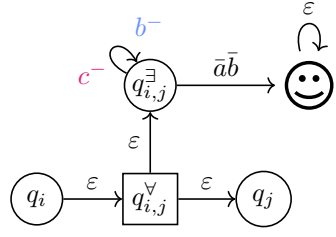

\begin{restatable}{theorem}{undeciiiANDiv}\label{thm:undec-iii-iv}  Let $\Gamma$ be one of the valence graphs (iii) or (iv). Then $\FICEG(\Gamma)$ and $\UICEG(\Gamma)$ are both undecidable. 
\end{restatable}
This follows almost directly from Abdulla, Atig, Hofman, Mayr, Kumar and Totzke~\cite{AbdullaAHMKT14}, who show undecidability of pushdown games with a single additional counter, and where the objective includes non-negativity of that counter (invalid pops, however, are not visible in the arena). Moreover, valence systems over (iii) are precisely one-dimensional pushdown VASS. %
 In the case (iv), we need a small modification, because the right-most vertex has a self-loop: We replace $c$ ($\bar{c}$) transitions with $cac$ ($\bar{c}\bar{a}\bar{c}$). \onlyFull{See \cref{app:undec-iii-iv} for details of both reductions.}\onlyConference{See the full version~\cite{full} for details of both reductions.}

\begin{restatable}{theorem}{undeciANDii}\label{thm:undec-i-ii}  Let $\Gamma$ be one of the valence graphs represented by (i) or (ii). Then $\FICEG(\Gamma)$ and $\UICEG(\Gamma)$ are both undecidable. \end{restatable}
\myparagraph{The case (ii)} Valence systems over (ii) belong to
$\Pushdown(\VASS)$, i.e.\ they have a stack whose entries are $d\ge 0$ VASS
counters, and (ii) is the case of $d=2$. Here, $a$ is the pushdown letter, and
$b$ and $c$ correspond to the counters.  We exploit that a pop
is only valid if the top-most entry's counters are both zero. This permits a
reduction from two-counter machines (with instructions $\mathtt{dec}(k)$, $\mathtt{inc}(k)$, $\mathtt{zero}(k)$, for $k=1,2$), similar to undecidability of
games on VASS arenas~\cite[Thm.~13.1]{GamesOnGraphs}. We still sketch (ii), 
to set the stage for the new idea for (i).

For (ii), we use a stack with two
entries (the lower containing a single $b$), and the top-most entry
containing the two counter values: We encode $(m_1,m_2)\in\N^2$ by
$bab^{m_1}c^{m_2}$. Increments and decrements are mimiced directly on
$b$ and $c$. For zero-tests, the gadget in \cref{fig:gadgets-ii}
lets $\forall$ win if the tested counter, say $b$, was not zero: If
$\forall$ challenges, then $\exists$ first gets a chance to
reset the untested counter, say $c$, in $q_{i,j}^\exists$.  Since this
cannot go on forever (this would eventually lead to an illegal
decrement), $\exists$ must move right eventually, but this is only
victorious for $\exists$ if both counters are actually zero: Otherwise, we obtain an
element $bab^{m_1}c^{m_2}\bar{a}\bar{b}$ with $m_1,m_2$ not both zero,
which is not right-invertible (here, the single $b$ is there to make
this element non-right-invertible even if $a$ has a self-loop). \onlyFull{See
\cref{app:undec-ii} for details.}\onlyConference{See the full version~\cite{full} for details.}

\myparagraph{The case (i)} With (i), we need a new trick. Here, we have a pushdown where each entry contains one $\N$-valued counter ($b$) and one $\Z$-valued counter ($c$). The gadget \cref{fig:gadgets-ii} would not work, because when zero-testing $b$, the $\bar{c}$ loop would let $\exists$ win just by looping $\bar{c}$. 

The trick is to encode $(m_1,m_2)\in\N^2$ as $bab^{m_1+m_2}c^{m_2}$.
Then, in the new zero-test gadget (\cref{fig:gadgets-i}, right), when
$\exists$ gets a chance to deplete the untested counter in $q^\exists_{i,j}$,
each of her loops will include a decrement on $b$, which means (as above) she
cannot do this forever. However, the encoding as $bab^{m_1+m_2}c^{m_2}$ does not in itself guarantee $m_1\ge 0$ and $m_2\ge 0$: This element can be right-invertible despite $m_1<0$ or $m_2<0$. Therefore, decrements also use a gadget (\cref{fig:gadgets-i}, left) where $\forall$ wins as soon as $m_1<0$ or $m_2<0$. \onlyFull{See \cref{app:undec-i} for details.}\onlyConference{See the full version~\cite{full} for details.}

 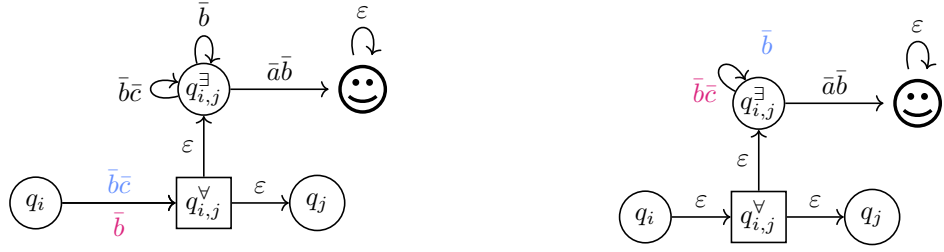
\begin{figure}
  \centering
  \tikzset{every state/.style={minimum size=50pt}}
  \tikzset{every loop/.style={looseness=2}}
  \begin{subfigure}{0.45\textwidth}
    \centering
    \begin{tikzpicture}[
      auto,
      node distance=1.5cm,
      semithick
      ]
       \node[circle, draw] (qi) {$q_i$};
       \node[rectangle, draw] (qij) [right= 1.5cm of qi] {$q^{\forall}_{i,j}$};
       \node[circle, draw] (qj) [right of=qij] {$q_j$};
       \node[circle, draw] (qij-exists) [above of=qij] {\phantom{$q_i$}};
       \node at (qij-exists) (qij-label) {$q^{\exists}_{i,j}$};
       \node (smiley) [right= 1.3cm of qij-exists] {\scalebox{2.5}{$\Smiley$}};
       \draw[->] (qi) -- node[above]{\textcolor{myBlue}{$\bar{b}\bar{c}$}} (qij); 
       \draw[->] (qi) -- node[below] {\textcolor{myMagenta}{$\bar{b}$}} (qij); 
       \draw[->] (qij) -- node[above] {$\varepsilon$} (qj); 
       \draw[->] (qij) -- node[left] {$\varepsilon$} (qij-exists); 
       \draw[->] (qij-exists) -- node[above] {$\bar{a}\bar{b}$} (smiley); 
       \draw[->] (smiley) edge [loop above] node {$\varepsilon$} ();
       \draw[->] (qij-exists) edge [loop above] node {$\bar{b}$} ();
       \draw[->] (qij-exists) edge [loop left] node {$\bar{b}\bar{c}$} ();
    \end{tikzpicture}
    \end{subfigure}
    \hfill
    \begin{subfigure}{0.45\textwidth}
 \centering
    \begin{tikzpicture}[
      auto,
      node distance=1.5cm,
      semithick
      ]
       \node[circle, draw] (qi) {$q_i$};
       \node[rectangle, draw] (qij) [right of =qi] {$q^{\forall}_{i,j}$};
       \node[circle, draw] (qj) [right of=qij] {$q_j$};
       \node[circle, draw] (qij-exists) [above of=qij] {\phantom{$q_i$}};
       \node at (qij-exists) (qij-label) {$q^{\exists}_{i,j}$};
       \node (smiley) [right= 1.3cm of qij-exists] {\scalebox{2.5}{$\Smiley$}};
       \draw[->] (qi) -- node[above]{$\varepsilon$}  (qij); 
       \draw[->] (qij) -- node[above] {$\varepsilon$} (qj); 
       \draw[->] (qij) -- node[left] {$\varepsilon$} (qij-exists); 
       \draw[->] (qij-exists) -- node[above] {$\bar{a}\bar{b}$} (smiley); 
       \draw[->] (smiley) edge [loop above] node {$\varepsilon$} ();
       \draw[->] (qij-exists) edge [in=120, out=153, loop] node {} ();

       \node[above left=2mm of qij-exists, yshift=-0.5cm, xshift=-0.05cm] (labelc) {\textcolor{myMagenta}{$\bar{b}\bar{c}$}};   
       \node[above left=2mm of qij-exists, xshift=0.7cm, yshift=0.13cm] (labelb) {\textcolor{myBlue}{$\bar{b}$}};   

    \end{tikzpicture}
    \end{subfigure}
    \caption{Reduction of (i): $\texttt{dec}(k)$ (left) and $\texttt{zero}(k)$ (right) gadgets for $k \in \{1, 2\}$. Gadgets differ only in the colored edges, where $k=1$ and $2$ gadgets take pink (down) and blue (up) edges, resp.}
    \label{fig:gadgets-i}
  \end{figure}

\section{Completeness of characterization}\label{sec:completeness}
In this section, we prove \cref{thm:without-illegal-subgraph}.
\begin{proof}[Proof of \cref{thm:without-illegal-subgraph}]
  Suppose $\Gamma$ contains no illegal graph as an induced subgraph. Let $U$ be the set of unlooped vertices of $\Gamma$. If $U$ is empty, then $\Gamma$ even belongs to $\Group$, and thus to the union as desired. Hence, we assume $U\ne\emptyset$.  We distinguish two cases, illustrated in \cref{fig:without-illegal-subgraph}.

  First, suppose $U$ contains two adjacent vertices $a \ne b$. Then every other vertex in $V$ must be adjacent to both $a$ and $b$, as otherwise, $\Gamma$ would contain one of the illegal graphs (ii), (iii), or (iv). This means, the vertices in $U$ form an unlooped clique, and all looped vertices are adjacent to all vertices in $U$. Thus, $\Gamma$ belongs to $\VASS\times\Group$.

  Now suppose the vertices in $U$ are independent. We partition the \emph{looped}
vertices of $\Gamma$ into the sets $X$ and $Y$, where (i)~$X$ contains those
looped vertices that are adjacent to some vertex in $U$ and (ii)~$Y$ contains all
other looped vertices. %

\newcommand{\radius}{7pt}
\begin{figure}[h]
  \centering
  \begin{minipage}{0.47\textwidth}
    \centering
    \begin{tikzpicture}[scale=0.6]

  \node[circle, fill=myPurple, inner sep=2pt] (Lneck) at (0.5,2) {};
  \node[circle, fill=myPurple, inner sep=2pt] (Rneck) at (1.5,2) {};
  \node[left=4mm of Lneck, text=myPurple] {$L$};

  \node[circle, fill=myGreen, inner sep=2pt] (A) at (0.5,1) {};
  \node[circle, fill=myGreen, inner sep=2pt] (B) at (1.5,1) {};
  \node[circle, fill=myGreen, inner sep=2pt] (C) at (0.5,0) {};
  \node[circle, fill=myGreen, inner sep=2pt] (D) at (1.5,0) {};
  \node[left=4mm of C, text=myGreen] {$U$};

  \draw[myPurple] ($(Rneck.north)+(0,4pt)$) circle [radius=\radius];
  \draw[myPurple] ($(Lneck.north)+(0,4pt)$) circle [radius=\radius];

  \draw[myPurple] (Lneck) -- (A);
  \draw[myPurple] (Rneck) -- (B);
  \draw[myPurple] (Lneck) -- (B);
  \draw[myPurple] (Rneck) -- (A);

  \draw[myGreen] (A) -- (B);
  \draw[myGreen] (C) -- (D);
  \draw[myGreen] (A) -- (C);
  \draw[myGreen] (B) -- (D);
  \draw[myGreen] (A) -- (D);
  \draw[myGreen] (B) -- (C);

  \draw[myPurple] (Lneck) -- (D);
  \draw[myPurple] (Rneck) -- (C);

\draw[myPurple] (Lneck.west)
  to[out=200, in=160, looseness=0.8] (C.west);

\draw[myPurple] (Rneck.east)
  to[out=-20, in=20, looseness=0.8] (D.east);
\end{tikzpicture}
    \caption*{{\color{myGreen}$\VASS$}$ \times ${\color{myPurple}{$\Group$}}}
  \end{minipage}\hfill
  \begin{minipage}{0.47\textwidth}
    \centering
    \begin{tikzpicture}[scale=0.6]

  \node[circle, fill=myOrange, inner sep=2pt] (Ua) at (0, 0.5) {};
  \node[circle, fill=myOrange, inner sep=2pt] (Ub) at (0, 1.5) {}; 
  \node[circle, fill=myOrange, inner sep=2pt] (Uc) at (0, 2.5) {};
  \node[above=4mm of Uc, text=myOrange] {$U$};

  \node[circle, fill=myPurple, inner sep=2pt] (Xa) at (2, 0.5) {};
  \node[circle, fill=myPurple, inner sep=2pt] (Xb) at (2, 1.5) {};
  \node[circle, fill=myPurple, inner sep=2pt] (Xc) at (2, 2.5) {};
  \node[above=4mm of Xc, text=myPurple] {$X$};

  \node[circle, fill=myMagenta, inner sep=2pt] (Ya) at (4, 0.5) {};
  \node[circle, fill=myMagenta, inner sep=2pt] (Yb) at (4, 1.5) {};
  \node[circle, fill=myMagenta, inner sep=2pt] (Yc) at (4, 2.5) {};
  \node[above=4mm of Yc, text=myMagenta] {$Y$};

  \draw[myPurple] ($(Xa.south)-(0,4pt)$) circle [radius=\radius];
  \draw[myPurple] ($(Xb.south)-(0,4pt)$) circle [radius=\radius];
  \draw[myPurple] ($(Xc.north)+(0,4pt)$) circle [radius=\radius];

  \draw[myMagenta] ($(Ya.south)-(0,4pt)$) circle [radius=\radius];
  \draw[myMagenta] ($(Yb.north)+(0,4pt)$) circle [radius=\radius];
  \draw[myMagenta] ($(Yc.north)+(0,4pt)$) circle [radius=\radius];

  \draw[myPurple] (Ua) -- (Xa);
  \draw[myPurple] (Ub) -- (Xa);
  \draw[myPurple] (Uc) -- (Xa);
  \draw[myPurple] (Ua) -- (Xb);
  \draw[myPurple] (Ub) -- (Xb);
  \draw[myPurple] (Uc) -- (Xb);
  \draw[myPurple] (Ua) -- (Xc);
  \draw[myPurple] (Ub) -- (Xc);
  \draw[myPurple] (Uc) -- (Xc);

  \draw[myPurple] (Ya) -- (Xa);
  \draw[myPurple] (Yb) -- (Xa);
  \draw[myPurple] (Yc) -- (Xa);
  \draw[myPurple] (Ya) -- (Xb);
  \draw[myPurple] (Yb) -- (Xb);
  \draw[myPurple] (Yc) -- (Xb);
  \draw[myPurple] (Ya) -- (Xc);
  \draw[myPurple] (Yb) -- (Xc);
  \draw[myPurple] (Yc) -- (Xc);

  \draw[myPurple] (Xb) -- (Xc);
  \draw[myMagenta] (Ya) -- (Yb);
\end{tikzpicture}
    \caption*{{\color{myOrange}$\Pushdown$}$(${\color{myMagenta}$\Group$}$)\times${\color{myPurple}$\Group$}}
  \end{minipage}\caption{Any $\Gamma$ without an illegal graph as induced subgraph is in one of these two classes.}\label{fig:without-illegal-subgraph}
\end{figure}

\textit{Claim 1:} All vertices in $X$ are adjacent to all vertices in $U$. Indeed, otherwise there is an $x\in X$ not adjacent to some $u\in U$. However, $x$ is adjacent to some $u'\in U$ by definition of $X$, and so $u,u',x$ induce an illegal graph~(i). This establishes Claim~1.

\textit{Claim 2:} All vertices in $X$ are adjacent to all vertices in $Y$. Indeed, otherwise there are non-adjacent $x\in X$ and $y\in Y$. But since there is $u\in U$, the vertices $y,u,x$ together induce an illegal graph~(i). This establishes Claim~2.

Clearly, the graph induced by $U\cup Y$ belongs to $\Pushdown(\Group)$. By our claims, the vertices in $X$ are adjacent to all vertices in $U \cup Y$, hence the entire $\Gamma$ belongs to $\Pushdown(\Group)\times\Group$.
\end{proof}

\label{beforebibliography}
\newoutputstream{pages}
\openoutputfile{main.pages.ctr}{pages}
\addtostream{pages}{\getpagerefnumber{beforebibliography}}
\closeoutputstream{pages}
\bibliography{references}

\begin{thebibliography}{10}

\bibitem{AbdullaAHMKT14}
Parosh~Aziz Abdulla, Mohamed~Faouzi Atig, Piotr Hofman, Richard Mayr,
  K.~Narayan Kumar, and Patrick Totzke.
\newblock Infinite-state energy games.
\newblock In Thomas~A. Henzinger and Dale Miller, editors, {\em Joint Meeting
  of the Twenty-Third {EACSL} Annual Conference on Computer Science Logic
  {(CSL)} and the Twenty-Ninth Annual {ACM/IEEE} Symposium on Logic in Computer
  Science (LICS), {CSL-LICS} '14, Vienna, Austria, July 14 - 18, 2014}, pages
  7:1--7:10. {ACM}, 2014.
\newblock \href {https://doi.org/10.1145/2603088.2603100}
  {\path{doi:10.1145/2603088.2603100}}.

\bibitem{DBLP:conf/concur/AbdullaMSS13}
Parosh~Aziz Abdulla, Richard Mayr, Arnaud Sangnier, and Jeremy Sproston.
\newblock Solving parity games on integer vectors.
\newblock In Pedro~R. D'Argenio and Hern{\'{a}}n~C. Melgratti, editors, {\em
  {CONCUR} 2013 - Concurrency Theory - 24th International Conference, {CONCUR}
  2013, Buenos Aires, Argentina, August 27-30, 2013. Proceedings}, volume 8052
  of {\em Lecture Notes in Computer Science}, pages 106--120. Springer, 2013.
\newblock \href {https://doi.org/10.1007/978-3-642-40184-8_9}
  {\path{doi:10.1007/978-3-642-40184-8_9}}.

\bibitem{DBLP:conf/lics/AnandSSZ24}
Ashwani Anand, Sylvain Schmitz, Lia Sch{\"{u}}tze, and Georg Zetzsche.
\newblock Verifying unboundedness via amalgamation.
\newblock In Pawel Sobocinski, Ugo~Dal Lago, and Javier Esparza, editors, {\em
  Proceedings of the 39th Annual {ACM/IEEE} Symposium on Logic in Computer
  Science, {LICS} 2024, Tallinn, Estonia, July 8-11, 2024}, pages 4:1--4:15.
  {ACM}, 2024.
\newblock \href {https://doi.org/10.1145/3661814.3662133}
  {\path{doi:10.1145/3661814.3662133}}.

\bibitem{DBLP:conf/stoc/BiziereC25}
Clotilde Bizi{\`{e}}re and Wojciech Czerwinski.
\newblock Reachability in one-dimensional pushdown vector addition systems is
  decidable.
\newblock In Michal Kouck{\'{y}} and Nikhil Bansal, editors, {\em Proceedings
  of the 57th Annual {ACM} Symposium on Theory of Computing, {STOC} 2025,
  Prague, Czechia, June 23-27, 2025}, pages 1851--1862. {ACM}, 2025.
\newblock \href {https://doi.org/10.1145/3717823.3718149}
  {\path{doi:10.1145/3717823.3718149}}.

\bibitem{DBLP:conf/icalp/BrazdilJK10}
Tom{\'{a}}s Br{\'{a}}zdil, Petr Jancar, and Anton{\'{\i}}n Kucera.
\newblock Reachability games on extended vector addition systems with states.
\newblock In Samson Abramsky, Cyril Gavoille, Claude Kirchner, Friedhelm~Meyer
  auf~der Heide, and Paul~G. Spirakis, editors, {\em Automata, Languages and
  Programming, 37th International Colloquium, {ICALP} 2010, Bordeaux, France,
  July 6-10, 2010, Proceedings, Part {II}}, volume 6199 of {\em Lecture Notes
  in Computer Science}, pages 478--489. Springer, 2010.
\newblock \href {https://doi.org/10.1007/978-3-642-14162-1_40}
  {\path{doi:10.1007/978-3-642-14162-1_40}}.

\bibitem{DBLP:conf/mfcs/BuckheisterZ13}
P.~Buckheister and Georg Zetzsche.
\newblock Semilinearity and context-freeness of languages accepted by valence
  automata.
\newblock In Krishnendu Chatterjee and Jir{\'{\i}} Sgall, editors, {\em
  Mathematical Foundations of Computer Science 2013 - 38th International
  Symposium, {MFCS} 2013, Klosterneuburg, Austria, August 26-30, 2013.
  Proceedings}, Lecture Notes in Computer Science, pages 231--242. Springer,
  2013.
\newblock \href {https://doi.org/10.1007/978-3-642-40313-2_22}
  {\path{doi:10.1007/978-3-642-40313-2_22}}.

\bibitem{Cachat02}
Thierry Cachat.
\newblock Symbolic strategy synthesis for games on pushdown graphs.
\newblock In Peter Widmayer, Francisco~Triguero Ruiz, Rafael~Morales Bueno,
  Matthew Hennessy, Stephan~J. Eidenbenz, and Ricardo Conejo, editors, {\em
  Automata, Languages and Programming, 29th International Colloquium, {ICALP}
  2002, Malaga, Spain, July 8-13, 2002, Proceedings}, volume 2380 of {\em
  Lecture Notes in Computer Science}, pages 704--715. Springer, 2002.
\newblock \href {https://doi.org/10.1007/3-540-45465-9_60}
  {\path{doi:10.1007/3-540-45465-9_60}}.

\bibitem{DBLP:journals/fuin/Chaloupka13}
Jakub Chaloupka.
\newblock Z-reachability problem for games on 2-dimensional vector addition
  systems with states is in {P}.
\newblock {\em Fundam. Informaticae}, 123(1):15--42, 2013.
\newblock \href {https://doi.org/10.3233/FI-2013-798}
  {\path{doi:10.3233/FI-2013-798}}.

\bibitem{ChatterjeeDHR10}
Krishnendu Chatterjee, Laurent Doyen, Thomas~A. Henzinger, and
  Jean{-}Fran{\c{c}}ois Raskin.
\newblock Generalized mean-payoff and energy games.
\newblock In Kamal Lodaya and Meena Mahajan, editors, {\em {IARCS} Annual
  Conference on Foundations of Software Technology and Theoretical Computer
  Science, {FSTTCS} 2010, December 15-18, 2010, Chennai, India}, volume~8 of
  {\em LIPIcs}, pages 505--516. Schloss Dagstuhl - Leibniz-Zentrum f{\"{u}}r
  Informatik, 2010.
\newblock \href {https://doi.org/10.4230/LIPICS.FSTTCS.2010.505}
  {\path{doi:10.4230/LIPICS.FSTTCS.2010.505}}.

\bibitem{Secondary-ChatterjeeRR12}
Krishnendu Chatterjee, Mickael Randour, and Jean{-}Fran{\c{c}}ois Raskin.
\newblock Strategy synthesis for multi-dimensional quantitative objectives.
\newblock In Maciej Koutny and Irek Ulidowski, editors, {\em {CONCUR} 2012 -
  Concurrency Theory - 23rd International Conference, {CONCUR} 2012, Newcastle
  upon Tyne, UK, September 4-7, 2012. Proceedings}, volume 7454 of {\em Lecture
  Notes in Computer Science}, pages 115--131. Springer, 2012.
\newblock \href {https://doi.org/10.1007/978-3-642-32940-1_10}
  {\path{doi:10.1007/978-3-642-32940-1_10}}.

\bibitem{DBLP:conf/lics/ColcombetJLS17}
Thomas Colcombet, Marcin Jurdzinski, Ranko Lazic, and Sylvain Schmitz.
\newblock Perfect half space games.
\newblock In {\em 32nd Annual {ACM/IEEE} Symposium on Logic in Computer
  Science, {LICS} 2017, Reykjavik, Iceland, June 20-23, 2017}, pages 1--11.
  {IEEE} Computer Society, 2017.
\newblock \href {https://doi.org/10.1109/LICS.2017.8005105}
  {\path{doi:10.1109/LICS.2017.8005105}}.

\bibitem{DBLP:conf/mfcs/CourtoisS14}
Jean{-}Baptiste Courtois and Sylvain Schmitz.
\newblock Alternating vector addition systems with states.
\newblock In Erzs{\'{e}}bet Csuhaj{-}Varj{\'{u}}, Martin Dietzfelbinger, and
  Zolt{\'{a}}n {\'{E}}sik, editors, {\em Mathematical Foundations of Computer
  Science 2014 - 39th International Symposium, {MFCS} 2014, Budapest, Hungary,
  August 25-29, 2014. Proceedings, Part {I}}, volume 8634 of {\em Lecture Notes
  in Computer Science}, pages 220--231. Springer, 2014.
\newblock \href {https://doi.org/10.1007/978-3-662-44522-8_19}
  {\path{doi:10.1007/978-3-662-44522-8_19}}.

\bibitem{DavisJanuszkiewicz2000}
Michael~W. Davis and Tadeusz Januszkiewicz.
\newblock Right-angled artin groups are commensurable with right-angled coxeter
  groups.
\newblock {\em Journal of Pure and Applied Algebra}, 153(3):229--235, 2000.
\newblock \href {https://doi.org/10.1016/S0022-4049(99)00175-9}
  {\path{doi:10.1016/S0022-4049(99)00175-9}}.

\bibitem{DBLP:conf/lics/DOsualdoMZ16}
Emanuele D'Osualdo, Roland Meyer, and Georg Zetzsche.
\newblock First-order logic with reachability for infinite-state systems.
\newblock In Martin Grohe, Eric Koskinen, and Natarajan Shankar, editors, {\em
  Proceedings of the 31st Annual {ACM/IEEE} Symposium on Logic in Computer
  Science, {LICS} '16, New York, NY, USA, July 5-8, 2016}, pages 457--466.
  {ACM}, 2016.
\newblock \href {https://doi.org/10.1145/2933575.2934552}
  {\path{doi:10.1145/2933575.2934552}}.

\bibitem{esparza1994decidability}
Javier Esparza.
\newblock On the decidability of model checking for several $\mu$-calculi and
  {Petri} nets.
\newblock In {\em Colloquium on Trees in Algebra and Programming}, pages
  115--129. Springer, 1994.
\newblock \href {https://doi.org/10.1007/BFb0017477}
  {\path{doi:10.1007/BFb0017477}}.

\bibitem{esparza2024decidability}
Javier Esparza and Mogens Nielsen.
\newblock Decidability issues for petri nets -- a survey, 2024.
\newblock \href {http://arxiv.org/abs/2411.01592} {\path{arXiv:2411.01592}},
  \href {https://doi.org/10.48550/arXiv.2411.01592}
  {\path{doi:10.48550/arXiv.2411.01592}}.

\bibitem{GamesOnGraphs}
Nathanaël Fijalkow, C.~Aiswarya, Guy Avni, Nathalie Bertrand, Patricia Bouyer,
  Romain Brenguier, Arnaud Carayol, Antonio Casares, John Fearnley, Paul
  Gastin, Hugo Gimbert, Thomas~A. Henzinger, Florian Horn, Rasmus Ibsen-Jensen,
  Nicolas Markey, Benjamin Monmege, Petr Novotný, Pierre Ohlmann, Mickael
  Randour, Ocan Sankur, Sylvain Schmitz, Olivier Serre, Mateusz Skomra,
  Nathalie Sznajder, and Pierre Vandenhove.
\newblock Games on graphs: From logic and automata to algorithms, 2025.
\newblock \href {http://arxiv.org/abs/2305.10546} {\path{arXiv:2305.10546}},
  \href {https://doi.org/10.48550/arXiv.2305.10546}
  {\path{doi:10.48550/arXiv.2305.10546}}.

\bibitem{DBLP:conf/lics/GanardiMZ22}
Moses Ganardi, Rupak Majumdar, and Georg Zetzsche.
\newblock The complexity of bidirected reachability in valence systems.
\newblock In Christel Baier and Dana Fisman, editors, {\em {LICS} '22: 37th
  Annual {ACM/IEEE} Symposium on Logic in Computer Science, Haifa, Israel,
  August 2 - 5, 2022}, pages 26:1--26:15. {ACM}, 2022.
\newblock \href {https://doi.org/10.1145/3531130.3533345}
  {\path{doi:10.1145/3531130.3533345}}.

\bibitem{DBLP:journals/corr/abs-2504-05015}
Roland Guttenberg, Eren Keskin, and Roland Meyer.
\newblock {PVASS} reachability is decidable.
\newblock {\em CoRR}, abs/2504.05015, 2025.
\newblock \href {http://arxiv.org/abs/2504.05015} {\path{arXiv:2504.05015}},
  \href {https://doi.org/10.48550/ARXIV.2504.05015}
  {\path{doi:10.48550/ARXIV.2504.05015}}.

\bibitem{DBLP:conf/rp/HaaseH14}
Christoph Haase and Simon Halfon.
\newblock Integer vector addition systems with states.
\newblock In Jo{\"{e}}l Ouaknine, Igor Potapov, and James Worrell, editors,
  {\em Reachability Problems - 8th International Workshop, {RP} 2014, Oxford,
  UK, September 22-24, 2014. Proceedings}, volume 8762 of {\em Lecture Notes in
  Computer Science}, pages 112--124. Springer, 2014.
\newblock \href {https://doi.org/10.1007/978-3-319-11439-2_9}
  {\path{doi:10.1007/978-3-319-11439-2_9}}.

\bibitem{HsuWise1999}
Tim Hsu and Daniel~T Wise.
\newblock On linear and residual properties of graph products.
\newblock {\em Michigan Mathematical Journal}, 46(2):251--259, 1999.
\newblock \href {https://doi.org/10.1307/mmj/1030132408}
  {\path{doi:10.1307/mmj/1030132408}}.

\bibitem{JurdzinskiLS15}
Marcin Jurdzinski, Ranko Lazic, and Sylvain Schmitz.
\newblock Fixed-dimensional energy games are in pseudo-polynomial time.
\newblock In Magn{\'{u}}s~M. Halld{\'{o}}rsson, Kazuo Iwama, Naoki Kobayashi,
  and Bettina Speckmann, editors, {\em Automata, Languages, and Programming -
  42nd International Colloquium, {ICALP} 2015, Kyoto, Japan, July 6-10, 2015,
  Proceedings, Part {II}}, volume 9135 of {\em Lecture Notes in Computer
  Science}, pages 260--272. Springer, 2015.
\newblock \href {https://doi.org/10.1007/978-3-662-47666-6_21}
  {\path{doi:10.1007/978-3-662-47666-6_21}}.

\bibitem{DBLP:conf/icalp/LerouxST15}
J{\'{e}}r{\^{o}}me Leroux, Gr{\'{e}}goire Sutre, and Patrick Totzke.
\newblock On the coverability problem for pushdown vector addition systems in
  one dimension.
\newblock In Magn{\'{u}}s~M. Halld{\'{o}}rsson, Kazuo Iwama, Naoki Kobayashi,
  and Bettina Speckmann, editors, {\em Automata, Languages, and Programming -
  42nd International Colloquium, {ICALP} 2015, Kyoto, Japan, July 6-10, 2015,
  Proceedings, Part {II}}, volume 9135 of {\em Lecture Notes in Computer
  Science}, pages 324--336. Springer, 2015.
\newblock \href {https://doi.org/10.1007/978-3-662-47666-6_26}
  {\path{doi:10.1007/978-3-662-47666-6_26}}.

\bibitem{LiptonZalcstein1977}
Richard~J. Lipton and Yechezkel Zalcstein.
\newblock Word problems solvable in logspace.
\newblock {\em Journal of the ACM}, 24(3):522--526, July 1977.
\newblock \href {https://doi.org/10.1145/322017.322031}
  {\path{doi:10.1145/322017.322031}}.

\bibitem{lohrey2008submonoid}
Markus Lohrey and Benjamin Steinberg.
\newblock The submonoid and rational subset membership problems for graph
  groups.
\newblock {\em Journal of Algebra}, 320(2):728--755, 2008.
\newblock \href {https://doi.org/10.1016/j.jalgebra.2007.08.025}
  {\path{doi:10.1016/j.jalgebra.2007.08.025}}.

\bibitem{LyndonSchupp1977}
Roger~C. Lyndon and Paul~E. Schupp.
\newblock {\em Combinatorial Group Theory}.
\newblock Springer Berlin, Heidelberg, 1977.
\newblock \href {https://doi.org/10.1007/978-3-642-61896-3}
  {\path{doi:10.1007/978-3-642-61896-3}}.

\bibitem{mandel2026complexitydownwardclosuresindexed}
Richard Mandel, Corto Mascle, and Georg Zetzsche.
\newblock The complexity of downward closures of indexed languages, 2026.
\newblock \href {http://arxiv.org/abs/2601.19466} {\path{arXiv:2601.19466}},
  \href {https://doi.org/10.48550/arXiv.2601.19466}
  {\path{doi:10.48550/arXiv.2601.19466}}.

\bibitem{Borel}
Donald~A. Martin.
\newblock Borel determinacy.
\newblock {\em Annals of Mathematics}, 102(2):363--371, 1975.
\newblock \href {https://doi.org/10.2307/1971035} {\path{doi:10.2307/1971035}}.

\bibitem{DBLP:conf/concur/MeyerMZ18}
Roland Meyer, Sebastian Muskalla, and Georg Zetzsche.
\newblock Bounded context switching for valence systems.
\newblock In Sven Schewe and Lijun Zhang, editors, {\em 29th International
  Conference on Concurrency Theory, {CONCUR} 2018, Beijing, China, September
  4-7, 2018}, volume 118 of {\em LIPIcs}, pages 12:1--12:18. Schloss Dagstuhl -
  Leibniz-Zentrum f{\"{u}}r Informatik, 2018.
\newblock \href {https://doi.org/10.4230/LIPICS.CONCUR.2018.12}
  {\path{doi:10.4230/LIPICS.CONCUR.2018.12}}.

\bibitem{Muskalla23}
Sebastian Muskalla.
\newblock {\em Certificates for automata in a hostile environment}.
\newblock PhD thesis, 2023.
\newblock \href {https://doi.org/10.24355/dbbs.084-202401201643-0}
  {\path{doi:10.24355/dbbs.084-202401201643-0}}.

\bibitem{DBLP:journals/entcs/RaskinSB05}
Jean{-}Fran{\c{c}}ois Raskin, Mathias Samuelides, and Laurent~Van Begin.
\newblock Games for counting abstractions.
\newblock In Michael Huth, editor, {\em Proceedings of the Fouth International
  Workshop on Automated Verification of Critical Systems, AVoCS 2004, London,
  UK, September 4, 2004}, volume 128 of {\em Electronic Notes in Theoretical
  Computer Science}, pages 69--85. Elsevier, 2004.
\newblock \href {https://doi.org/10.1016/J.ENTCS.2005.04.005}
  {\path{doi:10.1016/J.ENTCS.2005.04.005}}.

\bibitem{DBLP:conf/concur/ShettyKZ21}
Aneesh~K. Shetty, S.~Krishna, and Georg Zetzsche.
\newblock Scope-bounded reachability in valence systems.
\newblock In Serge Haddad and Daniele Varacca, editors, {\em 32nd International
  Conference on Concurrency Theory, {CONCUR} 2021, Virtual Conference, August
  24-27, 2021}, volume 203 of {\em LIPIcs}, pages 29:1--29:19. Schloss Dagstuhl
  - Leibniz-Zentrum f{\"{u}}r Informatik, 2021.
\newblock \href {https://doi.org/10.4230/LIPICS.CONCUR.2021.29}
  {\path{doi:10.4230/LIPICS.CONCUR.2021.29}}.

\bibitem{Walukiewicz01}
Igor Walukiewicz.
\newblock Pushdown processes: Games and model-checking.
\newblock {\em Inf. Comput.}, 164(2):234--263, 2001.
\newblock \href {https://doi.org/10.1006/INCO.2000.2894}
  {\path{doi:10.1006/INCO.2000.2894}}.

\bibitem{DBLP:conf/icalp/Zetzsche13}
Georg Zetzsche.
\newblock Silent transitions in automata with storage.
\newblock In Fedor~V. Fomin, Rusins Freivalds, Marta~Z. Kwiatkowska, and David
  Peleg, editors, {\em Automata, Languages, and Programming - 40th
  International Colloquium, {ICALP} 2013, Riga, Latvia, July 8-12, 2013,
  Proceedings, Part {II}}, volume 7966 of {\em Lecture Notes in Computer
  Science}, pages 434--445. Springer, 2013.
\newblock \href {https://doi.org/10.1007/978-3-642-39212-2_39}
  {\path{doi:10.1007/978-3-642-39212-2_39}}.

\bibitem{DBLP:conf/stacs/Zetzsche15}
Georg Zetzsche.
\newblock Computing downward closures for stacked counter automata.
\newblock In Ernst~W. Mayr and Nicolas Ollinger, editors, {\em 32nd
  International Symposium on Theoretical Aspects of Computer Science, {STACS}
  2015, Garching, Germany, March 4-7, 2015}, volume~30 of {\em LIPIcs}, pages
  743--756. Schloss Dagstuhl - Leibniz-Zentrum f{\"{u}}r Informatik, 2015.
\newblock \href {https://doi.org/10.4230/LIPICS.STACS.2015.743}
  {\path{doi:10.4230/LIPICS.STACS.2015.743}}.

\bibitem{Zetzsche2016c}
Georg Zetzsche.
\newblock {\em Monoids as Storage Mechanisms}.
\newblock {PhD} thesis, Technische Universit{\"a}t Kaiserslautern, 2016.
\newblock URL:
  \url{https://kluedo.ub.rptu.de/frontdoor/index/index/docId/4400}.

\bibitem{DBLP:journals/iandc/Zetzsche21}
Georg Zetzsche.
\newblock The emptiness problem for valence automata over graph monoids.
\newblock {\em Inf. Comput.}, 277:104583, 2021.
\newblock \href {https://doi.org/10.1016/J.IC.2020.104583}
  {\path{doi:10.1016/J.IC.2020.104583}}.

\bibitem{DBLP:conf/rp/Zetzsche21}
Georg Zetzsche.
\newblock Recent advances on reachability problems for valence systems (invited
  talk).
\newblock In Paul~C. Bell, Patrick Totzke, and Igor Potapov, editors, {\em
  Reachability Problems - 15th International Conference, {RP} 2021, Liverpool,
  UK, October 25-27, 2021, Proceedings}, volume 13035 of {\em Lecture Notes in
  Computer Science}, pages 52--65. Springer, 2021.
\newblock \href {https://doi.org/10.1007/978-3-030-89716-1_4}
  {\path{doi:10.1007/978-3-030-89716-1_4}}.

\end{thebibliography}

\newpage

\onlyFull{
\appendix
\section{Additional material for Section~\ref{preliminaries}}\label{app:preliminaries}
\subsection{About dead ends in game arenas}\label{app:dead-ends}
One could also define valence game arenas where dead ends---i.e.\ states with no outgoing transitions--- are allowed. In this case, one would define a play that ends in a dead end as winning for the universal player. It is not difficult to see that for \nrgs, this more general setting reduces to the dead-end-free setting for every graph $\Gamma$: 
\begin{enumerate}
\item If $\Gamma$ has an unlooped vertex $a$, we can just add to every dead-end
$q$ a self-loop $q\xrightarrow{\bar{a}} q$. Then, reaching a dead-end means the
rest of the play consists of this transition, which will eventually reach a
non-right-invertible element.
\item If $\Gamma$ has only looped vertices, then $\Mon_\Gamma$ is a group and
every element of $\Mon_\Gamma$ is right-invertible. In this case, the game is
won by the existential player if and only if they have a strategy to avoid
dead-ends. It is thus a simple safety game, which can be decided in polynomial
time.
\end{enumerate}

\section{Additional material for Section~\ref{main-results}}\label{app:main-results}

\subsection{Non-termination games}
Here, we give a brief definition of non-termination games. These are games played on configuration graphs of infinite-state, where the objective of the existential player is to play infinitely long. 

The purpose of introducing these games is to highlight their decidability
properties in contrast to \nrgs. Therefore, crucially, the arena only consists
of valid configurations of the underlying infinite-state systems. In valence
systems, valid configurations are those that are right-invertible. Therefore, a
\emph{play (in a non-termination game)} is a (finite or infinite) run $(q_1,x_1)\to
(q_2,x_2)\to\cdots$ such that $x_i\in\RInv{\Mon_\Gamma}$ for every $i$.

Strategies are defined as in \nrgs, but notice that since runs can only visit
$\RInv{\Mon_\Gamma}$, it can happen that a play gets stuck (and is thus maximal
and finite), even if every control-state has an outgoing transition: For
example, applying a decrement in a VASS may not be possible without making the
counter negative. In this case, the play \emph{is won by $\forall$}, since the
objective of $\exists$ is to play infinitely long.

\subsection{Non-termination games for $\Pushdown(\ZVASS)$}\label{app:infinite-runs-pdtwoz}
Let us now show undecidability of non-termination games for the graph in
\cref{fig:pdtwoz}, which belongs to $\Pushdown(\Group)$, and specifically to
$\Pushdown(\twoz{0.75})\subseteq\Pushdown(\ZVASS)$. Thus, valence systems over
the graph in \cref{fig:pdtwoz} are pushdown systems where the stack contains
pushdown letters and between them a vector in $\Z^2$.

As in the undecidability proof for non-termination games on VASS configuration graphs~\cite[Thm.~13.1]{GamesOnGraphs}, we reduce from the non-termination problem
for deterministic two-counter machines. To simplify the construction, we assume that the two-counter machine's counters range over $\Z$: Since it has zero-tests available, that still leaves its non-termination game problem undecidable. We take the two-counter machine and
\begin{figure}
\pdtwoz{1}
\caption{Graph in $\Pushdown(\Group)$ with undecidable non-termination game problem}\label{fig:pdtwoz}
\end{figure}
turn operations on the counters $i=0,1$ into labels $b_i$ (to increment counter
$i$) and $\bar{b}_i$ (to decrement counter $i$). However, before starting the simulation, we first apply $a$, in a fresh initial state. All these states belong to $\exists$-player.

\begin{figure}
  \tikzset{every state/.style={minimum size=50pt}}
  \tikzset{every loop/.style={looseness=2}}
    \begin{tikzpicture}[
      auto,
      node distance=1.5cm,
      semithick
      ]
\node[draw, rectangle] (zi) {$z_i$};
\node[draw, rectangle, above right=1cm of zi] (tp) {$+$};
\node[draw, rectangle, below right=1cm of zi] (tn) {$-$};
\node[draw, rectangle, right=2cm of zi] (di) {$d_i$};
\path[->] (zi) edge node[above] {$\bar{b}_i$} (tp);
\path[->] (zi) edge node[below] {$b_i$} (tn);
\path[->] (tp) edge[loop above] node {$\bar{b}_i$, $b_{1-i}$, $\bar{b}_{1-i}$} (tp);
\path[->] (tn) edge[loop below] node {$b_i$, $b_{1-i}$, $\bar{b}_{1-i}$} (tn);
\path[->] (tp) edge node[above] {$\bar{a}$} (di);
\path[->] (tn) edge node[below] {$\bar{a}$} (di);
\node[below=0.5cm of zi] (dummybelow) {};
\node[above=0.5cm of zi] (dummyabove) {};
\path[->] (dummybelow) edge node[above] {} (zi);
\path[->] (zi) edge node[above] {} (dummyabove);
\path[->] (di) edge[loop right] node {$\bar{a}$} (di);
\end{tikzpicture}
\caption{Gadget for zero tests in non-termination games in $\Pushdown(\ZVASS)$.}\label{fig:undecidability-gadget-pdtwoz}
\end{figure}
The only non-obvious task is to simulate zero tests. In order to zero-test counter $i$, we use the gadget in \cref{fig:undecidability-gadget-pdtwoz}: If there is a zero-test edge in the two-counter machine from state $p$ to $q$, we use a copy of the gadget in \cref{fig:undecidability-gadget-pdtwoz} and add an edge $p\to z_i$ and $z_i\to q$ (indicated by the two arrows without source resp. target in \cref{fig:undecidability-gadget-pdtwoz}). Note that all states in the gadget belong to $\forall$-player.

This has the following effect. If counter $i$ is zero-tested, then $\forall$
can challenge this, either by claiming counter $i$ is in fact positive (by
going to state $+$) or by claiming it is in fact negative (by going to $-$). If
$\forall$ goes to $+$ and counter $i$ had indeed been positive, then $\forall$
can go to a configuration where both counters are zero (hence the entire
storage configuration is $a$), and then use the $\bar{a}$ to move to $d_i$.
Since then, the $\bar{a}$ loop at $d_i$ is not present (since it would be
invalid to pop another $a$), the run gets stuck in $d_i$ and $\forall$ wins.
Similarly, if the counter $i$ had in fact been negative, $\forall$ can force a
stuck game by going to $d_i$ via $-$.

However, if $\forall$'s challenge was incorrect (i.e. $\forall$ claimed ``positive'' and counter $i$ was $\le 0$; or it guessed ``negative'' and counter $i$ was $\ge 0$), 
then going to $+$ or $-$ will leave the game in that respective state: Counter
$i$ will then always be negative (in $+$) or always be positive (in $-$), so
that it is never possible to move to $d_i$. Thus, the run is stays in $+$ or
$-$, never gets stuck and is thus infinite.

Therefore, $\exists$ can enfore an infinite run if and only if the (unique) run
of the two-counter machines is infinite.

\begin{remark}
Note that the construction does not work for \nrgs: There, the transitions
$+\xrightarrow{\bar{a}} d_i$ and $-\xrightarrow{\bar{a}} d_i$ would
always be there if the counters are not both zero, and they would mean
victory for $\forall$. Thus, $\exists$ would not be able to enforce an
infinite run after performing a correct zero test.

In other words, this undecidability proof exploits the fact that the
transitions $+\xrightarrow{\bar{a}} d_i$ and $-\xrightarrow{\bar{a}}
d_i$ are not present in the arena after a correct zero test; their
absence makes $\exists$ win.
\end{remark}

\begin{remark}
	One might wonder why there is a self-loop $d_i\xrightarrow{\bar{a}}
	d_i$: after all, without it, any configuration in $d_i$ would also be
	a dead end. However, we wanted to show that undecidability even holds
	if in the game, every state has an outgoing transition (like we assume for
	\nrgs).
\end{remark}

\section{Additional Material for Section~\ref{proof-overview}}\label{sec:app-proof-overview}

\subsubsection*{Equivalence of non-termination games on Pushdown graphs (\cite{Walukiewicz01,Cachat02}) and \nrgs ~on $\Pushdown$ graphs}
A pushdown game graph is a tuple $(Q^\texttt{P} = Q^\texttt{P}_\exists \sqcup Q^\texttt{P}_\forall, \Sigma^\texttt{P}, \delta^\texttt{P})$ where $Q^\texttt{P}$ is a finite set of states, where $Q^\texttt{P}_\exists$ and $Q^\texttt{P}_\forall$ are existential player and universal player states, as expected. $\Sigma^\texttt{P}$ is a finite stack alphabet and $\delta^{\texttt{P}} \subseteq Q^\texttt{P} \times \Sigma^\texttt{P} \times Q^\texttt{P} \times (\Sigma^\texttt{P})^*$ is a finite transition system. A configuration $(q, w) \in Q^\texttt{P} \times (\Sigma^\texttt{P})^*$ consists of a state and a stack content, and 
$(p, v\gamma) \to (q, v w)$ for $(p, \gamma, q, w) \in \delta^\texttt{P}$. In other words, the transitions that give rise to non-right-invertible configurations, which are allowed in \nrgs, are disallowed in these games.
While defining the winning condition in pushdown games, one typically (\cite{Cachat02, Walukiewicz01}) lets the existential (universal) player to win every game that ends in a deadlock in a universal (existential) player state.
For non-termination games, the existential player wins every infinite game.

Non-termination games on a Pushdown graphs are seen to be equivalent to \nrgs~ on $\Pushdown$ graphs ($\Gamma = (\Sigma^\texttt{P}, \emptyset)$) via the following reductions. To reduce a \nrg~ to a non-termination game: Whenever there is a pop transition $p \xrightarrow{\bar{\gamma}} q$ defined on a universal state $p$ in a \nrg~ on $\Pushdown$ graph, turn it into two transitions $p \xrightarrow{\varepsilon} p' \xrightarrow{\bar{\gamma}} q$ by adding a new existential state $p'$. This shifts the burden of the non-right-invertible transition from the universal player to the existential player, and the play is lost by the existential player in both games if the state $p$ is reached with a stack content that does not end with $\gamma$.

To reduce a non-termination game to a \nrg : Do not alter the pop transitions on existential states, as the existential player cannot take them in the non-termination game (and loses if there is no other transition to take), and she avoids them in the \nrg~ (and loses if there is no other transition to take). Whenever there is a pop transition $p \xrightarrow{\bar{\gamma}} q$ from a universal state $p$ in the non-termination game, replace it with a gadget in the reduced \nrg~ to ensure that the universal player loses in the \nrg~ if he takes this edge with a configuation that does not end with a $\gamma$. 
The gadget consists of a fresh existential state $p'$ and the edges edge $p \xrightarrow{\varepsilon} p' \xrightarrow{\bar{\gamma}} q$. Additionally, there are edges $p' \xrightarrow{\bar{\gamma}_i} \Smiley$ for every other pushdown letter $\gamma_i \in  \Sigma^\texttt{P} \setminus \{\gamma\}$ and $\Smiley$ is a $\exists$-winning existential state with an $\varepsilon$ self-loop. Intuitively, in the reduced \nrg~, the universal player will take the edge $p \xrightarrow{\varepsilon} p'$ only if the top-of-the-stack is $\gamma$. Otherwise, the existential player pops the top of the stack symbol to reach $\Smiley$ and wins the game.

\section{Additional Material for Section~\ref{sec:FICEG}}\label{sec:app-FICEG}

In this section, we provide provide an $\EXPTIME$ decision algorithm for $\FICEG(\Gamma)$ where $\Gamma\in \Pushdown(\Group)$, proving~\cref{main-pushdown} for $\FICEG$. 
A summary of the procedure is given in~\cref{sec:FICEG}. Here we will provide it in a more detailed and self-contained fashion.

The following is the main theorem of this section.

\begin{theorem} Let $\Gamma \in \Pushdown(\Group) \times \Group  $. Then $\FICEG(\Gamma)$ can be decided in $\EXPTIME$.
\end{theorem}

We remark that we will be using the notations $\FICEG(\gr)$ and $\UICEG(\gr)$ throughout the appendix.

 \begin{remark}
    By slight abuse of notation, we denote the instance $\gr$ of the decision problems $\FICEG$ and $\UICEG$ as $\FICEG(\gr)$ and $\UICEG(\gr)$, respectively. That is, $\exists$-player wins $\FICEG(\gr)$ for an initial credit $\lambda$ if and only if she wins $(q_{init}, [\lambda])$ in the game $(\gr, \RIO)$ where $q_{init}$ is the initial state of $\gr$. Similarly, $\exists$-player wins $\UICEG(\gr)$ if and only if there exists an initial credit $\lambda$ for which she wins $\FICEG(\gr)$.
 \end{remark}

  Before we start presenting the $\EXPTIME$ decision procedure, we will need some preliminaries. 

\subsection{Preliminiaries}
We start by showing that we can without loss of generality assume that the pushdown alphabet of $\Gamma = (V, E)$ has a single letter.

\lemSingleUnloopedNode*

\begin{proof}[\cref{single-unlooped-node}] Let
    $V = \{\gamma_1, \ldots, \gamma_m\} \cup W$ where $\gamma_i$ are independent vertices (i.e. the pushdown alphabet) and $(W, E)$ is a non-empty group graph. Let $g \in W$. We show that for any valence game arena $\gr$ over $\Gamma$, there is an equivalent valence game arena $\gr'$ over $\Gamma' = (\{\gamma\} \cup W, E)$ that is obtained by the following translation.  

    Every push transition $p \xrightarrow{\gamma_i} q$ in $\gr$ is turned into  $p \xrightarrow{\gamma^i g \gamma^i} q$ in $\gr'$, and similarly, every pop transition $p \xrightarrow{\bar{\gamma_i}} q$ is turned into $p \xrightarrow{{\bar{\gamma}^i}{\bar{g}}{\bar{\gamma}}^i} q$. The translation is clearly polynomial and the equivalence is easy to verify.
\end{proof}

From now on we take $V = \{\gamma\} \cup W$ where $\gamma$ is an indepent vertex, and $\Delta = (W, E)$ is a group graph. 

We have seen in~\cref{sec:FICEG} that the configurations of $\Pushdown(\Group)$ can bew written as $[w_{r+1} \gamma w_r \cdots \gamma w_1]$ and viewed as a stack, where $w_i \in \Mon_{\Delta}$.
Referring to this perspective, we will sometimes call the configurations, \emph{stack content} and $\gamma$, the \emph{pushdown letter}. Whenever it creates no confusion, we will further call the pushes ($\gamma$) and pops ($\bar{\gamma}$) of $\gamma$, simply \emph{pushes} and \emph{pops}. We will also refer to the removal of a suffix $\texttt{suf} = \gamma w_i \cdots \gamma w_1$ of the stack content as \emph{popping} $\texttt{suf}$, and denote it with $\bar{\texttt{suf}}$.

\begin{observation} \label{obs: push-on-universal} For every game arena $\gr$, there is an equivalent polynomially larger game arena $\gr'$ where the push ($\gamma$) transitions originate solely from universal player states. 
\end{observation}
To obtain $\gr'$ it is sufficient to add an additional universal state $r$ for each push transition $(p, \gamma, q)$ with $p \in Q_\exists$ in $\gr$, and replace $(p, \gamma, q)$ with the transitions $(p, \varepsilon, r)$ and $(r, \gamma, q)$ in $\gr'$.

From now on, we will assume the push transitions originate from universal player nodes.

\subparagraph*{Notation on trees} We call a finite tree $t$ a \emph{finite prefix subtree} of a tree $\tau$ if $t$ can be obtained from $\tau$ by removing a suffix of every branch. So every branch of $t$ is a prefix of a branch in $\tau$.

By the \emph{height of a finite tree}, we refer to the number of nodes on its longest branch. 

\vspace{2mm}
Next we introduce strategy trees.
An $\iota$-strategy tree is a tree-shaped representation of a $\iota$-strategy that resolves a $\iota$-player's choices locally at every history.

\subparagraph*{Strategy trees}
Let a game arena $\gr = (\Gamma, Q = Q_\exists \sqcup Q_\forall, \delta, q_{init})$. For $\iota\in \{\exists, \forall\}$, an $\iota$-strategy tree over $\gr$ is a labelled tree $\tau$ rooted at a node $(\texttt{rt}, [u]) \in Q \times \Mon_\Gamma$ such that:
    \begin{itemize} 
        \item The nodes of $\tau$ are configurations in $Q \times \Mon_\Gamma$, \label{it:strtree-one}
        \item The edges are labelled with $V \cup \bar{V} \cup \{\varepsilon\}$, \label{it:strtree-two}
        \item If a node $(q, [u])$ has $q \in Q_{1-\iota}$, then for every transition $q\xrightarrow{v} q'$ in $\delta$, the node $(q, [u])$ has a child $(q', [uv])$ and the edge is labelled with $v$,\label{it:strtree-three}
        \item If a node $(q, [u])$ has $q \in Q_{\iota}$, then the node has a single child $(q', [uv])$ where $q\xrightarrow{v} q'$ is in $\delta$ and the edge is labelled with $v$.\label{it:strtree-four}
    \end{itemize}

Intuitively, an $\iota$-strategy tree $\tau$ unfolds an $\iota$-strategy $\sigma_\iota$ into a tree that assigns a unique move to each $\iota$-node, while leaving the opponent's choices unconstrained. %
Clearly, every play consistent with $\sigma_\iota$ corresponds to a branch in the corresponding strategy tree. Thus, $\sigma_\iota$ is winning if and only if every branch of $\tau$ is an winning play for $\iota$ in the game. Such trees are called \emph{winning} $\iota$-strategy trees. 

We can represent winning $\forall$-strategy trees via finite trees. In a winning $\forall$-strategy, every branch visits a node $(q, [u])$ where $u \not \in \RInv{\Mon_\Gamma}$. We consider finite prefix subtrees of winning $\forall$-strategy where each branch has its first non-right-invertible configuration as a leaf. We call these trees (finite) winning $\forall$-strategy trees and constrain our attention these trees.
We call the label of the last edge of a finite branch, the \emph{leaf transition} of the branch.
Note that, all leaf transition in a (finite) winning $\forall$-strategy are pops ($\bar{\gamma}$).

Let $\tau$ be a $\forall$-strategy tree rooted at $(q, [u])$. We denote by $\tau|_{Q}$ the projection of $\tau$ to states $Q$. So $\tau|_{Q}$ is a tree rooted at $q$ that has states in $Q$ as nodes instead of configurations in $Q \times \Mon_\Gamma$.
Similarly, for a projection tree $\tau$ rooted at $q$ with nodes over $Q$, we denote by $\tau_{\gets (q, [u])}$ it's extension to configurations as nodes. So, for $\tau' = \tau|_{Q}$, the extension  $\tau'_{\gets (q, [u])}$ is exactly $\tau$.

 \begin{definition} The stack height of a (finite or infinite) branch of a tree is the maximum number of unmatched pushes\footnote{a push $\gamma$ not neutralized by a pop $\bar{\gamma}$.} a prefix of the branch contains. \end{definition}
 For example, the finite branch $$\texttt{rt} \xrightarrow{\gamma} p_1 \xrightarrow{g^-_1} p_2 \xrightarrow{\gamma} p_3 \xrightarrow{\bar{\gamma}} p_4 \xrightarrow{g_1} p_5 \xrightarrow{\bar{\gamma}} p_6 \xrightarrow{\gamma} p_7 \xrightarrow{\bar{\gamma}} p_8 \xrightarrow{g_2}$$ has stack height $2$, due to its prefix until $p_3$ having $2$ unmatched $\gamma$s. 
 
 We say a winning $\forall$-strategy (starting from some configuration $(q, [u])$) has stack height bounded by $i$ if in the (finite) winning $\forall$-strategy tree it corresponds to (rooted at $(q, [u])$), every branch has stack height $\leq i$.

\noindent \subparagraph{Valid and invalid pops} A pop ($\bar{\gamma}$) in a play (or, branch), is called \emph{invalid} if it is reached with a stack content $[\varepsilon]$ or $[w_{r+1} \gamma \cdots \gamma w_1]$ where $w_1 \ne 1$. All other pops are called \emph{valid}. Further, a valid pop is called $0$-valid, if it is reached with the stack content $[\gamma]$.

\noindent \subparagraph{Tree automata over a game arena}
We build a (top-down) tree automaton $\mathcal{A}_{\gr}= (Q \cup \{q_\emptyset\}, \texttt{rt}, \Lambda = \bigcup_{q \in Q} (\Lambda_q \cup \Lambda_q^+))$ which accepts all $\forall$-strategy trees (restricted to states) over the game arena $\gr$ rooted at $\texttt{rt} \in Q$ -- and, possibly more trees.
\begin{itemize}
\item We take $\texttt{rt}$ as the root state.
\item For each $q\in Q_\exists$, fix once and for all an enumeration $succ(q) = (p_1, \ldots, p_k)$ of its successors and an enumeration of transition labels $w_1, \ldots, w_k$ leading to the successors; i.e. $(q, w_i, p_i) \in \delta$ for all $i$. 
\item For $q \in Q_\exists$ with $succ(q) = (p_1, \ldots, p_k)$, we set a \emph{single} transition
\begin{equation}
    \Lambda_q(q) = ((w_1, p_1), \ldots, (w_k, p_k)), \quad \quad  \Lambda^+_q = \emptyset.
\end{equation}
    So, in state $q$, the automaton branches to \emph{all} successors of $q$. %
\item For $q \in Q_\forall$, we set $\Lambda_q(q) = \{(m_i, p_i) \mid (q, m_i, p_i) \in \delta \}$. So, in state $q$, the automaton can nondeterministically choose \emph{one} successor of $q$. This corresponds to the universal player's choice at a universal node, with a given history of the game. Furthermore, the automaton can nondeterministically choose the \emph{succesors} of $q$ from the set $\Lambda_q^+$ which sends $q$ to a subset of $Q$ with $\varepsilon$-transitions, i.e. $(\varepsilon, Q)^i$ where $i \leq |Q|$; or to $\bot$ with $\emptyset$-transition, i.e. $(\emptyset, q_\emptyset)$. This additional set $\Lambda_q^+$ intuitively allows the universal player to pick a subset of nodes of the game as the successor of $q$. Here, the distinguished state $q_\emptyset$ led with the $\emptyset$-transition represents the case where $q$ is mapped to the empty successor set; introducing a separate state and transition label for this case is a purely technical device that simplifies the subsequent construction. The state $q_\emptyset$ has no outgoing transitions. 
\item  We allow $\mathcal{A}_{\gr}$ to accept both finite and infinite trees, and assign all nodes in $Q \cup \{q_{\emptyset}\}$ as accepting leaf nodes for the finite case. So, $\mathcal{A}_{\gr}$ accepts all trees rooted at $\texttt{rt}$, finite or infinite, whose branching rules comply with $\Lambda$.
\end{itemize}

We can show a finite branch of a tree accepted by $\mathcal{A}_{\gr}$ as a sequence of transitions $\texttt{rt} \xrightarrow{w_1} q_1 \xrightarrow{w_2} \cdots \xrightarrow{w_l} q_l $. Here, we call $q_l$ the \emph{leaf node} and $w_l$ the \emph{leaf transition} of the branch. 

If $\Lambda_q^+ =\emptyset$, then $\mathcal{A}_{\gr}$ accepts exactly the $\forall$-strategy trees rooted at $\texttt{rt}$ and their subtrees. By adding more transitions to $\Lambda_q^+$, one allows \emph{short-cuts} from universal player nodes to certain subset of nodes, without any change in the stack content. 

Whenever convenient, we give the choice of the root node as an input to $\mathcal{A}_\gr$, i.e. $\mathcal{A}_\gr(q)$ is $\mathcal{A}_\gr$ with $\texttt{rt} = q$.

\subsection{The idea behind the saturation}

Our aim in a this decision algorithm, is to figure out for each $q \in Q$, whether $\forall$-player has a winning strategy from $(q, \varepsilon)$. For this, we want to understand the set of $\forall$-strategy trees rooted at $(q, \varepsilon)$, and classify them as winning or losing. If there exists a winning $\forall$-strategy tree rooted at $(q, \varepsilon)$, we declare $(q, \varepsilon)$ to be winning for $\forall$, and if $(q_{init}, \varepsilon)$ is won by $\forall$, we declare the game $\FICEG(\gr)$ to be won by him. Otherwise, the game is won by $\exists$ due to determinacy.

   One of the main challanges in deciding whether a $\forall$-strategy tree is winning is that, in order for a pop in a branch to qualify $\forall$-player to win (an invalid pop) in principle, one needs to keep track of arbitrarily high stack content. This would require infinite computational power, and would not provide a decidable algorithm. To bypass this need, we will abstract the winning $\forall$-strategy trees by a fixpoint computation on a finite set. Each iteration of this fixpoint computation will take exponential time, and the fixpoint will be reached within exponentially many steps; providing an $\EXPTIME$ algorithm. 

   Here, we present these finite sets.

\subparagraph*{$P$- and $(P, i)$-trees} A $P$-tree is a finite prefix subtree of a $\forall$-strategy tree rooted at node $(q, [w])$ such that the leaf node in each branch is either (i) the first $(q, [u])$ for $u \in \RInv{\Mon_\Gamma}$ (as in the definition of a winning tree), or (ii) %
 $(q', [\varepsilon])$ for $q' \in P$ and the leaf transition is $\bar{\gamma}$. 

A $P$-tree is intuitively a ``relaxed version'' of a $\forall$-winning strategy tree as every branch is either won by $\forall$, or reaches at a node in $P$ with empty stack -- intuitively, can be won from $P$ in the future. Clearly, $\forall$ wins the configuration $(q, \gamma)$ in $\gr$, iff there exists a $\emptyset$-tree rooted at $(q, \gamma)$.

A $(P, i)$-tree is a $P$-tree in which the stack height of every branch is bounded by $i$.

Then, $\forall$ wins the configuration $(q, \gamma)$ in $\gr$ with a strategy of stack height bounded by $i$, iff there exists a $(\emptyset, i)$-tree rooted at $(q, \gamma)$.

\subparagraph*{Saturating along stack height} We will have a saturation procedure over the finite abstraction sets $R_q^i$: the set of all $P \subseteq Q$ such that there is a ``relaxed $\forall$-strategy tree of bounded stack height $i$'' that reaches a subset of $P$ with empty-stack. Formally,
 \begin{align*}R^i_q := &\{ P \subseteq Q \mid \text{there exists a $P, i$-tree } \tau \text{ rooted at }(q, \gamma).\}
\end{align*} 

Clearly, $\emptyset \in R_q^i$ if and only if the universal player has a winning strategy from $(q, \gamma)$ that has stack height bounded by $i$. 
We claim that if the universal player wins $(q, \gamma)$, then there is a large enough $i$. 

\begin{claim}\label{claim:exists-Ri} Universal player wins $(q, \gamma)$ if and only if there exists an $i$ s.t. $\emptyset \in R^i_q$.
\end{claim}
\begin{claimproof} Since the universal player wins $(q, \gamma)$, there exists a (finite) winning $\forall$-strategy tree $\tau$. 
Let $i$ be the maximum stack height achieved amongst the branches of $\tau^q$. Then $\emptyset \in R^i_q$. On the other hand, if $\emptyset \in R^i_q$ then the finite tree $\tau$ that witnesses $\emptyset$ is a winning $\forall$-strategy tree rooted at $(q, \gamma)$. \end{claimproof}

Clearly, $\{R^i_q\}_{i \in \mathbb{N}}$ is a non-decreasing sequence as $i$ tends to $\infty$. As the set is bounded, this sequence reaches a fixpoint. 
We denote this fixpoint by $R_q = \bigcup_{ i \in \mathbb{N}} R^i_q$. By~\cref{claim:exists-Ri}, we know that this fixpoint contains sufficient information to identify the nodes in $Q$ such that $(q, \gamma)$ is $\forall$-winning. That is, %
\begin{lemma}\label{lemma:Rq}
$\emptyset \in R_q \quad \text{if and only if} \quad \text{there exists a winning } \forall\text{-strategy from } (q,\gamma)$.
\end{lemma}

Thus, we aim to construct an algorithm to saturate $R^i_q$ to reach $R_q$.

\subsection{Calculating the saturation sets via tree automata (Part 1)}
Our idea is the obtain the sets $P \in R_q^i$ from the leaves of trees accepted by a tree automaton. 
In particular, for each $i$, and some special element $\top$, we want to have a tree automata $\mathcal{T}^i_q$ such that the following set calculates $R_q^i$.
$$\{P \subseteq Q \mid \text{ there exists a finite tree } \tau \text{ accepted by }\mathcal{T}^i_q \text{ s.t. if $p \times \top $ is a leaf of }\tau, \text{ then }p \in P \}$$

It is clear that, we want $\mathcal{T}^i$ to accept trees that are somehow equivalent to $(P, i)$-trees rooted at $(q, \gamma)$. 

 Rougly, the algorithm will proceed as follows: Recall that $\mathcal{A}_\gr$ with $\Lambda^+ = \emptyset$ accepts all $\forall$-strategy trees (and all finite prefices). We will let $\mathcal{A}^0_\gr = \mathcal{A}_\gr$ and use it to accept $(P, 0)$-trees rooted at $(q, \gamma)$. Then for each push transition $p \xrightarrow{\gamma} q$ in $\delta$ and every $P \in R_p^0$, we will add to $\mathcal{A}^0_\gr$ an $\varepsilon$ transition from $p$ to $P$ using $\Lambda^+$. We will therefore obtain $\mathcal{A}^1_\gr$ which will be used to $(P, 1)$-trees rooted at $(q, \gamma)$, and so on.

 As $(P, 0)$-trees have no pushes, we need to get rid of the trees accepted by $\mathcal{A}^0_\gr $ that contain pushes. Furthermore, 
 as $\mathcal{A}^0_\gr$ accepts trees with game states as nodes instead of configurations, we also need to make sure that the 
 leaves of the accepted trees indeed correspond to (i) the first non-right-invertible configuration of the branch or (ii) $(q', [\varepsilon])$ with the leaf transition $\bar{\gamma}$. 
 This is ensured by requiring that the leaf transitions of every branch is the first pop of the branch -- and checking if the pop is valid. 

 We use a small automaton called `push-pop' automaton, to ensure that the accepted trees do not contain pushes, and the leaf transitions are the first pops of the branch:
\subparagraph*{`Push-pop' automaton} Consider a simple automaton $\pp$ that accepts %
    sequences of group elements that end with a pop or an $\emptyset$-transition, i.e. $\{ w \bar{\gamma},  w \emptyset \mid w \in (\Mon_{\Delta})^* \} $. When taken product with $\mathcal{A}^0_{\gr}$ this automaton eliminates infinite trees, trees that contain $\gamma$, that do not contain a $\bar{\gamma}$ as a leaf transition, and trees that contain a $\bar{\gamma}$ anywhere else. With slight abuse of notation, we ignore the $\pp$ component while talking about $\mathcal{A}^0_{\gr} \times \pp$, treating $\pp$ only as an acceptor. That is, the language of $\mathcal{A}^0_{\gr} \times \pp$ is a subset of the language of $\mathcal{A}^0_{\gr}$.

The language of $\mathcal{A}^0_{\gr} \times \pp$ is the set of all finite prefix subtrees of $\forall$-strategy trees that contain \emph{no pushes}, and each branch has exactly one pop -- that is as the leaf transition. Still, we have no idea whether these pops at the leaf transitions are valid or invalid. To understand this, we need to keep track of group elements accumulated during the branches. Once more, doing this naively would require infinite computational power, and not be decidable. 

Next, we explain how we tackle this problem. 

\subsection{Main challange: Group elements}
Another main challenge in our algorithm is to perform the saturation in the presence of group elements.
Our crucial observation is that, we can classify the group elements as ``big” or ``small”. While we need to store the small elements, whenever we come across a big element, we can simply declare victory and forget about the group element.

For a group element $g \in \Mon_\Delta$, we denote by $\|g\|$ the \emph{word length} of $g$, i.e. the length $|w|$ of the shortest word $w \in (W \cup\bar{W})^*$ that satisfies $[g] = [w]$. Concretely, we call $g$ ``big'' if $\|g\| > |Q|$ and ``small'' otherwise.

Our crucial observation that we can ``declare victory'' whenever we come accross a ``big'' group element relies on~\cref{lemma:i-eligible}, that we will prove later.

\subparagraph*{Eligibility} A state $q \in Q$ is called $i$\emph{-eligible} if there exists a $(P,i)$-tree rooted at $(q, \gamma)$. 
Observe that if a state $q$ is not $i$-eligible, every $\forall$-strategy tree $\tau$ rooted at state $q$  has a branch that is either (i) bounded by stack height $i$ and winning for $\exists$-player (so, has no invalid pops), or (ii) exceeds stack height $i$. %

\begin{lemma} \label{lemma:i-eligible} Suppose $p \in Q$ is $i$-eligible and $g \in \Mon_\Delta$ has $\|g\|>|Q|$. Then $(p, w g)$ is the root of some $(\emptyset, i)$-tree for every $w \in \{[\varepsilon]\} \cup \{[u \gamma] \mid u \in (V \cup \bar{V})^*\}$. 
\end{lemma} 

The lemma implies that every configuration $(q, [u])$ where $q$ is $i$-eligible for some $i$, and the top-of-the-stack group element of $u$ is ``big'', is won by $\forall$. So, whenever we come accross such a configuration, we can declare victory!

As explained in~\cref{sec:FICEG}, we use \emph{Cayley automaton}, an automaton version of \emph{Cayley graph}, to track the ``small'' group elements, and  declare victory on ``big'' ones. Here, we introduce this automaton formally. 

\subparagraph*{Cayley automaton} \emph{Cayley automaton} $\mathcal{C}^{\leq |Q|} = (C, \Sigma^{\mathcal{C}}, \delta^{\mathcal{C}}, [\varepsilon])$ is defined over the state set $C$ and the alphabet $\Sigma^\mathcal{C} = W \cup \bar{W}$. The state set $C$ consists of the set of congruence classes $[g]$ s.t. $g \in \Mon_{\Delta}$ and $\|g\| \leq |Q|$ together with a trap state $c^{trap}$. The initial state 
    of the automaton is $[\varepsilon]$, the congruence class of $\varepsilon$.
    
    For every $u, u' \in (W \cup \bar{W})^*$ with $\|u\| \leq |Q|$ and every $w \in (W \cup \bar{W} \cup {\varepsilon})$ s.t. $[uw] = [u']$, if $\|u'\| \leq |Q|$, then there is a transition $[u] \xrightarrow{w} [u']$ in $\delta^{\mathcal{C}}$; otherwise, there is a transition $[u] \xrightarrow{w} c^{trap}$.

    In other words, the Cayley automata starts from the empty word, and keeps track of group elements of word length at most $ |Q|$.
     If the word length exceeds $ |Q|$, it enters a trap state. All states of $\mathcal{C}$ are accepting. For convenience, we make the Cayley automaton ignore the $\bar{\gamma}$ and $\emptyset$ transitions, i.e. $c \xrightarrow{\bar{\gamma}} c$ and $c \xrightarrow{\emptyset} c$ for every $c \in C$.
    
\subsection{Calculating the saturation sets via tree automata (Part 2)}
With the introduction of Cayley automaton, we are ready to calculate $\mathcal{T}^i_q$, which is the product automaton $\mathcal{A}^i_\gr(q) \times \pp \times \mathcal{C}^{\leq |Q|}$, and the special element $\top$ is $[\varepsilon] \in C$. 

So we claim~\cref{lemma:R-sets-equivalence} holds for the set $R(\mathcal{T}^i_q) $ defined below.
\begin{align*}R(\mathcal{T}^i_q) = \{P \subseteq Q \mid &\text{ there exists a finite tree } \tau \text{ accepted by }\,\mathcal{T}^i_q \,\,\text{ such that if } p \times [\varepsilon] \\ &\text{ is a leaf of }\tau, \text{ then }p \in P \}
\end{align*}

\begin{lemma} \label{lemma:R-sets-equivalence} $R(\mathcal{T}^i) = R_q^i$ for every $q \in Q$. 
\end{lemma}
\subsection{The Base Case}
We will first prove~\cref{lemma:i-eligible} and then~\cref{lemma:R-sets-equivalence} for $i = 0$.

\begin{lemma*}[\cref{lemma:i-eligible} for $i = 0$]\label{lemma:0-eligible}
    Suppose $p \in Q$ is $0$-eligible and $g \in \Mon_\Delta$ has $\|g\|>|Q|$. Then $(p, w g)$ is the root of some $(\emptyset, 0)$-tree for every $w \in \{[\varepsilon]\} \cup \{[u \gamma] \mid u \in (V \cup \bar{V})^*\}$.%
\end{lemma*}
\begin{proof}
    Let $\tau$ be a $(P,0)$-tree rooted at $(p, \gamma)$. Every leaf of $\tau$ is either the first non-right-invertible configuration $(q, [u])$ of the branch, or $(q', [\varepsilon])$ with $q' \in P$ led by a pop transition. As the stack height is bounded by $0$, the leaf transitions of both kind of leaves are the first pop transitions of their branches. Then the projection tree $\tau|_{Q}$ is accepted by $\mathcal{A}^0_\gr(q) \times \pp$. As $\mathcal{A}^0_\gr(p)$ has $|Q|$-many states (plus $q_\emptyset$, but that can only appear on the leaf nodes due to $\pp$), we can remove repeating occurrences of inner nodes in $\tau|_{Q}$ by cutting a branch at the first occurrence of a node $q$ and attaching the subtree rooted at the last occurrence of an inner node $q$. We thereby obtain a small witness tree $t \in \mathcal{A}^0_\gr(p) \times \pp$ where each branch has length $\leq |Q|+1$. Note that we get the length $|Q|+1$ and not $|Q|$ because we need to preserve the leaf transitions, as $\pp$ is accepting trees based on leaf transitions. 

    The extension tree $t' = t_{\gets (p, wg)}$ is a $(P',0)$-tree of height $|Q|+1$. As leaf transitions of $t'$ are pops, any branch of $t'$ can neutralise a group element of word length at most $|Q|$. As $\|g\| > |Q|$, for every leaf $(q, [u])$ of $t'$ it holds that $[u] \not \in \RInv{\Mon_{\Gamma}}$. Therefore, $t'$ is a $(\emptyset, 0)$-tree rooted at $(p, wg)$.

\end{proof}

\begin{lemma*}[\cref{lemma:R-sets-equivalence} for $i = 0$] $R(\mathcal{T}_q^0) = R_q^0$ for every $q \in Q$.
\end{lemma*}

The nodes of trees accepted by $\mathcal{T}_q^0 = \mathcal{A}^0_q \times \pp \times \mathcal{C}^{\leq |Q|}$ are of the form $q \times c \in (Q \cup \{q_\emptyset\}) \times C $. Recall that $\mathcal{A}^0_\gr(q) \times \pp $ accepts the set of all finite prefix subtrees of $\forall$-strategy trees rooted at $q$, that contain \emph{no pushes}, and each branch has exactly one pop -- that is as the leaf transition. Then the leaf transitions $p \times c \xrightarrow{\bar{\gamma}} r \times c$ of trees accepted by $\mathcal{T}_q^0 $ are also pops, the Cayley component $c$ does not change at the lead transition, as the Cayley automaton skips pops. As the inner transitions (i.e. not leaf transitions) of the tree are labelled with $g \in (W \cup \bar{W} \cup \{\varepsilon\})$, (i) $c = c^{trap}$ if there exists a (topmost) ancestor node $(q', [u'])$ with $\|u'\| > |Q|$, (ii) $c = [u] \ne [\varepsilon]$ if the word length has not exceeded $|Q|$ and the leaf transition is an invalid pop, and (iii) $c = [\varepsilon]$ if the leaf transition is a valid pop. By~\cref{lemma:i-eligible} for the base case, we know that there is a $(\emptyset, 0)$-tree $t$ rooted at $(q', [u'])$. For every branch (i), we replace the subtree rooted at $(q', [u'])$ with $t$. We thereby obtain a tree $\tau$ in $\mathcal{T}_q^0$ where every leaf is of form (ii) or (iii). Clearly, $\tau$ extends to a $(P, 0)$-tree rooted at $(q, \gamma)$ such that if $p \times [\varepsilon]$ is a leaf of $\tau$, then $p \in P$. This shows $R(\mathcal{T}_q^0) \subseteq R_q^0$. 

The other direction is similar. Let $P \in R_q^0$. Then, there is a $(P, 0)$-tree $\tau$ rooted at $(q, \gamma)$. The leaves of $\tau$ are either (I) the first non-right-invertible configuration of the branch, or (II) a leaf node $(q',[\emptyset])$ for $q' \in P$ led by a pop transition. As the stack height is bounded by $0$, all leaf transitions of $\tau$ are pop transitions -- and furthermore, they are the first pops of their branches. (I) corresponds to invalid pops and (II) corresponds to valid pops. Then $\tau|_{Q}$ is accepted by $\mathcal{T}^i_q$ and the leaves corresponding to valid pops are of the form $q' \times [\varepsilon]$ and the leaves corresponding to invalid pops are of the form $q' \times c^{trap}$ or $q' \times [u] \neq [\varepsilon]$. Then clearly, $P \in R(\mathcal{T}_q^0) $.

\subsection{The Inductive Step}

We describe how to obtain $\mathcal{T}^{i+1}_q$ from $\mathcal{T}^{i}_q$ for every $q \in Q$. This boils down to obtaining
$\mathcal{A}_\gr^{i+1}$ from $\mathcal{A}_\gr^{i}$. 

Here is how we obtain $\mathcal{A}_\gr^{i+1}$ from $\mathcal{A}_\gr^{i}$: For every push transition $(p, \gamma, r) \in \delta$ (recall that, by our assumption $p \in Q_\forall$), add the following new transitions to $\mathcal{A}^{i}_{\gr} $ to obtain $\mathcal{A}^{i+1}_{\gr}$: 
For every $\emptyset \neq R \in R^{i}_r$, add to $\Lambda^+$ a new $\varepsilon$-transition from $p$ to $R$, i.e. add $\Lambda^+_p(p) = (\varepsilon, r_1) \times \ldots \times (\varepsilon, r_m)$ where $R = \{r_1, \ldots, r_m\}$.
This transition allows the trees accepted by the tree automaton to branch out at $p$ to $r_1, \ldots, r_m$ via $\varepsilon$-edges. Intuitively, if there exists a strategy $\tau \in \mathcal{A}^{i}_{\gr}(r)$ that allows $(r, \gamma)$ to reach a set $R' \subseteq R$ via 0-valid pops, then we allow $p$ to arrive at all nodes in $R$ without changing its stack content. This transition corresponds to the universal player strategy from $p$ of taking the edge $p \xrightarrow{\gamma} r$ and then following the $\forall$-strategy $\tau$ of stack height bounded by $i$ until he reaches the nodes in $R$. 
If $\emptyset \in R^{i}_r$, then there is a winning $\forall$-strategy from $r$ stack height bounded by $i$. In this case, add the transition $\Lambda^+_p(p) = (\emptyset, q_\emptyset)$ to $\Lambda$. This will allow $\pp$ to recognise this state as ``already won'' by $\forall$-player and allow a branch to end with the leaf transition $r \xrightarrow{\emptyset} q_\emptyset$.

We assign $\mathcal{T}^{i+1}_q = \mathcal{A}^{i+1}_\gr(q)\times \pp \times\mathcal{C}^{\leq |Q|}$ for every $q \in Q$. 

With the inductive hypothesis that~\cref{lemma:i-eligible,lemma:R-sets-equivalence} holds for all $k \leq i$, we prove them for $i+1$. 

\begin{lemma*}[{\protect\cref{lemma:i-eligible} for $i+1$}]
    Suppose $p \in Q$ is $(i+1)$-eligible and $g \in \Mon_\Delta$ has $\|g\|>|Q|$. Then $(p, w g)$ is the root of some $(\emptyset, i+1)$-tree for every $w \in \{[\varepsilon]\} \cup \{[u \gamma] \mid u \in (V \cup \bar{V})^*\}$.
\end{lemma*}

\begin{proof}
    Let $\tau$ be a $(P, i+1)$-tree rooted at $(p, \gamma)$. Every leaf of $\tau$ is either the first non-right-invertible configuration $(q, [u])$ of the branch (i.e. an invalid pop), or $(q', [\varepsilon])$ with $q' \in P$ led with a pop transition (i.e. a 0-valid pop). Once more, all leaf transitions of $\tau$ are pops. However this time, they do not need to be the first pops of their branches. Due to the bounded stack height of $i+1$, the branches are allowed to have $(i+1)$ push and matching pops before arriving at an invalid, or a 0-valid pop. 
    Consider a subtree $\tau'$ rooted at the node reached via the first push of a branch of $\tau$. So, $\tau'$ is rooted at some $(r', [u\gamma])$ with the parent $(r, [u])$. Then consider its finite prefix subtree that is a $(P', i)$ tree rooted at $(r', \gamma)$. That is, each leaf is either an invalid pop reached within a pop depth of $i$, or a $0$-valid pop from $(r', \gamma)$, reaching a node in $P'\subseteq Q$. Due to the I.H., the set of nodes $P'$ the nodes in $\tau'$ reached via 0-valid pops, is in $R^i_q$. Thus, there is an $\varepsilon$-transition from $r$ to $P'$ in $\mathcal{A}^{i+1}_\gr$. If $P' = \emptyset$, the transition is $r \xrightarrow{\emptyset} q_\emptyset$. 
    
    We go through the branches of $\tau$ from top to bottom, and repeat the following procedure until there are no pushes left. Find the first $(r, [u])$ leading to some $(r', [u\gamma])$. Replace the subtree rooted at $(r, [u])$ and ends at the leaves of the $(P',i)$-subtree rooted at $(r', [\gamma])$, with the transition rooted at $(r, [u])$ and has children $(p_1, [u]), \ldots, (p_\ell, [u])$ for $P' = \{p_1, \ldots, p_l\}$, where each edge is labelled with $\varepsilon$; if $P'=\emptyset$, replace it with the transition $(r, [u]) \xrightarrow{\emptyset} (q_{\emptyset}, [u])$.

    The resulting tree $\tilde{\tau}$ has no pushes. Furthermore, all leaf transitions are either (i) $\emptyset$-transitions, or (ii) the first pop of the branch. (i) is the case if the branch has an invalid pop with stack height $< i+1$. (ii) is the case if either the leaf is a node $(q', [\varepsilon])$ (a 0-valid pop), or the leaf transition is an invalid-pop and the branch has stack height exactly $i+1$.
    In any case, ${\tilde{\tau}}|_{Q}$ is accepted by $\mathcal{A}^{i+1}_\gr(q) \times \pp$ as, $\tilde{\tau}$ has no pushes, every leaf transition is either $\emptyset$-transition, or the first pop of the branch; and the transitions we used to replace subtrees of $\tau$ exist in $\mathcal{A}^{i+1}_\gr$. 

    Now we are ready to replicate the argument from the proof of the lemma for the base case. As the automaton $\mathcal{A}^{i+1}_\gr$ has $|Q|$-many states (used as inner nodes of a tree), and there is a tree $\tilde{\tau}|_{Q}$ rooted at $p$ that accepted by $\mathcal{A}^{i+1}_\gr(q) \times \pp$, there is a small witness tree $t$ rooted at $p$  of height $\leq |Q|+1$ that is also accepted by $\mathcal{A}^{i+1}_\gr(q) \times \pp$. Applying the $\forall$-strategy $t$ from root configuration $(p, wg)$ gives an $(\emptyset, i+1)$-tree, as before. 
\end{proof}

\cref{lemma:i-eligible} for $i+1$ shows that in the product $\mathcal{T}^{i+1}_q  = \mathcal{A}^{i+1}_\gr \times \pp \times \mathcal{C}^{\leq |Q|}$ the nodes with Cayley component $c^\trap$ are indeed nodes from which $\forall$ has a winning strategy. So, the Cayley component works correctly in distinguishing ``small'' and ``big'' elements at iteration $i+1$. 

\begin{lemma*}[\cref{lemma:R-sets-equivalence} for $i+1$] $R(\mathcal{T}_q^{i+1}) = R_q^{i+1}$ for every $q \in Q$.
\end{lemma*}

\begin{proof}%
Let us first show $R(\mathcal{T}^{i+1}_q) \subseteq R^{i+1}_q$.
Let $\tau$ be a tree accepted by $\mathcal{T}^{i+1}_q$ and let $P$ be the set of leaves $p \times [\varepsilon]$ of $\tau$ s.t. $p \in P$. We will show that $P \in R^{i+1}_q$. By I.H., every $\varepsilon$-transition from $p$ to $R$ in $\mathcal{A}^{i+1}_\gr$ represents a push transition $p \xrightarrow{\gamma} r \in \delta$ and a $(R, i)$-subtree rooted at $(r, \gamma)$. Similarly, a $\emptyset$-transition represents a $(\emptyset, i)$-subtree rooted at $(r, \gamma)$.
Replacing the $\emptyset$- and $\varepsilon$-transitions in $\tau$ with the respective subtrees, we obtain a $(P', i+1)$-tree rooted at $(q, \gamma)$. Now we need to show $P'= P$. But this follows from the facts that (i) $\varepsilon$-transitions indeed represent subtrees where the leaves are reached from the root with the same stack content, (ii) by \cref{lemma:i-eligible} that the Cayley component distinguishes ``big'' configurations as intended in iteration $i+1$.  

Now let us show $R(\mathcal{T}^{i+1}_q) \supseteq R^{i+1}_q$. Let $P \in R^{i+1}_q$ and let $\tau$ be the $(P, i+1)$-tree rooted at $(q, \gamma)$. The leaves of $\tau$ are either (I) the first non-right-invertible configuration of the branch (so, an invalid pop), or (II) a leaf node $(q',[\varepsilon])$ for $q' \in P$ led by a pop transition (so, a 0-valid pop).
Similar to the proof of~\cref{lemma:i-eligible} for $i+1$, we go over the branches of $\tau$ from top to bottom repeatedly, and for each push $(r, [u]) \xrightarrow{\gamma} (r, [u\gamma])$ replace the subtree starting at $(r, [u])$ and ending at the leaves of the
$(P, i)$-subtree rooted at $(r', \gamma)$, with $\varepsilon$-transitions from $(r, [u])$ to $(p_i, [u])$ for every $p_i \in P$. Replace it with $(r, [u]) \xrightarrow{\emptyset} (q_\emptyset, [u])$ is $P = \emptyset$. These transitions clearly exist in $\mathcal{A}^{i+1}_\gr$.
The resulting tree $\tilde{\tau}$ has no pushes, and as it is obtained from a $(P, i+1)$-tree $\tau$, each leaf transition of  $\tilde{\tau}$ is either $\emptyset$, or the first pop of the branch. As the $\varepsilon$- and $\emptyset$-transitions used to replace subtrees all exist in $\mathcal{A}^{i+1}_\gr$, $\tilde{\tau}|_{Q} \in \mathcal{A}^{i+1}_\gr(q) \times \pp$. 

Let $\tilde{\tau}|_Q \times C$ be $\tilde{\tau}|_Q$ extended with the Cayley components in $\mathcal{T}^{i+1}_q$. Lastly, we need to ensure that the set $P$ of nodes reached via valid pop transitions in $\tilde{\tau}$ are exactly the set of nodes $p \times [\varepsilon]$ in $\tilde{\tau}|_Q \times C$. 
The fact that $\emptyset$-transitions correspond to invalid pops with stack height comes $\leq i$ from the I.H. So we only need to ensure that the Cayley automaton distinguishes the ``big'' configurations correctly in iteration $i+1$. However, this comes from the proof of~\cref{lemma:i-eligible} for $i+1$, as explained above. This concludes the proof.

\end{proof}\subsection{Complexity and Correctness}
The saturation algorithm saturates the tree automaton $\mathcal{A}^0_\gr$ on $|Q|+1$ states by adding transitions. As the number of transitions that can be added is exponential, the saturation takes at most exponentially many steps. At each step, 
a calculation is performed on the exponential sized automaton $\mathcal{A}^0_\gr \times \pp \times \mathcal{C}^{\leq |Q|}$. As explained in Section~\ref{sec:FICEG}, each step requires at most exponential time, as we can find the minimal representative of a word over $\Delta$ in $\EXPTIME$. So, the overall complexity of the saturation procedure is $\EXPTIME$. 

At the end of the saturation, we obtain the set $R_q$ for each $q \in Q$. We know by~\cref{lemma:Rq} that 
$\emptyset \not \in R_{q_i}$ if and only if $(q_{init}, \gamma)$ is won by $\exists$-player. To compute the winner of $(q_{init}, \varepsilon)$, it is sufficient to adjust the game arena slightly by adding an initial pop ($\bar{\gamma}$) transition preceeding every outgoing transition of $q_{init}$ (WLOG assuming $q_{init}$ has no self-loops). Then $(q_{init}, \gamma)$ is won by $\exists$ in the adjusted game arena if and only if $(q_{init}, \varepsilon)$ is won by $\exists$ in the original game arena.

\section{Additional Material for Section~\ref{sec:UICEG}}\label{sec:app-UICEG}
In this section, we provide the proofs missing in~\cref{sec:UICEG}.

\initialCreditUb*

We will prove~\cref{thm:initial-credit-ub} in two parts. In~\cref{thm:initial-credit-length} we will show that the $\Pushdown$-length of $\lambda$ can be bounded by $2^n$ and in~\cref{thm:initial-credit-structure} we will show that the word length of a group element inbetween two pushes of $\lambda$ can be bounded by $n$. 

\begin{restatable}{theorem}{initialCreditLength}\label{thm:initial-credit-length} If $\exists$ wins $(q_{init}, [\lambda'])$ for some initial credit $\lambda'$, she also wins
   $(q_{init}, [\lambda])$ for some initial credit $\lambda$ where $|\lambda|_{\gamma} < 2^n$.
\end{restatable}

\begin{proof}[Proof of \cref{thm:initial-credit-length}]
   We show that whenever $\exists$ wins $\gr$ with the initial credit $\lambda'$ with $|\lambda'|_{\gamma} > 2^n$, we can find some $\lambda$ with $|\lambda|_{\gamma} < |\lambda'|_{\gamma}$ such that she wins the game with initial credit $\lambda$ too. Take $\lambda' = [w_{r+1} \gamma w_r \cdots \gamma w_1]$ where $ r > 2^n$ and fix a winning $\exists$-strategy tree $\tau$ rooted at $(q_{init}, \lambda')$. 
 We first make some observations: 
 \begin{enumerate}
    \item In every branch of $\tau$, a maximal suffix $\gamma w_{s}\cdots \gamma w_1$  of $\lambda'$ is popped. We call $s$ the \emph{popping depth} of this branch. %
    \item If a branch of $\tau$ reaches $(q, [u])$ after popping a suffix $\gamma w_{d}\cdots \gamma w_1$ of $\lambda'$, then\\ $(q, [w_{r+1}  \gamma w_{r} \cdots \gamma  w_{d+1}] )$ is won by $\exists$.  \label{it:two}
    \item Every branch of $\tau$ eventually reaches some $(q, [u])$ s.t. $(q, [\varepsilon])$ is won by $\exists$.  \label{it:three}
    \item If $(q, [u])$ is won by $\exists$ then so is $(q, [vu])$ for any $v,u \in (V \cup \bar{V})^*$. \label{it:four}
 \end{enumerate}

    We represent the branches of $\tau$ by the following finite representation: 
    $$ b^i := q^i_0 \xrightarrow{\bar{\gamma}} q^i_1 \xrightarrow{\bar{\gamma}} \ldots \xrightarrow{\bar{\gamma}} q^i_{s^i}$$ where for all $i$, $q^i_0 = q_{init}$, and for $j \in [1, s^i]$, each edge $q^i_{j-1} \xrightarrow{\bar{\gamma}} q^i_{j}$ represents consecutive sections in the branch where $q^i_{j-1} \xrightarrow{[\overline{w_{j-1}}]}^* p \xrightarrow{\bar{\gamma}} q^i_{j}$. Furthermore, $q^i_{s^i}$ is the first state in the branch where $(q^i_{s^i}, [\varepsilon])$ is won by $\exists$ (which exists by Obs.~\ref{it:three}). 

     As edges represent consecutive pops in $\lambda'$, for all $i$, $s^i \leq r$. We further restrict the branches to their first $2^n$-many pops ($\bar{\gamma}$ ) of $\lambda'$, whereas for $i$ where the popping depth $s^i < 2^n$, we pad the branch with $\varepsilon$ edges on a state $p$ visited infinitely often in the rest of the branch where $(p, [\varepsilon])$ is won by $\exists$; until the length is exactly $2^n$ (such a $p$ exists due to Obs.~\ref{it:three} and pigeonhole principle).  
      We will denote this restriction of each branch $b^i$ by $\mathbf{b}^i$, that is,
    \begin{align} \mathbf{b}^i := \begin{cases} q^i_0 \xrightarrow{\bar{\gamma}} q^i_1 \xrightarrow{\bar{\gamma}} \ldots \xrightarrow{\bar{\gamma}} q^i_{2^n} &\text{ if }s^i \geq 2^n \text{ and }\\
    q^i_0 \xrightarrow{\bar{\gamma}} q^i_1 \xrightarrow{\bar{\gamma}} \ldots \xrightarrow{\bar{\gamma}} q^i_{s^i} \xrightarrow{\varepsilon} q^i_{s^i+1} = p \xrightarrow{\varepsilon} \ldots \xrightarrow{\varepsilon} q^i_{2^n} = p &\text{ if }s^i < 2^n %
    \end{cases}
    \end{align}

    Note that in the case where $s^i < 2^n$, both $(q^{s^i}, [\varepsilon])$ and $(p, [\varepsilon])$ are won by $\exists$. The $\varepsilon$-transitions inbetween them do not imply that $p$ can be reached from $q^{s^i}$ with the same stack content; they simply denote that fact that we do not pop the stack any further to win from these nodes. 

    We will denote the set of all $\mathbf{b}^i$ by $B$. Note that $B$ is a finite set. We claim the following.

  \textbf{Claim:} There exist two indices $j < \ell \in [1, 2^n]$ such that for every $k \leq |B|$ there exists a $k' \leq |B|$ such that $q^k_j = q^{k'}_{\ell}$. 
  
  We first argue that if the claim holds, then~\cref{thm:initial-credit-length} is proven. 
  
  Let $\lambda'|_{< i} = \gamma w_{i-1}  \cdots \gamma w_1$ and $\lambda'|_{\geq i} = w_{r+1} \gamma w_r \cdots \gamma w_i$. Then the claim implies that $(q_{init}, \lambda)$ is won by $\exists$, for
  $\lambda = [\lambda'|_{\geq \ell} \cdot \lambda'|_{< j}]$. In particular, it implies that for every branch of $\tau$, there exists a state $p$ and a suffix $\texttt{suf} = \gamma w_{d}  \cdots \gamma w_1$ of $\lambda'|_{< j}$ s.t. $q_{init} \xrightarrow{\overline{\texttt{suf}}}{\phantom{.}^*} \,{p}$ \footnote{Here, the edge relation is restricted to the chosen branch.} and $(p, \lambda'|_{\geq{\ell}} \cdot \texttt{pre})$ is won by $\exists$. Here $\texttt{pre}$ is the prefix of $\lambda'|_{< j}$ that satisfies $\texttt{pre}\cdot \texttt{suf} = \lambda'|_{< j}$. 

   In particular, for the branch $\mathbf{b}^k$, if $s^k \geq j$, then $p$ is $q^k_j$, and $\texttt{suf}=\lambda'|_{< j}$ (therefore, $\texttt{pre} = \varepsilon$); otherwise, $p$ is $q^k_{s^k}$ and $\texttt{suf} = \gamma w_{s^k-1}  \cdots \gamma w_1$ since the branch $\mathbf{b}^k$ only pops this suffix of $\lambda'$. In the latter case, $(p = q^k_{s^k}, [\varepsilon])$ is won by $\exists$, then so is $(v, [\lambda'|_{\geq{\ell}} \cdot \texttt{pre}])$ by Obs.~\ref{it:four}. In the former case, branch $\mathbf{b}^{k'}$ will witness that $(p, [\lambda'|_{\geq{\ell}}]) = (p, [\lambda'|_{\geq{\ell}} \cdot \texttt{pre}])$ is won by $\exists$: If $s^{k'} \geq \ell$, then $\delta^{k'}$ reaches $p = q^k_j = q^{k'}_{\ell}$ by popping $\lambda|_{<\ell}$; thus, $(p, [\lambda'|_{\geq{\ell}}])$ is won by $\exists$ by Obs.~\ref{it:two}. Otherwise, if $s^{k'} < \ell$, then $(p = q^k_j = q^{k'}_{\ell}, [\varepsilon])$ is won by $\exists$, then so is $(p, [\lambda'|_{\geq{\ell}}])$ by Obs.~\ref{it:four}. This shows that every branch of $\tau$ reaches a state $p$ via popping a suffix $\texttt{suf}$ of $\lambda'|_{< j}$ such that $(p, [\lambda'|_{\geq{\ell}} \cdot \texttt{pre}])$ is won by $\exists$. Therefore, every branch of $\tau$ wins $q_{init}$ with the initial credit $[\lambda'|_{\geq{\ell}} \cdot \texttt{pre}\cdot \texttt{suf}] = \lambda'|_{\geq{\ell}} \cdot \lambda'|_{< j} = \lambda$. 
   So, $(q_{init}, \lambda)$ is won by $\exists$, which proves~\cref{thm:initial-credit-length} as $|\lambda|_{\gamma} < |\lambda'|_{\gamma}$ due to $j < \ell$.

  Now we prove the claim. 
  For $i \in [1, 2^n]$, let $S^i$ be the set of all states $q \in Q$ that appear as $q_i$ for some branch, i.e. $S^i : = \{ q^j_i \mid j \in [1, |B|]\}$. Clearly $ 1 \leq |S^i| \leq |Q|$ for all $i$.
  Then $S^i$ is a non-empty subset of $Q$, therefore can get $2^n - 1$ many different values. It follows by pigeonhole principle that there are two different indices $j < \ell \in [1, 2^n]$ such that $S^j = S^{\ell}$. Clearly, the statement of the claim holds over $j$ and $\ell$, which concludes the proof of the claim.

\end{proof}

In~\cref{thm:initial-credit-length}, we have given an upper bound on the number of pushdown letters (or, pushes) in an initial credit. Next, we give an upper bound on the word length of the group elements between two pushdown letters.

\begin{restatable}{theorem}{initialCreditStructure}\label{thm:initial-credit-structure} 
   If $\exists$ wins $(q_{init}, [\lambda \gamma w \gamma  \lambda'])$ where $w$ is a group element with $\|w\| > n$, she also wins 
   $(q_{init}, [\gamma \lambda'])$.
\end{restatable}

\noindent \textit{Proof Sketch of~\cref{thm:initial-credit-structure}}. 
 Let $\tau$ be a winning $\exists$-strategy tree rooted at $(q_{init}, [\lambda  \gamma  w \gamma  \lambda'])$. Every branch of $\tau$ reaches some $(p, [\lambda  \gamma  w])$, which is won by $\exists$ (as all nodes of $\tau$ are). We claim that $(p, [\varepsilon])$ is also won by $\exists$ because otherwise, there would be a winning $\forall$-strategy rooted at $(p, [\varepsilon])$. Then there would also one of height $\leq |Q|+1$ where every leaf transition is a $\bar{\gamma}$ (recall from the small witness trees in the proof of~\cref{lemma:i-eligible}). This small winning $\forall$-strategy tree would allow $\forall$ to win $(p, [\lambda  \gamma  w])$ as non of its branches is long enough to pop $w$ before reaching a pop $\bar{\gamma}$.
 By this contradiction, we conclude that indeed, for all such $p$, $(p, [\varepsilon])$ is won by $\exists$. It then follows that $(q_{init}, [\gamma  \lambda'])$ is also won by $\exists$. Hence, $\exists$ wins $(q_{init}, [\gamma \lambda'])$. \qed

\begin{proof}[Proof of \cref{thm:initial-credit-structure}]
 Let $\tau$ be a winning $\exists$-strategy tree rooted at  $(q_{init}, [\lambda \gamma w \gamma \lambda'])$. Every branch of $\tau$ reaches some $(p, [\lambda  \gamma  w])$, which is won by $\exists$ (as all nodes of $\tau$ are). We claim that $(p, [\varepsilon])$ is also won by $\exists$. The claim implies $(q_{init}, [\gamma \lambda'])$ is won by $\exists$ -- i.e. $\exists$ wins $(q_{init}, [\gamma \lambda'])$
 -- and ends the proof. 

 Assume to the contrary that some $(p, [\varepsilon])$ is won by $\forall$. Then, there exists a winning $\forall$-strategy tree rooted at $(p, [\varepsilon])$. Let $\tau^p$ be this (finite) $\forall$-strategy tree, where every branch stops at their first invalid pop $\bar{\gamma}$. 
 Let the stack height of $\tau^p$ be bounded by $i$ (such an $i$ exists by~\cref{claim:exists-Ri}). Then $\tau^p$ is a $(\emptyset, i)$-tree rooted at $(p, [\varepsilon])$. In particular, note that all leaf transitions of $\tau^p$ are (invalid) pop. 
 We obtain a tree $\tilde{\tau}^p$ from $\tau^p$ by iteratively replacing the subtrees starting with pushes by transitions in $\mathcal{A}^{i}_{\gr}$ as in the proofs of~\cref{lemma:i-eligible,lemma:R-sets-equivalence} for $i+1$. Then  $\tilde{\tau}^p$ has no pushes left, and have pops only at the leaf transitions, i.e. $\tilde{\tau}^p|_Q \in \mathcal{A}^{i}_{\gr} \times \pp$. As the tree automaton has $|Q|$-many states, there exists a small witness tree $t$ rooted at $p$ with height $\leq |Q|+1$ accepted by $\mathcal{A}^{i}_{\gr} \times \pp$. Using the $\forall$-strategy represented by $t$ rooted at $(p, [\lambda \gamma w])$ allows $\forall$ to win, as the branches of $t$ are not long enough to consume the group element $w$ before the pop at the leaf transition -- rendering all pops invalid. This gives us the required contradiction, and concludes the proof. 
 
\end{proof}

Now we bring everything together for the proof of \cref{thm:initial-credit-ub}.

\begin{proof}[Proof of \cref{thm:initial-credit-ub}]
For an initial credit $\lambda \in (V \cup \bar{V})^*$ \cref{thm:initial-credit-length} proves that $|\lambda|_{\gamma} < 2^n$ and \cref{thm:initial-credit-structure} proves that $\|\lambda\| \leq n$. 
\end{proof}

\subsection{Reducing $\UICEG(\gr)$ to $\FICEG(\gr)$}%
In this section, we will give a polynomial reduction from $\UICEG(\gr)$ to $\FICEG(\gr)$. Given a $\Pushdown(\Group)$ game arena $\gr$, we construct a polynomially sized $\Pushdown(\Group)$ game arena $\gr^\expp$ that allows $\exists$-player to pick any initial credit $\lambda$ that respects the structure given in~\cref{thm:initial-credit-ub}, and then enter the game $\gr$. It follows from~\cref{thm:initial-credit-ub} that $\FICEG(\gr)$ and $\UICEG(\gr)$ are won by the same player. As $\gr^\expp$ is polynomially sized, this gives an $\EXPTIME$ decision procedure for $\UICEG(\gr)$.

In order to construct the game $\gr^\expp$ for every $\gr$, we need a subgame that allows the existential player to collect any initial credit $\lambda$ s.t. $|\lambda|_{\gamma} < 2^n$ and $\|\lambda\|\leq n$. The difficulty in constructing such a subgame is, if we do it naively and add a loop on an $\exists$-controlled subgame with group elements or $\gamma$, the existential player takes this loop forever to win. So, we need a machinery that allows the universal player to punish the existential player if $\lambda$ violates the structure -- in particular, if $|\lambda|_{\gamma} \geq 2^n$. 

In order to count up to $2^n$ with a polynomially sized subgame, we will adopt a technique from DFAs. 

\DFAs*
\begin{proof}[Proof of \cref{lemma:DFAs}]
We explicity define the DFA set $(D_i)_{i \in [1, m]}$. For the sake of simplicity in explaining the acceptance condition, we will picture imaginary $\bot$ symbols at the start and the end of every word. Then, verbally $D_1$ accepts all words over $\Sigma$ that have at most one $\gamma_1$ inbetween each $\bot$ (so, at most one $\gamma_1$ in total). $D_2$ accepts all words that have at most one $\gamma_2$ inbetween each $\gamma \in \{\gamma_1, \bot\}$. $D_3$ accepts all words that have at most one $\gamma_3$ inbetween each $\gamma \in \{\gamma_1, \gamma_2, \bot\}$, and so on. 
 
 $D_i$ is given in~\cref{fig:DFA-i}. The accepted language of $D_i$ is the set of all words over $\Sigma$ that visit at most one $\gamma_i$ inbetween each pair of letters from the set $\{\gamma_1, \ldots, \gamma_{i-1}, \bot\}$. 
Let $w$ be a word of maximal length in $\bigcap_{i \in [1,m]} L(D_i)$. Then, $w$ contains one $\gamma_1$. The number of $\gamma_2$s it can contain is one less than the number of $\{\bot, \gamma_1\}$ it contains, as it has one $\gamma_2$ inbetween each character from this set. In general, for each $i \in [1, m]$, $w$ contains $2^{i-1}$-many $\gamma_i$s. So the length of $w$ is $\Sigma_{i \in [1, m]}2^{i-1} = 2^{m} -1$, ignoring the $\bot$ symbols.
It is easy to see that $w$ is the unique word of this length.

Lastly, let us observe that $L$ is suffix-closed. Let $w'' \cdot w' = w \in L$. We want to show that $w' \in L$. As we will be investigating the membership of $w$ and $w'$ in $L$, we will again imagine $\bot$ symbols at the start and the end of these words. Assume $w' \not \in L$. Then there exists a $D_i$, $w'$ is not in the language of. Then, $w'$ has at least two $\gamma_i$ inbetween some characters from the set $\{\gamma_1, \ldots, \gamma_{i-1}, \bot\}$. Let $w' = \bot u \gamma_i u' \gamma_i u'' \bot$. 
If $u$ contains a $\gamma_j$ with $j \in [1, i-1]$, then $w$ is also not in the language of $D_i$. If $u$ does not contain any such $\gamma_j$, $w = \bot \cdot w'' \cdot w'$ and either $w''$ contains such a $\gamma_j$, or the $\bot$ in the beginning serves as this violating character. In any case, $w$ cannot be in the language of $D_i$; which causes a contradiction.

\end{proof}

\begin{restatable}{theorem}{Gexp}\label{thm:Gexp}
There exists a polynomial sized $\Pushdown(\Group)$ game arena $\gr^{\expp}$ such that $\exists$-player wins $\UICEG(\gr)$ if and only if she wins $\FICEG(\gr^\expp)$.
\end{restatable}

\begin{proof}[Proof of \cref{thm:Gexp}]
Recall that $\gr = (\Gamma, Q = Q_{\exists}\sqcup Q_\forall, \delta, q_{init})$ where $\Gamma \in \Pushdown(\Group)$ where the pushdown alphabet is $\{\gamma\}$ and the group alphabet is $W$. Then formally, $\gr^\expp = (\Gamma^\expp, Q^{\expp}  = Q^{\expp}_{\exists} \sqcup Q^{\expp}_\forall, \delta^{\expp}, p_{init})$ where $\Gamma^\expp \in \Pushdown (\Group)$  where the pushdown alphabet is $\{\gamma_0, \ldots, \gamma_n\}$ and the group alphabet is $W$, $Q^\expp_\exists = Q_\exists \cup \{p_{init}\} \cup \{p_j, p'_j, d^i_j, \tilde{d}^i_j \mid j \in [0, n], i \in [1, n]\} \cup \{d_i^1, d_i^2 \mid i \in [1, n]\} \cup \{\Smiley, \Sadey\}$ and $Q^\expp_{\forall} = Q_\forall \cup \{\texttt{ch}\}$. 

The game is given in~\cref{fig:G-expp}. The initial state $p_{init}$ leads to $p_0$ by pushing a new pushdown letter, $\gamma_0$, which is not pushed anywhere else in the game, and is popped by the existential player in the subgames $\texttt{D}_i$ to check the emptiness of the stack (in~\cref{fig:DFA-i} (b)).
Intuitively, at the $p_0 - p_n$ loop, the existential player tries to accumulate an initial credit that allows her to win the game $\UICEG(\gr)$. By~\cref{thm:initial-credit-ub}, we know that if such an initial credit exists, then there exists one such that $\lambda = w_{r+1} \gamma w_{r} \cdots \gamma w_1$ with $r < 2^n$ and $\|\lambda\| \leq n$.  
So, the existential player will try to accumulate this stack content $\lambda$, with one ceveat: She will choose the $\gamma$'s to fit the unique longest word $w_\cap$ accepted by $\cap_{i \in [1,n]} D_i$. Then the stack content $\lambda'$ that the existential player aims to accumulate is $w_{r+1} \gamma^{r} w_r \gamma^{r-1} \cdots \gamma^1 w_1$ such that $\gamma^{r} \cdots \gamma^1$ (called the push-sequence of $\lambda'$) is the $r$-length suffix of $w_{\cap}$\footnote{Actually, here we need the reverse of $w_{\cap}$ as the initial credit will be popped by $\bigcap \texttt{D}_i$, but we simply use $w_{\cap}$ as it is a palindrome.}. She accumulates $\lambda'$ by taking the loop on $p_{n}$ $r$-times by and pushing $\lambda'$ in reverse. The last time she is at $p_0$, it's stack has $w_{r+1}\gamma^{r} w_{r} \gamma^{r-1} \cdots w_2 \gamma^1$ which she multiplies with $w_1$ as she reaches $p_{n}$, then she proceeds to the initial state of the game arena $\gr[\gamma \gets \gamma_{[1,n]}]$: The game arena $\gr$ where each push transition $p \xrightarrow{\gamma} r$ is replaced by $n$-many push transitions $p \xrightarrow{\gamma_i} r$ for each $i \in [1, n]$. Similarly each pop transition $p \xrightarrow{\bar{\gamma}} r$ is replaced by $n$-many pop transitions $p' \xrightarrow{\bar{\gamma}_i} r$ for each $i \in [1, n]$ where $p' = p$ if $p \in Q_\exists$; and otherwise, $p'$ is a new existential state introduced together with the transition $p \xrightarrow{\varepsilon} p'$.
It is not hard to see $\exists$-player wins $\gr[\gamma \gets \gamma_{[1,n]}]$ with the initial credit $\lambda'$ if and only if she wins $\gr$ with the initial credit $\lambda$.

While accumulating the initial credit $\lambda'$, the universal player can challange the existential player after every push $\gamma_i$, at his state $\texttt{ch}$, and claim that the push-sequence of the accumulated initial credit is not in $\bigcap_{i \in [1, n]} L(D_i)$. 
He does so by taking the edge from $\texttt{ch}$ to $\cap_{i \in [1, n]} \texttt{D}_i$: this is a subgame whose initial states are $\mathbf{d}^i$ in $\texttt{D}_i$ for each $i \in [1, n]$, where $\texttt{D}_i$ is the game-version of the DFA $D_i$ given in~\cref{fig:DFA-i} (b). 
When the universal player enters the challange, he proceeds to pick the DFA $D_i$ he claims the initial credit violates, and moves to the initial state $\mathbf{d}^i$ of $\texttt{D}_i$. 
The game-versions $\texttt{D}_i$ are different from the DFAs $D_i$ in two ways.
First, the DFAs push $\gamma_i$ whereas the game-versions pop them, as they are trying to \emph{read} the initial credit. As $w_\cap$ is a palindrome, $\cap_{i \in [1,n]} D_i$ pushes it and $\cap_{i \in [1,n]} \texttt{D}_i$ pops it. 
Secondly, as the initial credit is allowed to have group elements of word length at most $n$ inbetween each push, each state of $\texttt{D}_i$ allows $n$-word length group elements to be consumed inbetween each pop (which is hidden in the figure). We assure this by putting a sequence of $n$-states --a copy of the sequence $p_0\to \cdots \to p_n$-- inside ``states'' $\mathbf{d}^i$ and $\tilde{\mathbf{d}}^i$. So the ``state'' $\mathbf{d}^i$ actually consists of $n$-states $d^i_0, \ldots, d^i_n$ such that there is an edge $d^i_{j-1} \xrightarrow{g} d^i_{j}$ for each $j \in [1,n]$, and for each $g \in W \cup \bar{W} \cup \{\varepsilon\}$. Similarly, the ``state'' $\mathbf{\tilde{d}}^i$ consists of $n$-states $\tilde{d}^i_0, \ldots, \tilde{d}^i_n$. Every edge depicted as entering $\mathbf{d}^i$ and $\mathbf{\tilde{d}}^i$ in the figure, actually enters $d^i_0$ and $\tilde{d}^i_0$; and every edge depicted as leaving $\mathbf{d}^i$ and $\mathbf{\tilde{d}}^i$ actually leaves $d^i_n$ and $\tilde{d}^i_n$, respectively. 

If the accumulated initial created has a push-sequence that violates the language of $D_i$, or contains a group element of word length larger than $n$, then $\texttt{D}_i$ enters the state $\Sadey$, or pops a wrong $\gamma_j$ at $\mathbf{d}^i$ or $\mathbf{\tilde{d}}^i$, which causes the existential player to lose the game. Otherwise, the existential player correctly pops all the $\gamma_j$s, and reaches the first symbol $\gamma_0$ of the stack without reaching $\Sadey$. Then she pops $\gamma_0$, enters $\Smiley$ and wins the game. 

Assume the existential player wins $\UICEG(\gr)$. Let $\lambda$ be an initial credit with $|\lambda|_\gamma < 2^n$ and $\|\lambda\| \leq n$ such that 
 the existential player wins $\gr$ with the initial credit $\lambda$. If the universal player does not challenge the exitential player, she wins the game by accumulating the $\lambda'$ described above, and entering $\gr[\gamma \gets \gamma_{[1, n]}]$. If the universal player challanges, then the play enters a subgame $\texttt{D}_i$ where the stack content is a prefix of $\lambda'$.  The existential player wins the subgame if the stack content is exactly  $\lambda'$, as (i) $|\lambda'|_{\gamma} < 2^n$ and $\|\lambda\|\leq n$, and (ii) the push-sequence of $\lambda'$ is $w_\cap$. She also wins $\texttt{D}_i$ entered with any prefix of $\lambda'$ in the stack, as $\cap_{i \in [1, n]}{L(D_i)}$ is suffix-closed (therefore, $\cap_{i \in [1, n]}{\texttt{D}_i}$ is won by a prefix-closed set of stack content); and (i) naturally holds for prefices of $\lambda'$. So, the existential player wins $\gr^\expp$ with the initial credit $\varepsilon$ if she wins $\UICEG(\gr)$.

Now assume the universal player wins $\UICEG(\gr)$. Then there exists no initial credit $\lambda$ the existential player wins $\gr$, and therefore $\gr^\expp$, with. Then the only way for the existential player to win in $\gr^\expp$ is to avoid entering the game arena $\gr[\gamma \gets \gamma_{[1,n]}]$. For this, she nxeeds to take the loop $p_0 - p_n$ forever. This strategy of $\exists$-player allows the $\forall$-player wait until the accumulated stack content $\lambda'$ achieves the $\Pushdown$-length $2^n$ and challange at his state $\texttt{ch}$. As 
the push-sequence of $\lambda'$ is longer than $2^n$, it is not accepted by 
$\cap_{i \in [1,n]} D_i$. Let $D_i$ be a DFA that rejects the push-sequence of $\lambda'$. Then the universal player wins by entering the subgame $\texttt{D}_i$, as the existential player will not be able to pop the push-sequence of $\lambda'$ correctly in this subgame. 

We conclude that the existential player wins $\UICEG(\gr)$ if and only if she wins $\FICEG(\gr^\expp)$.
\end{proof}

\section{Additional Material for Section~\ref{sec:undecidability-results}}\label{sec:app-undecidability-results}
This section consists of undecidability proofs for the illegal graphs given in~\cref{fig:illegal-graphs}. In the first subsection, we provideundecidability proof for illegal graphs (iii) and (iv), reducing from pushdown energy games. In the next two subsections we provide the undecidability of illegal graphs (ii) and (i), respectively, using a novel gadget construction. As reductions in~\cref{app:undec-ii,app:undec-i} both use counter machines, we introduce counter machines in the beginning of~\cref{app:undec-ii}. 
The results of the last two subsections together prove~\cref{thm:undec-i-ii}. 

\subsection{Undecidability of illegal graphs (iii) and (iv)}\label{app:undec-iii-iv}
We first show the equivalence of pushdown energy games(\cite{AbdullaAHMKT14}) and \nrgs~ on $\Pushdown \times \VASS$ graphs. Then using this equivalence, we will prove~\cref{thm:undec-iii-iv}.

\subparagraph*{Equivalence of pushdown energy games(\cite{AbdullaAHMKT14}) and \nrgs~ on $\Pushdown \times \VASS$ graphs}
An $m$-dimensional pushdown energy game is a tuple $(Q^\texttt{E} = Q^\texttt{E} _\exists \cup Q^\texttt{E} _\forall, \Sigma^\texttt{E} , \delta^\texttt{E}, m)$ where $Q^\texttt{E} $ and $\Sigma^\texttt{E} $ are as in non-termination games on pushdown game graphs, $m$ is the dimension of energy and $\delta^\texttt{E}  \subseteq Q^\texttt{E}  \times \Sigma^\texttt{E} \times Q^\texttt{E}  \times (\Sigma^\texttt{E})^* \times \{-1, 0, 1\}^m $ introduces a game graph over $Q^\texttt{E}  \times (\Sigma^\texttt{E})^* \times \Z^m$ as follows:
If $(p, \gamma, q, w, c) \in \delta^\texttt{E}$, then $(p, \gamma v, E) \to (q, w v, E')$ where $E' = E + c$. 

Here, the universal player wins every play in which the $\Z$-component reaches a negative value. The existential player wins every other play. 
We assume WLOG that every \emph{configuration} $(q, w, E)$ has an outgoing transition.
This assumption can be made, since every state can be enriched with dummy outgoing pop transitions for every pushdown letter $\gamma$ leading to losing or winning trap states that have self loops of $-1$ or $0$ accordingly.

We will show that an $m$-dimensional pushdown energy game is equivalent to a \nrg~ on a graph $\Gamma \in \Pushdown \times \VASS$ graph such that the $\VASS$ subgraph is an $m$-clique, i.e. $\Gamma = (U \cup V, E)$ where there are no edges on the set $U$ and $(V,E)$ is a clique on $m$-nodes. Note that the reduction is not polynomial, but we do not need it to be, since we use it to show that \nrg~ games on $\Pushdown \times \VASS$ inherit the undecidability of pushdown energy games.

First, we reduce a \nrg~ on graph $\Gamma = (U \cup V, E)$ to a $m$-dimensional pushdown energy game with stack alphabet $\Sigma^\texttt{E} = U$.
We use the counters to mimic the increment and decrement of letters in $V = \{v_1, \ldots, v_m\}$ in a straighforward manner: We use counter $i$ increments to mimic $v_i$ transitions, and counter $i$ decrements to mimic $\bar{v}_i$ transitions. What is slightly less obvious is how to mimic the \emph{forbidden pop transitions}, i.e. pop transitions that give rise to non-right-invertible pushdown configurations. 
These transitions are allowed for the universal player in~\nrgs, letting him win the game; but are disallowed in the pushdown energy game. 
We raplace every pop $p \xrightarrow{\bar{\gamma}} q$ in~\nrg~where $p$ is a universal node, with three transitions $p \xrightarrow{\varepsilon} p' \xrightarrow{\bar{\gamma}} q$ where $p'$ is a new existential state. As before, this shifts the burden of the non-right-invertibility to the existential player. Then, we add additional $\varepsilon$-transition from $p'$ to a trap state that is $\forall$-winning and has an $(-1)$-self loop. This makes sure that every configuration has an outgoing transition, and still the universal player wins in both games, if the state $p$ is reached with stack content with top-of-stack other than $\gamma$. We do not need to replace pop transitions on existential states, as the existential player already avoids these states in ~\nrgs.

Finally, we reduce an $m$-dimensional pushdown energy game to a \nrg~ on $\Gamma$. We use the letters of the $\VASS$ alphabet $V$ to mimic the $\Z$ counters in the pushdown energy game in a straightforward manner.
To mimic forbidden transitions, we apply a gadget similar to the one from the reduction of non-termination games on pushdown graphs.  
 Whenever there is a pop transition $p \xrightarrow{\bar{\gamma}} q$ for $\gamma \in \Sigma^\texttt{E}$ defined on a universal state $p$ in the pushdown energy game, we replace it with the transitions $p \xrightarrow{\varepsilon} p' \xrightarrow{\bar{\gamma}} q$ and the transitions $p' \xrightarrow{\bar{\gamma}_i} \Smiley$ for every $\gamma_i \in \Sigma^\texttt{E} \setminus \{\gamma\}$ where $p'$ is a fresh existential state and $\Smiley$ is an $\exists$-winning state with a $0$-self loop. This makes sure that the universal player takes the pop transitions in the reduced \nrg~ only when it is reached with $\gamma$ as top-of-stack-symbol.
 Note that again, we do not need to alter pop transitions originating from $\exists$-player nodes, as $\exists$-player already avoids them in the ~\nrg.

\undeciiiANDiv*

\begin{proof}[Proof of~\cref{thm:undec-iii-iv}]
The undecidability (both $\FICEG$ and $\UICEG$ variants) of (iii), which corresponds to $2\text{-}\Pushdown \times 1\text{-}\VASS$, follows from the work of Abdulla et al~\cite{AbdullaAHMKT14}.  The authors show \enquote{For pushdown energy games, both the fixed and the unknown initial credit problem are undecidable, even if the energy is just single-dimensional}. As explained in~\cref{app:undec-iii-iv}, pushdown energy games with 2 pushdown letters and single-dimensional energy corresponds to $2$-$\Pushdown \times 1$-$\VASS$ in our framework.

The undecidability of (iv) follows from a reduction from (iii) to (iv): Assume that the vertices in graph (iii) and (iv) represent letters $a,b,c$ from left-to-right.
For a given valence game arena $\gr$ over the graph (iii), we will construct an equivalent valence game arena $\gr'$ whose over the graph (iv); both over the alphabet $V \cup \bar{V}$ where $V = \{a, b, c\}$.
Both arenas are on the same finite game graph, have the same initial state, and the transitions of $\gr$ are inherited directly from $\gr'$ except for the ones labelled $c$ or $\bar{c}$. The $c$ transitions of $\gr$ are turned into $cac$ transitions in $\gr'$; and similarly, the $\bar{c}$ transitions of $\gr$ are turned into $\bar{c}\bar{a}\bar{c}$ transitions in $\gr'$. Then the existential player wins $\gr$ if and only if she wins $\gr'$. This is intuitively because, we mimic the lack of right-invertibility of the letter $c$ in (iii) with the lac of right-invertibility of the compound element $cac$ in (iv), while preserving $c$'s commutativity with $b$, and lack of commutativity with $a$.
\end{proof}

\subsection{Undecidability of illegal graphs (ii)}\label{app:undec-ii}

As the results of this subsection, as well as the next, rely on a reduction from counter machines, we first introduce those.

\noindent \subparagraph*{Counter Machines} A counter machine is a finite state machine, equipped with a fixed number of non-negative counters. We consider 2-counter machines, which already capture the Turing-complete behaviour we desire. 

A 2-counter machine $\Minsky = (S, s_0, \delta_{\Minsky})$ consists of a finite set of control states $S$, an initial state $s_0$ and a finite set of transitions $\delta_{M}$ of the form $s \xrightarrow{\mu} s'$ where $s, s' \in S$ and $\mu \in \{\texttt{inc}(i), \texttt{dec}(i), \texttt{zero}(i) \mid i \in \{1, 2\} \,\}$, where $i$ denotes the counter being incremented, decremented, or being tested for zero. A \emph{configuration} of $\Minsky$ is a triple $(s, c_1, c_2) \in Q \times \mathbb{N}^2$. From a configuration $(s, c_1, c_2)$, the transition $w \xrightarrow{\texttt{dec}(i)} w'$ is enabled iff $s = w$ and $c_i \neq 0$; the transition  $w \xrightarrow{\texttt{zero}(i)} w'$ is enabled  iff $s = w$ and $c_i = 0$; and the transition  $w \xrightarrow{\texttt{inc}(i)} w'$ is enabled iff $s = w$. A \emph{run} of $\Minsky$ is a (finite or infinite) sequence of enabled transitions that start at the initial configuration $(s_0, 0, 0)$.  Without loss of generality, we assume $\delta_{\Minsky}$ to have an transition $s \to s'$ from each control state $s \in S$.

2-counter machines are Turing-complete models; therefore all non-trivial properties of them are undecidable, including the existence of an infinite run (or non-termination) problem. In our proofs, we exploit this source of undecidability.

\begin{restatable}{theorem}{undecii}\label{thm:undec-ii}  Let $\Gamma$ be one of the valence graphs represented by (ii). Then $\FICEG(\Gamma)$ and $\UICEG(\Gamma)$ are both undecidable. \end{restatable}

\begin{proof}[Proof of~\cref{thm:undec-ii}]
 To show~\cref{thm:undec-ii}, we will reduce from the non-termination problem in 2-counter machines. %
We reduce a 2-counter machine $\mathcal{M}$ on the state set $V = \{v_0, \ldots, v_k\}$ with the initial state $v_{0}$ to a \nrg ~over game arena $\gr = (\Gamma, Q = Q_\exists \sqcup Q_\forall, \delta, q_{init})$ for $\Gamma = (V, E)$ given by (ii).
The proof works for both graphs: with and without the self-loop on $a$.

The game arena $\gr$ is on the state set $Q = \{q_0, \ldots, q_k\} \cup \{q_{init}\} \cup \{\Smiley\} \subseteq Q_\exists$ together with a set $Q' \subseteq Q_\forall \cup Q_\exists$ which comes from the gadgets given in~\cref{fig:gadgets-ii}, and will be detailed in the proof. There is a transition from $q_{init}$ to $q_0$ with the compound element $ba$, as well as an $\varepsilon$ self-loop on $\Smiley$. Clearly, every play that reaches $(\Smiley, [u])$ with a right-invertible $u$ is won by the existential player. For every transition $v_i \xrightarrow{\mu} v_j$ in $\mathcal{M}$, where $\mu \in \{\texttt{inc}(i), \texttt{dec}(i) \mid i \in \{1, 2\} \,\}$, there is a transition $q_i \xrightarrow{\mu'} q_j$ in $\gr$ where $\mu'$ is equal to $b$, $c$, $\bar{b}$ and $\bar{c}$ for $\mu$ equal to $\texttt{inc}(1)$, $\texttt{inc}(2)$, $\texttt{dec}(1)$ and $\texttt{dec}(2)$, respectively. After mimicing the increments of decrements of $\mathcal{M}$ easily using $b$ and $c$, we now need some machinery to mimic the zero tests. For the transition $v_i \xrightarrow{\texttt{zero}(k)} v_j$ in $\mathcal{M}$, we will add the gadget in~\cref{fig:gadgets-ii} between $q_i$ and $q_j$. %
The gadget introduces a universal player state $q^{\forall}_{i, j}$, whose role is to either accept the claim of the zero test by allowing the existential player to move to $q_j$; or to challange the claim by sending her to the state $q^{\exists}_{i, j}$. The existential player wins from $q^{\exists}_{i, j}$ if and only if the configuration is free of the letter mimicing the counter value checked for zero -- i.e. the play reaches $(q^{\exists}_{i, j}, [w])$ where $ w = bac^m$ for $\texttt{zero}(1)$-gadget and $w = bab^m$ for $\texttt{zero}(2)$-gadget, for some $m \in \N$. In $\texttt{zero}(1)$-gadget, the self-loop on $q^{\exists}_{i, j}$ allows the existential player to decrement as many $c$'s as desired; and in $\texttt{zero}(2)$-gadget, as many $b$'s as desired.

Notice that the existential player cannot take this self-loop forever, as the number of $c$s and $b$s cannot go negative (since $\bar{c}$ and $\bar{b}$ are non-right-invertible). Then eventually, she needs to take the $\bar{a}\bar{b}$ transition to $\Smiley$. However, this transition can be taken only if the configuration does not contain any $b$ or $c$s after the $ba$, as $a$ does not commute with neither $b$ nor $c$, meaning, $b\bar{a}\bar{b}$ and $c\bar{a}\bar{b}$ are both non-right-invertible\footnote{Note that $b\bar{a}$ and $c\bar{a}$ are already non-right-invertible if the first vertex $a$ in (ii) does not have a self-loop. Therefore, the initial transition of $a$ (instead of $ba$) and the popping transition of $\bar{a}$ between $q^\exists_{i, j}$ and $\Smiley$ (instead of $\bar{a}\bar{b}$) is sufficient in this case. However, when this state has a self-loop, we need the compound element $ba$ to have the same blocking effect.}. So, this transition can only be taken if $w\bar{a}\bar{b}$ is right invertible, which requires $w$ to have the suffix $ba$. In $\texttt{zero}(1)$ case, if the first counter was zero at $q_{i}$ in $\mathcal{M}$ as claimed, then in the $\texttt{zero}(1)$-gadget the existential player arrives at the configuration $(q^{\exists}_{i, j}, [w])$ with $w = ba\cdot c^m$, takes the self-loop $m$-many times, and takes the transition to $\Smiley$ by removing $ba$, and wins the game. On the other hand, if the first counter was non-zero at $q_{i}$ in $\mathcal{M}$, then she reaches $(q^{\exists}_{i, j}, [w'])$ with $w'= ba\cdot b^k c^m$ and can never achieve a configuration $(q^{\exists}_{i, j}, [w])$ where $w$ has a suffix $ba$ to allow the transition to $\Smiley$; therefore losing the game. 

There exists an infinite path in $\mathcal{M}$ if and only if $\exists$-player wins $\FICEG(\gr)$: Let $\pi$ be an infinite path in $\mathcal{M}$. Then as $b$ and $c$ commute in (i), all increment and decrements of counter $1$ and counter $2$ in $\pi$ are easily mimiced by the existential player in $\gr$ by the push and pops of $b$ and $c$. The zero tests $v_i \xrightarrow{\texttt{zero}(k)} v_j$ for $k \in \{1, 2\}$ in $\pi$ branch into two plays in $\gr$: one that is challanged by the universal player and immediately reaches $\Smiley$, and another that is unchallenged that mimic $\pi$ through the gadget states $q_i \xrightarrow{\varepsilon} q^\forall_{i, j} \xrightarrow{\varepsilon} q_j$. If there exist no infinite paths in $\mathcal{M}$, then\footnote{As we assume no dead-ends.} every path in $\mathcal{M}$ ends with an invalid zero-test. Look at such a path $\pi = p^1 \ldots p^m$. Then $p^m$ has some outgoing edges, each with an invalid zero-test, that is $p^m \xrightarrow{\texttt{zero}(k)} p$ for some $p \in V$. Let $\Pi$ be the set of all such paths in $\mathcal{M}$. 
The universal player has the following winning strategy in $\FICEG(\gr)$: On every history that is a strict prefix $\pi'$ of a play $\pi \in \Pi$, take the non-challenging decision on the gadget state $q^\forall_{i, j}$, and mimic the path $\pi'$. Whenever the history of a play becomes equivalent to $\pi$ for some $\pi \in \Pi$ (minus the gadget states), take the challange decision on the gadget state $q^\forall_{i, j}$, forcing the play to reach $q^\exists_{i, j}$. As the number of $b$ and $c$s mimic the values of counters $1$ and $2$, respectively, the state $q^\exists_{i, j}$ is reached with a wrong number of $b$s or $c$s (depending on the zero-test) after the initial $ba$, and therefore the existential player loses the game.

Observe that the fixed initial empty credit plays no role in the proof. 
That is, existential player wins $\FICEG(\gr)$ if and only if she wins $\UICEG(\gr)$. The \enquote{if} direction is trivial. To see the \enquote{only if} direction, assume the existential player loses from $q_{init}$ with an empty initial credit, but wins with some initial credit $\lambda$. %
Let $\sigma$ be an existential strategy in $\gr$ that wins $(q_{init}, [\lambda])$. We will show that $\sigma$ also wins $(q_{init}, [\varepsilon])$. The only places in the game that might cause a $\sigma$-play $\pi$ starting from $(q_{init}, [\lambda])$ to be winning, and a $\sigma$-play $\pi'$ starting from $(q_{init}, [\varepsilon])$ to be losing are the branches of zero-test gadgets where the universal player challenges the claim, and sends the play to $q_{i, j}^\exists$. The play $\pi'$ reaches $(q_{i, j}^\exists, [w])$ with $ w = ba b^r c^m$, if and only if $\pi$ reaches $(q_{i, j}^\exists, [w'])$ with $w' = \lambda ba b^r c^m$. WLOG, assume the gadget is a $\texttt{zero}(1)$ gadget. Then as $\pi$ is winning, $r = 0$ and $\pi$ reaches $\Smiley$ by taking the self-loop $m$-times, and proceeding to $\Smiley$ by removing $ba$. Then, $\pi'$ reaches $\Smiley$ with the same strategy. Therefore, the reduction works to show the undecidability of $\UICEG(\gr)$ too. 
\end{proof}

\subsection{Undecidability of illegal graphs (i)}\label{app:undec-i}

\begin{restatable}{theorem}{undeci}\label{thm:undec-i} Let $\Gamma$ be one of the valence graph represented by (i). Then $\FICEG(\Gamma)$ and $\UICEG(\Gamma)$ are both undecidable. \end{restatable}

\begin{proof}[Proof of \cref{thm:undec-i}]

    We reduce a 2-counter machine $\mathcal{M}$ on the state set $V = \{v_0, \ldots, v_k\}$ with the initial state $v_{0}$ to a \nrg over game arena $\gr = (\Gamma, Q = Q_\exists \sqcup Q_\forall, \delta, q_{init})$ for $\Gamma = (V, E)$ given by (i).
As before, the reduced game $\gr$ is on the state set $Q = \{q_0, \ldots, q_k\} \cup \{q_{init}\} \cup \{\Smiley\} \subseteq Q_\exists$ together with a set $Q' \subseteq Q_\forall \cup Q_\exists$ which comes from the gadgets given in~\cref{fig:gadgets-i}. Once more, we have the initial transition $q_{init} \xrightarrow{ba} q_0$ in $\gr$, and the $\Smiley$ has an $\varepsilon$-self-loop, ensuring every play reaching $(\Smiley, [u])$ with a right-invertible $u$ is won by the existential player.

The gadgets differ from~\cref{fig:gadgets-ii} in two main aspects. Firstly, it is not sufficient to use one letter to mimic one counter in the 2-counter machine, because this time $c$ is not a blocking letter -- i.e. a word is right-invertible no matter how many $\bar{c}$s it has. As a result, in the $\texttt{zero}(1)$ gadget, the $\bar{c}$ self-loop in~\cref{fig:gadgets-ii} can be taken by the existential player forever to secure an infinite right-invertible play. Therefore, as we will continute to use $b$ to mimic the first counter as in the previous proof, this time we will use $bc$ to mimic the second counter, instead of $c$. For every transition $v_i \xrightarrow{\texttt{inc}(i)} v_j$ in $\mathcal{M}$, $\gr$ will have a transition $q_i \xrightarrow{\mu} q_j$ where $\mu = b$ if $i = 1$, and $\mu = bc$ if $i = 2$. The second difference is, this time we need to check non-emptiness, as well as the emptiness of a counter. This is because, $\bar{b}$ and $\bar{c}$ (and consequently $\bar{b}\bar{c}$) are not blocking as they were in the zero-test gadgets depicted in~\cref{thm:undec-ii}. We explained why $\bar{c}$ is not blocking, the reason $\bar{b}$ is not blocking, is due to how we mimic the counters. As the second counter is mimiced using both letters, a configuration $(0, 2)$ in $\mathcal{M}$ would correspond to a stack content $[babbcc]$ in $\gr$, which would allow $b$ to be popped in the reduced game, even though the first counter is empty in $\mathcal{M}$. We can see that mimicing the decrements using $\bar{b}$ and $\bar{b}\bar{c}$ transitions as in the proof of~\cref{thm:undec-ii} this time gives us incorrect behaviour.
Therefore, we need gadgets (\cref{fig:gadgets-i}) to mimic $\texttt{dec}(i)$ as well as $\texttt{zero}(i)$.

Intuitively, for every decrement (resp., zero-test) on $q_i$, the existential player moves from $q_i$ to $q_{i, j}$, by implementing the correct decrement via $\bar{b}$ for $\texttt{dec}(1)$ and $\bar{b}\bar{c}$ for $\texttt{dec}(2)$.
At $q_{i, j}$, the universal player either accepts that the decrement is legitimate (that, we decremented counter is non-empty) by moving to $q_j$, or challanges its legitimacy by moving to $q^{\exists}_{i,j}$. In both $\texttt{dec}(i)$ gadgets, existential player can decrement as many 
$b$'s and $bc$'s via the self-loops on $q^\exists_{i,j}$ as desired. Her aim is again
to reach $(q^\exists_{i,j}, [w])$ with some $w$ that has $ba$ as a suffix.%
It is not hard to see that she can obtain such a $[w]$ only if the decremented counter was non-zero at $q_i$; that is, the decrement was legitimate. Moreover, she cannot proceed to $\Smiley$ with any other stack content --i.e. $[bac^mb^{m+n}]$ where $m$ or $n$ is non-zero--  as $[b\bar{a}\bar{b}]$ and $[c\bar{a}\bar{b}]$ are both non-right-invertible. Here, we are using the number of $c$s and the fact that $c$ and $a$ do not commute to ensure that the number of popped $bc$s will be equal to the second counter's value, and the number of popped (remaining) $b$s will be equal to the first counter's value. Finally, that existential player cannot take the self-loops on $q^\exists_{i,j}$ forever, as the stack is not allowed to have a negative number of $b$s.

The zero-test gadgets are very similar to those in~\cref{fig:gadgets-ii}, except that, of course, the pop corresponding to the second counter is represented by $\bar{b}\bar{c}$ instead of $\bar{c}$. The correctness proof follows analogously.

As in the previous proof, the initial stack content does not play a role in the proof. Consequently, the reduction works to show the undecidability of both $\FICEG(\gr)$ and $\UICEG(\gr)$.  
\end{proof}

}

\label{afterbibliography}
\newoutputstream{pagestotal}
\openoutputfile{main.pagestotal.ctr}{pagestotal}
\addtostream{pagestotal}{\getpagerefnumber{afterbibliography}}
\closeoutputstream{pagestotal}

\newoutputstream{todos}
\openoutputfile{main.todos.ctr}{todos}
\addtostream{todos}{\arabic{@todonotes@numberoftodonotes}}
\closeoutputstream{todos}
\end{document}